\documentclass[runningheads,citeauthoryear]{llncs}

\usepackage{marvosym} %
\usepackage[libertine]{newtxmath}

\usepackage{booktabs} %
\usepackage[
	pdfencoding=auto,
	psdextra,
	colorlinks=true,
	linkcolor=blue!50!black,
	citecolor=blue!50!black,
	urlcolor=blue!50!black,
	bookmarks=true,
	bookmarksopen=true,
	bookmarksnumbered=true,
	pdfstartview=FitH,
	bookmarksdepth=2 %
]{hyperref}

\providecommand{\Description}[1]{}

\usepackage{amsmath}
\usepackage{mathtools}
\usepackage{xspace}
\usepackage{cleveref}
\usepackage{pifont}
\usepackage{tcolorbox}
\usepackage[utf8]{inputenc}
\usepackage[T1]{fontenc}
\usepackage[dvipsnames,x11names]{xcolor}
\usepackage[textsize=tiny,bordercolor=white]{todonotes}
\usepackage{stmaryrd}
\usepackage{float}
\usepackage{xfrac}
\usepackage[inline]{enumitem}
\usepackage{makecell}
\usepackage{tablefootnote}
\usepackage{tikz}
\usetikzlibrary{patterns,calc,arrows,shapes,snakes,automata,backgrounds,petri,positioning,decorations.pathmorphing,intersections}
\usepackage{thmtools}
\usepackage{adjustbox}
\usepackage{wrapfig}
\usepackage{caption}
\usepackage{orcidlink}

\usepackage[
  backend=biber,
  style=lncs
]{biblatex}
\usepackage{etoolbox}
\newtoggle{arxiv}
\toggletrue{arxiv} %

\definecolor{webgreen}{rgb}{0,.5,0}
\newcommand{\gray}[1]{\textcolor{gray}{#1}}

\newcommand{\black}[1]{\textcolor{black}{#1}}

\newcommand{\red}[1]{\textcolor{Firebrick2}{#1}}
\newcommand{\blue}[1]{\textcolor{DodgerBlue3!75!DodgerBlue4}{#1}}
\newcommand{\orange}[1]{\textcolor{orange}{#1}}
\newcommand{\purple}[1]{\textcolor{purple!70!red}{#1}}

\newcommand{\annocolor}[1]{\blue{#1}}
\newcommand{\annotate}[1]{{\annocolor{\!\fatslash\!\!\!\fatslash~~\vphantom{G'} {#1}}}}

\newcommand{\morespace}[1]{~{}#1{}~}
\newcommand{\qmorespace}[1]{\quad{}#1{}\quad}

\newcommand{\eeq}{\morespace{=}}

\newcommand{\pplus}{~{}+{}~}
\newcommand{\mminus}{\morespace{-}}
\newcommand{\ccirc}{\morespace{\circ}}

\newcommand{\tto}{~{}\to{}~}

\newcommand{\qeq}{\quad{}={}\quad}
\newcommand{\eeqq}{\qeq}

\newcommand{\lleq}{~{}\leq{}~}

\newcommand{\ie}{i.e.\xspace}
\newcommand{\eg}{e.g.\xspace}

\newcommand{\mcStates}[1][]{S_{#1}}
\newcommand{\mcTransitionProbs}[1][]{P_{#1}}

\newcommand{\mcPaths}[1][]{\mathsf{Paths}_{#1}}

\newcommand{\mcInducedMeasure}[2][]{\mathsf{Pr}_{#1}^{#2}}

\newcommand{\occupationMeasure}{\mathrm{OM}}
\newcommand{\mcOccupationMeasure}[2][]{\occupationMeasure_{#1}^{#2}}
\newcommand{\mcOccurences}[2][]{\mathsf{occ}_{#1}^{#2}}

\newcommand{\pwfont}[1]{\mathbf{#1}}
\newcommand{\pwcolor}{Turquoise!40!Black}
\newcommand{\pw}{\textcolor{\pwcolor}{\pwfont{pm}}}
\newcommand{\charfun}{\Phi_{g,C}}
\newcommand{\pvars}{\mathsf{Vars}}

\newcommand{\wpfont}[1]{\mathbf{#1}}
\newcommand{\wpcolor}{Red!40!Black}
\newcommand{\wpre}{\textcolor{\wpcolor}{\wpfont{wp}}}

\newcommand{\ertfont}[1]{\mathbf{#1}}
\newcommand{\ertcolor}{Green!40!Black}
\newcommand{\ert}{\textcolor{\ertcolor}{\ertfont{ert}}}

\newcommand{\pstates}{\ensuremath{\mathsf{States}}}
\newcommand{\pstatea}{\ensuremath{s}}

\newcommand{\loopMC}{\mathcal{L}}
\newcommand{\loopMCsem}[1]{\loopMC \sem{#1}}

\newcommand{\pkeywordcolor}{Blue}
\newcommand{\pcommentcolor}{Gray!80!Emerald}
\newcommand{\pkeywordfont}[1]{{\normalfont\textbf{\texttt{\textcolor{\pkeywordcolor}{#1}}}}}
\newcommand{\pcommentfont}[1]{\textit{\textcolor{\pcommentcolor}{#1}}}
\newcommand{\pgcl}{\textup{pGCL}\xspace}
\newcommand{\redip}{\textup{ReDiP}\xspace}
\newcommand{\closedRedip}{\textup{clReDiP}\xspace}

\newcommand{\sem}[1]{\llbracket {#1} \rrbracket}
\newcommand{\semapp}[2]{\sem{#1}\left({#2}\right)}
\newcommand{\pwapp}[2]{\pw \semapp{#1}{#2}}

\newcommand{\kozensem}{\mathcal{K}\,}
\newcommand{\kozensemapp}[2]{\kozensem \semapp{#1}{#2}}

\newcommand{\wpreapp}[2]{\wpre \semapp{#1}{#2}}

\newcommand{\ertapp}[2]{\ert \semapp{#1}{#2}}
\newcommand{\ertappnew}[1]{\ert\sem{#1}}

\newcommand{\pchoice}[3]{\{{#1}\}\,[{#2}]\,\{{#3}\}}
\newcommand{\pskip}{\pkeywordfont{skip}}
\newcommand{\ptrue}{\pkeywordfont{true}}

\newcommand{\pdiverge}{\pkeywordfont{diverge}}
\newcommand{\passign}[2]{{#1} := {#2}}
\newcommand{\psample}[2]{{#1} :\approx {#2}}
\newcommand{\pcomp}[2]{{#1}\fatsemi{#2}}
\newcommand{\pif}{\pkeywordfont{if}}
\newcommand{\pelse}{\pkeywordfont{else}}
\newcommand{\pwhile}{\pkeywordfont{while}}
\newcommand{\pite}[3]{\pif\,({#1})\,\{{#2}\}\,\pelse\,\{{#3}\}}
\newcommand{\ploop}[2]{\pwhile\,({#1})\,\{{#2}\}}
\newcommand{\pcomment}[1]{\textcolor{\pcommentcolor}{{/}^*~}\pcommentfont{#1}\textcolor{\pcommentcolor}{~{}^{*}\!{/}}}

\newcommand{\codify}[1]{\textnormal{\texttt{#1}}}
\newcommand{\progvar}[1]{#1}

\newcommand{\AssignSymbol}{\mathrel{\textnormal{$\mathtt{\coloneqq}$}}}
\newcommand{\ASSIGN}[2]{\ensuremath{#1 \AssignSymbol #2}}
\newcommand{\assign}[2]{\ASSIGN{#1}{#2}}

\newcommand{\decr}[1]{\progvar{{#1}\!-\!-}\xspace}

\newcommand{\incrasgn}[2]{
\ensuremath{
\progvar{#1}~+\!= {#2}
}
}

\newcommand{\iid}[2]{
	\ensuremath{
		\codify{iid}\left(#1,#2\right)
	}
}

\newcommand{\COMPOSE}[2]{\ensuremath{{#1}{\,\fatsemi}~ {#2}}}

\newcommand{\IFSYMBOL}{\ensuremath{\textnormal{\texttt{if}}}}

\newcommand{\ELSESYMBOL}{\ensuremath{\textnormal{\texttt{else}}}}

\newcommand{\ITE}[3]{\ensuremath{\IFSYMBOL\,\left(\, {#1} \,\right)\,\left\{\, {#2} \,\right\}\,\ELSESYMBOL\,\left\{\, {#3} \,\right\}}}

\newcommand{\WHILESYMBOL}{\ensuremath{\textnormal{\texttt{while}}}}

\newcommand{\WHILEDO}[2]{\ensuremath{
		\WHILESYMBOL \left(\, {#1} \,\right)\left\{\, {#2} \,\right\}
	}
}

\DeclareMathOperator\lfp{lfp}

\newcommand{\diff}{\mathop{}\!\mathrm{d}}
\newcommand{\nats}{\mathbb{N}}
\newcommand{\reals}{\mathbb{R}}
\newcommand{\nnreals}{\reals_{\geq 0}}
\newcommand{\nnextreals}{\nnreals^{\infty}}
\renewcommand{\vec}{\mathbf}
\newcommand{\abs}[1]{\left\vert{#1}\right\vert}
\newcommand{\powerset}[1]{\mathcal{P}\!\left({#1}\right)}

\newcommand{\accumulate}[3]{\sum_{#1}{ #2 \cdot #3}}
\newcommand{\kernelEvalAt}[2]{{#1}[#2]}
\newcommand{\cylinderSet}[1]{\mathrm{Cyl}(#1)}

\newcommand{\measuresOf}[1]{\mathcal{M}\left({#1}\right)}
\newcommand{\distOf}[1]{\mathcal{D}\left({#1}\right)}
\newcommand{\probDistOf}[1]{\mathcal{D}_{=1}\left({#1}\right)}
\newcommand{\expectedValueOf}[2][]{\mathbb{E}_{#1}\left[ #2 \right]}

\newcommand{\distfont}[1]{\mathrm{#1}}
\newcommand{\flip}[1]{\distfont{flip}(p)}
\newcommand{\dirac}[1]{\delta_{#1}}

\newcommand{\geometric}[1]{\distfont{geom}(#1)}

\newcommand{\indicatorFct}[1]{\mathbf{1}_{#1}}
\newcommand{\iverson}[1]{\left[#1\right]}

\newcommand{\gfVar}[1]{\MakeUppercase{#1}}
\newcommand{\gfCoefficient}[2]{{#1}({#2})}

\newcommand{\measureGF}[1]{\mathrm{GF}_{#1}}
\newcommand{\measureGFa}{\mathfrak{g}}
\newcommand{\measureGFb}{\mathfrak{h}}
\newcommand{\measureGFc}{\mathfrak{i}}

\newcommand{\gfRationalNumerator}[1]{\mathfrak{n}_{#1}}
\newcommand{\gfRationalDenominator}[1]{\mathfrak{d}_{#1}}

\newcommand{\gfFilter}[2]{#1 {\upharpoonright_{#2}}}
\newcommand{\gfSubstitute}[3]{#1 [{#2}/{#3}]}

\newcommand{\redipGfCharfun}[2]{\Phi_{#1, #2}}

\newcommand{\invariantTemplateVars}{\mathcal{T}}

\newcommand{\toolfont}[1]{\textsc{#1}}
\newcommand{\prodigy}{\toolfont{ProDiGy}\xspace}
\newcommand{\ginac}{\toolfont{GiNaC}\xspace}
\newcommand{\sympy}{\toolfont{SymPy}\xspace}

\newcommand{\qedhere}{\hfill\ensuremath{\square}}

\floatstyle{ruled} %
\newfloat{Program}{tbp}{lop}%
\Crefname{Program}{Program}{Programs}
\crefname{Program}{program}{programs}
\newenvironment{smashedalign}{\par$\!\aligned}{\endaligned$\par} %

\begin{document}

\title{Generating Functions Meet Occupation Measures: Invariant Synthesis for Probabilistic Loops\iftoggle{arxiv}{ (Extended Version)}{}}
\titlerunning{Generating Functions Meet Occupation Measures}

\author{Darion Haase$^{(\text{\Letter})}$\inst{1}\orcidlink{0000-0001-5664-6773} \and
  Kevin Batz\inst{2,3}\orcidlink{0000-0001-8705-2564} \and
  Adrian Gallus\inst{4}\orcidlink{0000-0002-2176-5075} \and\\
  Benjamin Lucien Kaminski\inst{5,2}\orcidlink{0000-0001-5185-2324} \and
  Joost-Pieter Katoen\inst{1}\orcidlink{0000-0002-6143-1926} \and\\
  Lutz Klinkenberg\inst{1}\orcidlink{0000-0002-3812-0572} \and
  Tobias Winkler\inst{1}\orcidlink{0000-0003-1084-6408}
}%
\authorrunning{D. Haase et al.}
\institute{
  RWTH Aachen University, Aachen, Germany \\
  \email{\{darion.haase, katoen, lutz.klinkenberg, tobias.winkler\}@cs.rwth-aachen.de} \and
  University College London, London, United Kingdom \and
  Cornell University, Ithaca, NY, USA \\
  \email{k.batz@ucl.ac.uk}\and
  Independent Researcher, Aachen, Germany \\
  \email{adrian.gallus@rwth-aachen.de} \and
  Saarland University, Saarbrücken, Germany \\
  \email{kaminski@cs.uni-saarland.de}
}

\maketitle

\begin{abstract}
    A fundamental computational task in probabilistic programming is to infer a program's output (posterior) distribution from a given initial (prior) distribution.
    This problem is challenging, especially for expressive languages that feature loops or unbounded recursion.
    While most of the existing literature focuses on statistical approximation, in this paper we address the problem of mathematically exact inference.
    To achieve this for programs with loops, we rely on a relatively underexplored type of probabilistic loop invariant, which is linked to a loop's so-called \emph{occupation measure}.
    The occupation measure associates program states with their expected number of visits, given the initial distribution.
    Based on this, we derive the notion of an \emph{occupation invariant}.
    Such invariants are essentially dual to probabilistic martingales, the predominant technique for formal probabilistic loop analysis in the literature.
    A key feature of occupation invariants is that they can take the initial distribution into account and often yield a proof of positive almost sure termination as a by-product.
    Finally, we present an automatic, template-based invariant synthesis approach for occupation invariants by encoding them as \emph{generating functions}.
    The approach is implemented and evaluated on a set of benchmarks.
\end{abstract}

\newcommand{\omIntro}{\mathit{OM}}

\section{Introduction}
\label{sec:intro}

\emph{Probabilistic programs}
(PPs) are like ordinary programs with the added ability to flip coins or, more generally, sample values from probability distributions.
PPs are ubiquitous in modern computing; they appear, for example, in randomized algorithms~\cite{DBLP:journals/cj/Hoare62}, random sampling~\cite{DBLP:journals/corr/abs-1304-1916}, statistical inference routines~\cite{carpenter2017stan,dippl,DBLP:journals/jmlr/BinghamCJOPKSSH19}, cognitive science~\cite{cogscience}, and autonomous systems~\cite{autsystems}.

\paragraph{Probabilistic Programs as Measure Transformers.}
Intuitively, running a PP on a given input produces a \emph{probability distribution}\footnote{Or \emph{sub-}distribution, if the program does not terminate with probability $1$.} over possible outputs.
Nearly half a century ago, \citeauthor{DBLP:journals/jcss/Kozen81}~\cite{DBLP:journals/jcss/Kozen81} has cast this intuition into a formal (denotational) program semantics,
which associates each program $C$ with a function\footnote{Notice that $\kozensem\sem{\cdot}$ operates on general \emph{measures}, including (sub-)probability distributions as a special case.}%
\[
    \gray{
    \overbrace{
        \black{
            \kozensem\sem{C}
        }
    }^{
    \mathclap{\qquad\text{$\mathcal{K}$ozen's measure transformer semantics of program $C$\qquad\qquad}}
    }
    }
    \colon
    \quad
    \gray{
        \underbrace{
            \black{
                \measuresOf{\pstates}
            }
        }_{
            \mathclap{\quad\qquad\qquad\text{measures over input states}}
        }
    }
    \tto
    \gray{
        \overbrace{
            \black{
                \measuresOf{\pstates}
            }
        }^{
            \mathclap{\qquad\qquad\qquad\text{measures over output states}}
        }
    }
    ~.
\]%
In this paper we restrict to \emph{discrete} measures over a countable state space $\pstates$.
For example, the program in \Cref{fig:intro} transforms the initial point-mass distribution on state $[x \mapsto 1,\, c\mapsto 0]$ to a distribution where $c$ is geometrically distributed with parameter $\sfrac 1 2$ (and $x=0$ with probability $1$).
Notice that \Cref{fig:intro} uses \emph{generating function} (GF) notation for distributions.
GFs play an important role for our development, which we explain later on page~\pageref{par:gfIntro}.

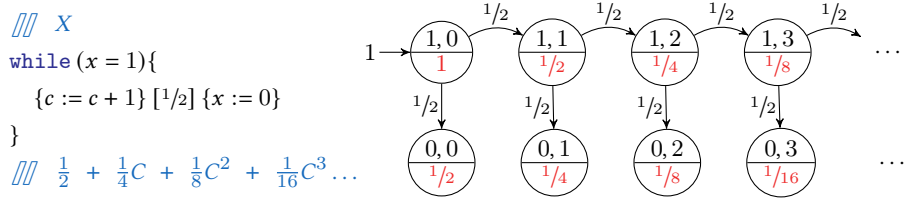
\begin{figure}[t]
    \centering
    \begin{minipage}{0.375\textwidth}
        \begin{smashedalign}
            & \annotate{\gfVar{x}} \\
            &\pwhile\,(x = 1)\{\\
            &\quad \pchoice{\passign{c}{c+1}}{\sfrac 1 2}{\passign{x}{0}}\\
            &\} \\
            & \annotate{\tfrac{1}{2}  \pplus \tfrac{1}{4}  \gfVar{c} \pplus \tfrac{1}{8}  \gfVar{c}^2 \pplus \tfrac{1}{16}  \gfVar{c}^3 \ldots}
        \end{smashedalign}
    \end{minipage}
    \begin{minipage}{0.615\textwidth}
        \begin{tikzpicture}[every initial by arrow/.style={>=stealth'},>=stealth',inner sep=1pt,node distance=6mm and 6mm]
            \node (s10) [initial,initial text=1,state with output] {$1,0$ \nodepart{lower} $\red{1}$};
            \node (s11) [state with output, right=of s10] {$1,1$ \nodepart{lower} $\red{\sfrac{1}{2}}$};
            \node (s12) [state with output, right=of s11] {$1,2$ \nodepart{lower} $\red{\sfrac{1}{4}}$};
            \node (s13) [state with output, right=of s12] {$1,3$ \nodepart{lower} $\red{\sfrac{1}{8}}$};
            \node (s14) [state, draw=none, right=of s13] {$\ldots$};

            \node (s00) [state with output, below=of s10] {$0,0$ \nodepart{lower} $\red{\sfrac{1}{2}}$};
            \node (s01) [state with output, right=of s00] {$0,1$ \nodepart{lower} $\red{\sfrac{1}{4}}$};
            \node (s02) [state with output, right=of s01] {$0,2$ \nodepart{lower} $\red{\sfrac{1}{8}}$};
            \node (s03) [state with output, right=of s02] {$0,3$ \nodepart{lower} $\red{\sfrac{1}{16}}$};
            \node (s04) [state, draw=none, right=of s03] {$\ldots$};

            \draw[->] (s10) edge[bend left] node[above] {$\sfrac{1}{2}$}  (s11);
            \draw[->] (s11) edge[bend left] node[above] {$\sfrac{1}{2}$}  (s12);
            \draw[->] (s12) edge[bend left] node[above] {$\sfrac{1}{2}$}  (s13);
            \draw[->] (s13) edge[bend left] node[above] {$\sfrac{1}{2}$}  (s14);

            \draw[->] (s10) edge node[left] {$\sfrac{1}{2}$}  (s00);
            \draw[->] (s11) edge node[left] {$\sfrac{1}{2}$}  (s01);
            \draw[->] (s12) edge node[left] {$\sfrac{1}{2}$}  (s02);
            \draw[->] (s13) edge node[left] {$\sfrac{1}{2}$}  (s03);

        \end{tikzpicture}
    \end{minipage}
    \caption{
        \textbf{Left:}
        A probabilistic loop, \annocolor{annotated} with initial and final distributions.
        The annotations use \emph{probability generating functions} (PGFs), where indeterminates represent program variables (e.g., $X$ for $x$), exponents represent values (variables with value 0 are omitted since $\gfVar{x}^{0} = 1$; exponent 1 is omitted since $\gfVar{x}^{1} = X$), and coefficients represent probabilities (coefficient 1 is also omitted)~\cite{DBLP:conf/lopstr/KlinkenbergBKKM20}.
        For example, the PGF $\gfVar{x}$ represents the Dirac distribution  where, with probability $1$, $x$ has value $1$ and $c$ has value $0$.%
        \\
        \textbf{Right:} Markov chain representation of the program (transitions mimic entire loop iterations).
        The upper part of each state is the variable valuation $x,c$.
        The \red{red numbers} constitute the \emph{occupation measure} (expected number of visits).}
    \label{fig:intro}
\end{figure}

\paragraph{Our Goal: Inference for Loops --- Exact and Automatic.}
Characterizing a program's exact output distribution is desirable for many applications --- yet extremely challenging in general, especially in the presence of loops.
The goal of this paper is to push the limits of \emph{exact} and \emph{automatic} loop analysis:
\medskip%
\begin{quote}
    \textbf{Problem statement:}~
    Given a discrete probabilistic loop $\ploop{\phi}{C}$ with a distribution $\mu$ over input states, \emph{automatically} compute an \emph{exact} representation of the output distribution $\kozensem\sem{\ploop{\phi}{C}}(\mu)$.
\end{quote}%
\medskip%
Two remarks on our problem statement are in order:
(1)~Due to standard undecidability results for while loops, no complete algorithmic solution exists. We thus provide \emph{heuristics} covering many instances.
(2)~Since the final distribution may have \emph{infinite support} (as in \Cref{fig:intro}), we employ \emph{closed-form} generating functions as a succinct encoding of probability distributions~\cite{DBLP:conf/lopstr/KlinkenbergBKKM20} (details follow).

\paragraph{Automatic Reasoning with Kozen's Definition --- Infeasible?}
\citeauthor{DBLP:journals/jcss/Kozen81}~\cite{DBLP:journals/jcss/Kozen81} characterizes the semantics of loops by the functional least fixed point equation\footnote{The equation may have multiple solutions; as usual, the \emph{least} one w.r.t.\ a suitable order is of interest (see \Cref{ssec:semantics:equivalence_to_kozen}).}
\[
    \kozensem\sem{\ploop{\phi}{C}}
    \qeq
    \kozensem\sem{\ploop{\phi}{C}} \ccirc \kozensem\sem{C} \ccirc \sem{\phi} \qmorespace{+} \sem{\neg\phi}
    \tag{$\dagger$}
    \label{eq:kozenFixedPoint}
    ~,
\]
where $\sem{\phi}$ and $\sem{\neg\phi}$ \emph{filter} the incoming distribution according to $\phi$ and $\neg\phi$, respectively.
This equation reflects the intuition that a loop is unaffected by unrolling it once.
However, there is a discrepancy between \eqref{eq:kozenFixedPoint} and our problem statement:
Our goal is to determine $\kozensem\sem{\ploop{\phi}{C}}(\mu)$ for a \emph{given initial distribution} $\mu$ --- yet \eqref{eq:kozenFixedPoint} does not depend on such a $\mu$.
In other words, an approach based on solving \eqref{eq:kozenFixedPoint} \emph{would require to reason about all possible initial distributions simultaneously} (see~\cite{DBLP:conf/lopstr/KlinkenbergBKKM20,DBLP:conf/cav/ChenKKW22,DBLP:journals/pacmpl/KlinkenbergBCHK24,KlinkenbergPHD}), whereas we are interested in the solution of the functional equation only at one point, namely at the initial distribution $\mu$.

\paragraph{Occupation Measures to the Rescue.}
An equivalent but seemingly less well-known\footnote{Beyond \cite{DBLP:journals/siamcomp/SharirPH84}, we are not aware of any other work on occupation measures of PPs.} characterization of a loop's output distribution is due to~\citeauthor{DBLP:journals/siamcomp/SharirPH84}~\cite{DBLP:journals/siamcomp/SharirPH84}.
Their technique takes the initial distribution into account and considers the so-called \emph{occupation measure}~(e.g.,~\cite{Pitman_1977}) of the induced stochastic process.
The occupation measure associates each program state with its \emph{expected number of visits at the loop header}.
To see how occupation measures help determine the final distribution after the loop, the following observations are key:%
\begin{itemize}
    \item For states \emph{satisfying the loop guard}, the number of visits is a random value in $\nats \cup \{\infty\}$.
          (Unlike in deterministic programs, a probabilistic loop can return to the same state multiple times without necessarily entering an infinite cycle.)
          The \emph{expected} number of visits thus lies in the interval $[0,\infty]$.
    \item For states \emph{violating the loop guard}, the situation is different:
          Such states are visited \emph{at most once} at the loop header, since the loop terminates immediately upon entering them.
          Therefore, the \emph{expected} number of visits to such a terminal state $s$ equals
          \begin{align*}
                     & 0 \cdot Pr(\text{$s$ is never reached})                                                            &                                                  & \hspace{-1em}\pplus 1 \cdot Pr(\text{$s$ is reached once}) \\
              \pplus & \gray{\underbrace{\black{2 \cdot  Pr(\text{$s$ is reached twice}) \pplus \ldots}}_{\mathclap{=0}}}
                     &                                                                                                    & \hspace{-1em}\eeq Pr(\text{$s$ is reached once})
              ~.
          \end{align*}
          Consequently, the expected number of visits to $s$ is precisely the \emph{probability} of ever reaching $s$ --- and thus the probability of the loop terminating in $s$.
\end{itemize}
To summarize, if we denote the occupation measure with $\omIntro$, then we have
\[
    \kozensem\sem{\ploop{\phi}{C}}(\mu) ~=~ \sem{\neg\phi}(\omIntro)
    ~;
    \tag{as formally proved in \Cref{theorem:pw-equals-kozen}}
\]
that is, by \emph{restricting} $\omIntro$ to the states not satisfying the guard $\phi$, we obtain the final distribution $ \kozensem\sem{\ploop{\phi}{C}}(\mu)$ of the loop $\ploop{\phi}{C}$. %

The occupation measure of our running example program is visualized by means of the program's Markov chain unfolding in \Cref{fig:intro}.
The states in the bottom row are the ones not satisfying the loop guard. Observe that their expected visiting times indeed constitute the anticipated geometric distribution.

\paragraph{Reasoning about Loops with Occupation Measures.}
It can be shown (\Cref{theorem:pw:relation_to_om}) that the occupation measure of a loop $\ploop{\phi}{C}$ with initial distribution $\mu$ is the least solution $\omIntro$ of the fixed point equation
\[
    \omIntro
    \quad=\quad
    \mu ~+~ \kozensem\sem{C}(\sem{\phi}(\omIntro))
    \tag{$\ddagger$}
    \label{eq:occFixPoint}
    ~.
\]
Note that, unlike \eqref{eq:kozenFixedPoint}, this fixed point equation \emph{depends on the initial distribution} $\mu$.
\Cref{eq:occFixPoint} suggests the following \emph{guess-and-check} heuristic:
\begin{enumerate}
    \item\label{step1} Guess a candidate solution $M$ of \eqref{eq:occFixPoint} --- think of this as a \emph{loop invariant}.
    \item\label{step2} Check whether the candidate $M$ is indeed a solution.
          If yes, then $M$ is a (pointwise) \emph{upper bound} on the occupation measure $\omIntro$, and thus --- by the argument from above --- $\kozensem\sem{\ploop{\phi}{C}}(\mu) \leq \sem{\neg\phi}(M)$.
\end{enumerate}
We can exploit another key feature of occupation measures:
Their total mass equals the loop's \emph{expected runtime}\footnote{We define the expected runtime as the expected number of visits at the loop header.} on $\mu$~\cite{DBLP:journals/siamcomp/SharirPH84}.
Therefore, any \emph{finite} measure $M$ satisfying \eqref{eq:occFixPoint} witnesses \emph{positive almost sure termination} (PAST) of the loop on input $\mu$.
Since terminating loops preserve the probability mass of their input measure, we can conclude the following:
If a finite measure $M$ satisfies \eqref{eq:occFixPoint} and $\sem{\neg\phi}(M)$ is a probability distribution (i.e., has total mass 1), then the inequality $\kozensem\sem{\ploop{\phi}{C}}(\mu) \leq \sem{\neg\phi}(M)$ from Step \ref{step2} becomes an \emph{exact} equality $\kozensem\sem{\ploop{\phi}{C}}(\mu) = \sem{\neg\phi}(M)$.
See \Cref{theorem:invariants:mass_preserving_is_exact} for details.

\paragraph{Effective Verification and Synthesis with Closed-form Generating Functions.}
\label{par:gfIntro}
Steps (\ref{step1}) and (\ref{step2}) in the foregoing paragraph still involve dealing with measures defined on the \emph{infinite} set of program states.
This is where GFs shine:
For example, the --- infinite support --- occupation measure of the loop in \Cref{fig:intro} has as GF
\begin{align*}
    \omIntro
     & \qeq
    \red{1}\gfVar{x} \pplus \red{\tfrac 1 2}  \gfVar{x} \gfVar{c} \pplus \red{\tfrac 1 4}  \gfVar{x} \gfVar{c}^2 \pplus \ldots
    \qmorespace{+}
    \red{\tfrac 1 2}   \pplus \red{\tfrac 1 4}  \gfVar{c} \pplus \red{\tfrac 1 8}  \gfVar{c}^2 \pplus \ldots~,
    \intertext{which has the following (rational) \emph{closed form}:}
     & \qeq \frac{1 + 2 \gfVar{X}}{2 - \gfVar{C}}~.
\end{align*}
Under suitable assumptions, it is possible to verify equation \eqref{eq:occFixPoint} algorithmically for candidate measures $M$ represented in such a closed form, that is, to implement Step (\ref{step2}) of the guess-and-check heuristic.
Based on this, we furthermore show that Step (\ref{step1}) can be implemented via \emph{template-based synthesis}, using parameterized rational GFs as templates (as detailed in \Cref{sec:synthesis}).

\paragraph{Contributions.}
We revisit the relatively unknown utility of occupation measures in probabilistic loop analysis~\cite{DBLP:journals/siamcomp/SharirPH84} and present the following novel contributions:
\begin{enumerate}
    \item We show that a loop's occupation measure can be characterized as the least fixed point of a measure transformer defined \emph{compositionally by induction on the program structure} (\Cref{theorem:pw:relation_to_om}).
    \item We formally relate our semantic perspective to Kozen's semantics~\cite{DBLP:journals/jcss/Kozen81} (\Cref{theorem:pw-equals-kozen}), to occupation measures on operational Markov chains (\Cref{theorem:pw:relation_to_om}), and to expected runtimes~\cite{DBLP:conf/esop/KaminskiKMO16} (\Cref{theorem:pw:relation_to_ert}).
    \item We propose GFs (generating functions) as an effective encoding of occupation measures over program states on nonnegative integer variables~(\Cref{sec:synthesis}).
    \item We present the first \emph{fully automatic synthesis of occupation measures} (or approximations thereof) based on templates for rational GFs.
          We evaluate the effectiveness of our algorithm through experiments (\Cref{sec:implementation}).
\end{enumerate}

\paragraph{Outline.}
\Cref{sec:prelim} introduces preliminaries on measure theory and Markov chains.
\Cref{sec:syntax} specifies the program syntax; \Cref{sec:semantics} develops the denotational semantics and establishes equivalence to Kozen's semantics.
\Cref{ssec:invariants} presents loop invariants based on occupation measures.
\Cref{sec:synthesis} describes symbolic synthesis using rational GFs.
\Cref{sec:implementation} provides implementation details and experimental results.
\Cref{sec:relwork} discusses related work, and \Cref{sec:conclusion} concludes.

\section{Preliminaries}
\label{sec:prelim}
$\powerset{X}$ denotes the powerset of $X$.
$\nats$ is the set of natural numbers including $0$.

\subsection{Measure Theory}
\label{ssec:measures}
We briefly recall necessary notions of measure theory~\cite{Bogachev2007}.%
\begin{definition}[$\sigma$-Algebra]%
    \label{def:sigma-algebra}%
    Let $X$ be a set.
    A subset $\Sigma \subseteq \powerset{X}$ is called a \emph{$\sigma$-algebra (on $X$)} if
    (1) $X \in \Sigma$,
    (2) $A \in \Sigma$ implies $X\setminus A \in \Sigma$,
    and (3) for all countable sets $I$, $(A_i)_{i \in I} \in \Sigma^I$ implies $\bigcup_{i \in I} A_i \in \Sigma$.
\end{definition}%
The sets that constitute a $\sigma$-algebra are called \emph{measurable}.
A \emph{measure} assigns a nonnegative real number (or $\infty$) to each measurable set in a $\sigma$-algebra:%
\begin{definition}[Measure and Distribution]%
    \label{def:measure}%
    Let $\Sigma$ be a $\sigma$-algebra on $X$.
    \begin{itemize}
        \item A function $\mu\colon \Sigma \to \nnextreals$ is called a \emph{measure} if for all countable families $(A_i)_{i\in I}$ of pairwise disjoint sets in $\Sigma$, we have $\mu\left(\bigcup_{i \in I}A_i\right) = \sum_{i \in I} \mu(A_i)$ (implying $\mu(\emptyset) = 0$).
        We write $\abs{\mu} \coloneqq \mu(X)$ for the total \emph{mass} of $\mu$.
        The set of measures on $X$ is denoted by $\measuresOf{X}$ (the $\sigma$-algebra will always be clear from the context).
        \item A measure $\mu$ is a \emph{probability distribution} if $\abs{\mu} = 1$ and a \emph{sub}-probability distribution if $\abs{\mu} \leq 1$.
        The sets of probability and sub-probabilty distributions over $X$ are denoted by $\probDistOf{X}$ and $\distOf{X}$, respectively.
    \end{itemize}
\end{definition}%
Throughout the paper, we often refer to sub-probability distributions simply as \emph{distributions}.
If $\mu$ and $\nu$ are measures on the same $\sigma$-algebra, then their pointwise sum $\mu + \nu$ and, for any $c\geq 0$, the pointwise scaled $c \cdot \mu$ are also measures.

A function $f \colon X \to \nnextreals$, where $X$ is equipped with a $\sigma$-algebra $\Sigma$, is \emph{measurable} if $f^{-1}((a, \infty]) \in \Sigma$ for all $a \in \nnreals$.
The \emph{expected value} of $f$ under $\mu \in \measuresOf{X}$ is defined as the Lebesgue integral $\expectedValueOf[\mu]{f} \coloneqq \int{f \diff \mu}$.
The \emph{Dirac measure} $\dirac{x}$ of an element $x \in X$ and the \emph{indicator function} $\indicatorFct{A}$ of a measurable set $A \in \Sigma$ are %
\begin{align*}
    \delta_{x} \colon \Sigma \to \nnextreals~,
    ~ B \mapsto
    \begin{cases}
        1, & \text{if } x \in B, \\
        0, & \text{else};
    \end{cases}
    \quad~\text{and}~\quad
    \indicatorFct{A} \colon X \to \nnextreals,
    ~ y \mapsto
    \begin{cases}
        1, & \text{if } y \in A, \\
        0, & \text{else}.
    \end{cases}
\end{align*}
We often represent a measurable set $A \in \Sigma$ by a predicate $\phi_A$ on $X$; in this case, the \emph{Iverson bracket} $\iverson{\phi_A}$ denotes the indicator function $\indicatorFct{A}$.
We write $\iverson{\phi_A} \cdot \mu$ for the measure of $\mu$ restricted to $A$, that is, $(\iverson{\phi_A} \cdot \mu)(B) \coloneqq \mu(A \cap B)$ for $B \in \Sigma$.

In this paper, $X$ is usually a \emph{countable} set for which the collection of all subsets $\powerset{X}$ is a $\sigma$-algebra.
We thus identify a measure $\mu \in \measuresOf{X}$ with a function $X \to \nnextreals$ by writing $\mu(x) = \mu(\{x\})$ for $x \in X$.
In this setting, every function $f \colon X \to \nnextreals$ is measurable and Lebesgue integrals become countable sums: $\expectedValueOf[\mu]{f} = \accumulate{x \in X}{f(x)}{\mu(x)}$.
The set of measures $\measuresOf{X}$ is partially ordered by $\mu \leq \nu$ iff $\mu(x) \leq \nu(x)$ for all $x \in X$; this order makes $(\measuresOf{X}, \leq)$ a \emph{complete lattice} (i.e., every subset of $\measuresOf{X}$ has an infimum and a supremum w.r.t.\ $\leq$).

\subsection{Markov Chains}
\label{ssec:MCs}
Markov chains model stochastic processes in which transition probabilities depend only on the current state, disregarding any prior history~\cite{Chung1967}.

\begin{definition}[Markov Chain]%
    A \emph{Markov chain} is a pair $C = (\mcStates[C], \mcTransitionProbs[C])$ where $\mcStates[C]$ is a countable set of states, and $\mcTransitionProbs[C] \colon \mcStates[C] \to \probDistOf{\mcStates[C]}$ is the transition probability function.
    We write $\kernelEvalAt{\mcTransitionProbs[C]}{s}(t) \coloneqq \mcTransitionProbs[C](s)(t)$ for $s,t \in \mcStates[C]$.
    We omit the subscript $C$ if the Markov chain is clear from the context.
\end{definition}

A Markov chain induces a measure on the set of its infinite paths:%
\begin{definition}[Infinite Path]%
    Let $C$ be a Markov chain.
    An \emph{infinite path} in $C$ is a sequence $\pi = \pi_0 \pi_1 \ldots \in \mcStates[C]^{\omega}$ with $\kernelEvalAt{\mcTransitionProbs[C]}{\pi_i}(\pi_{i+1}) > 0$ for $i \in \nats$.
    $\mcPaths[C]$ denotes the (generally uncountable) set of all infinite paths in $C$.
\end{definition}
For a finite sequence $\bar{\pi} \in \mcStates[C]^{\ast}$, the \emph{cylinder set} $\cylinderSet{\bar{\pi}} \coloneqq \{ \pi \in \mcPaths[C] \mid \exists \pi' \in \mcPaths[C] \colon \pi = \bar{\pi} \pi' \}$ consists of all infinite paths in $C$ with prefix $\bar{\pi}$.
$\mcPaths[C]$ is equipped with the $\sigma$-algebra generated by all cylinder sets; see \cite{Chung1967} for details.%
\begin{definition}[(Probability) Measure induced by a Markov Chain]%
    \label{def:markovChainMeasure}%
    Let $C$ be a Markov chain.
    Given an \emph{initial measure} $\iota \in \measuresOf{\mcStates[C]}$, $C$ induces a unique measure $\mcInducedMeasure[C]{\iota}$ on $\mcPaths[C]$ with
    \(
    \mcInducedMeasure[C]{\iota}(\cylinderSet{\varepsilon}) \coloneqq \abs{\iota}
    \)
    and, for $\bar{\pi} \in \mcStates[C]^{+}$,
    $
        \mcInducedMeasure[C]{\iota}(\cylinderSet{\bar{\pi}})
        \coloneqq
        \iota(\bar{\pi}_0) \cdot \prod_{i = 0}^{\abs{\bar{\pi}} - 1}{ \kernelEvalAt{\mcTransitionProbs[C]}{\bar{\pi}_{i}}(\bar{\pi}_{i+1})} ~.
    $
\end{definition}
Notice that $\mcInducedMeasure[C]{\iota}$ in \Cref{def:markovChainMeasure} is a probability measure if and only if $\iota$ is a probability distribution. Our framework, however, requires defining $\mcInducedMeasure[C]{\iota}$ for arbitrary measures $\iota \in \measuresOf{\mcStates[C]}$, including those with mass $\abs{\iota} < 1$ \emph{and} $\abs{\iota} > 1$.

We can now define occupation measures, a key concept of this paper:%
\begin{definition}[Occupation Measure]%
    \label{def:occ_measure}%
    Let $C$ be a Markov chain with initial measure $\iota \in \measuresOf{\mcStates[C]}$.
    We define the \emph{occupation measure $\mcOccupationMeasure[C]{\iota} \in \measuresOf{\mcStates[C]}$} as $\mcOccupationMeasure[C]{\iota}(s) \coloneqq \expectedValueOf[{\mcInducedMeasure[C]{\iota}}]{\mcOccurences[C]{s}}$, which, for each state $s \in \mcStates[C]$, gives the expected value of the measurable function $\mcOccurences[C]{s} \colon \mcPaths[C] \to \nnextreals, \pi \mapsto \sum_{i \in \nats}{ \iverson{\pi_i = s} }$.
\end{definition}
Intuitively, $\mcOccurences[C]{s}$ counts how many times a state $s$ occurs in a path of $C$.
More general definitions of occupation measures~\cite{Pitman_1977} include a stopping time $T$, which prescribes when to stop counting; in our case, $T = \infty$.
For a given initial distribution $\iota \in \distOf{\mcStates[C]}$, $\mcOccupationMeasure[C]{\iota}(s)$ is the \emph{expected total number of visits} of $C$ to $s$.

\section{Syntax of Probabilistic Programs}
\label{sec:syntax}
In this section, we introduce a discrete \emph{probabilistic guarded command language} (\pgcl) à la \citeauthor{DBLP:series/mcs/McIverM05}~\cite{DBLP:series/mcs/McIverM05}.
Let $\pvars \neq \emptyset$ be a finite set of (program) variables ranged over by $x$, $y$, etc.
A \emph{program state} $\pstatea\colon \pvars \to \nats$ assigns a natural number $\pstatea(x) \in \nats$ to every variable $x \in \pvars$.
We write $\pstatea[x/n]$ for the \emph{updated} state in which $x \in \pvars$ maps to $n \in \nats$, and every other variable $x' \neq x$ maps to $\pstatea(x')$.
We denote the countable set of all program states by $\pstates \coloneqq \pvars \to \nats$.
\begin{definition}[\pgcl]%
	\label{def:pgcl}%
	The set of \pgcl-programs adheres to the grammar:
	\[
		\begin{array}{rclclr}
			C & \rightarrow & \pskip            & | & \pdiverge        & \text{(effectless $\mid$ endless loop)}           \\
			  & |           & \passign{x}{E}    & | & \psample{x}{\mu} & \text{(assignment $|$ prob.\ assignment)}         \\
			  & |           & \pchoice{C}{p}{C} & | & \pcomp{C}{C}     & \text{(prob.\ choice $|$ seq.\ composition)} \\
			  & |           & \pite{\phi}{C}{C} & | & \ploop{\phi}{C}  & \qquad\text{(conditional branch $|$ while loop)}        \\
		\end{array}
	\]
	where $x \in \pvars$, $E \colon \pstates \to \nats$ is an arithmetic expression, $\mu \colon \pstates \to \probDistOf{\nats}$ is a distribution expression, $p \in [0,1] \cap \mathbb{Q}$, and $\phi \subseteq \pstates$ is a predicate.
\end{definition}%
A program $C$ is called \emph{loop-free}, if it contains neither $\pwhile$ nor $\pdiverge$ statements.
For practical verification purposes, a concrete syntax for the $E$'s, $\mu$'s, and $\phi$'s in \Cref{def:pgcl} will be introduced in \Cref{sec:synthesis}.
We do not yet require such a syntax here and thus omit it.

Let us briefly go over each program statement from \Cref{def:pgcl}.
$\pskip$ does nothing.
$\pdiverge$ is a shorthand for the loop $\ploop{\ptrue}{\pskip}$.
There are two types of assignment statements:
the \emph{deterministic assignment} $\passign{x}{E}$ updates the value of $x$ in the current program state $\pstatea$ according to $E(\pstatea)$, whereas the \emph{probabilistic} assignment $\psample{x}{\mu}$ samples a value from the probability distribution $\mu(\pstatea)$ and assigns the result to $x$.
The probabilistic choice $\pchoice{C_1}{p}{C_2}$ executes $C_1$ with probability $p$ and $C_2$ with probability $1-p$.
Finally, $\pcomp{C_1}{C_2}$, $\pite{\phi}{C_1}{C_2}$, and $\ploop{\phi}{C}$ are standard sequential composition, conditional branching, and while loops, respectively.

\section{Semantics of PPs via Occupation Measures}
\label{sec:semantics}
We now introduce our semantics based on occupation measures and relate it to several other well-established semantic concepts for probabilistic programs. More precisely, we introduce our key semantic functional in \Cref{sec:semantics:occ_measure}, establish equivalence to Kozen's semantics in \Cref{ssec:semantics:equivalence_to_kozen}, and relate our semantic functional to operational Markov chains and expected runtimes in \Cref{ssec:semantics:relation_to_om_and_ert}.

\begin{table}[t]
    \centering
    \caption[Rules defining $\pw$ and $\kozensem$ inductively on the program structure.]{
        Inductive definitions of the posterior measure of program $C$ w.r.t.\ initial measure~$g$ via occupation measures (column $\pwapp{C}{g}$) and à la \citeauthor{DBLP:journals/jcss/Kozen81}~\cite[Semantics 2]{DBLP:journals/jcss/Kozen81} (column $\kozensemapp{C}{g}$).
        Notably, the two definitions differ only in the treatment of loops;
        in particular, for $\pw$, the fixed point variable ($\nu$) is of type $\measuresOf{\pstates}$ whereas for $\kozensem$, the fixed point variable ($T$) is of type $\measuresOf{\pstates} \to \measuresOf{\pstates}$.
        In \Cref{theorem:pw-equals-kozen}, we show that the definitions are equivalent.
    }
    \label{tab:semantics}
    \renewcommand{\arraystretch}{1.25}
    \setlength{\tabcolsep}{3pt}
    \begin{adjustbox}{max width=.999\linewidth}%
        \begin{tabular}{@{}l|l|l@{}}
            \toprule
            $C$                     & $\pwapp{C}{g}$                                                                                                                                                    & $\kozensemapp{C}{g}$                                                 \\
            \midrule
            $\pskip$                & \multicolumn{2}{l}{$g$}                                                                                                                                                                                                                  \\
            $\pdiverge$             & \multicolumn{2}{l}{$0$}                                                                                                                                                                                                                  \\
            $\passign{x}{E}$        & \multicolumn{2}{l}{$\lambda \pstatea.\, \accumulate{\pstatea'}{ \iverson{\pstatea'[x/E(\pstatea')] = \pstatea} }{g(\pstatea')}$}                                                                                                         \\

            $\psample{x}{\mu}$      & \multicolumn{2}{l}{$\lambda \pstatea.\, \accumulate{\pstatea'}{ \iverson{\pstatea'[x/\pstatea(x)] = \pstatea} \cdot \mu(\pstatea')(\pstatea(x)) }{g(\pstatea')}$}                                                                        \\
            \midrule
            $\pchoice{C_1}{p}{C_2}$ & $\pwapp{C_1}{p \cdot g} + \pwapp{C_2}{(1 - p) \cdot g}$                                                                                                           & $\kozensemapp{C_1}{p \cdot g} + \kozensemapp{C_2}{(1 - p)  \cdot g}$ \\                       %
            $\pcomp{C_1}{C_2}$      & $\pwapp{C_2}{ \pwapp{C_1}{g}}$                                                                                                                                    & $\kozensemapp{C_2}{ \kozensemapp{C_1}{g}}$                           \\   %
            $\pite{\phi}{C_1}{C_2}$ & $\pwapp{C_1}{\iverson{\phi} \cdot g} + \pwapp{C_2}{\iverson{\neg \phi} \cdot g}$
                                    & $\kozensemapp{C_1}{\iverson{\phi} \cdot g} + \kozensemapp{C_2}{\iverson{\neg \phi} \cdot g}$                                                                                                                                             \\
            $\ploop{\phi}{C_0}$     & $\iverson{\neg \phi} \cdot \bigl( \lfp \nu. ~ g + \pwapp{C_0}{\iverson{\phi} \cdot \nu}\bigr)$                                                                    &
            $\Bigl(\lfp T.~  \lambda \nu.\iverson{\neg \phi} \cdot \nu + T\bigl(\kozensemapp{C_0}{\iverson{\phi} \cdot \nu}\bigr)\Bigr)(g)$                                                                                                                                    \\
            \bottomrule
        \end{tabular}%
    \end{adjustbox}
\end{table}

\subsection{Loop Semantics via Occupation Measures}
\label{sec:semantics:occ_measure}
Towards defining our semantic functional, we define arithmetic operations on measures pointwise, i.e., for $\circledcirc \in \{+,\cdot\}$, we let  $f \circledcirc g = \lambda \pstatea. f(\pstatea) \circledcirc g(\pstatea)$. Moreover, recall that $\iverson{\phi}$ is the Iverson bracket (i.e., the indicator function) of the predicate $\phi$. Now consider the following:%
\begin{definition}[$\pw$]%
    The function $\pw\sem{C} \colon \measuresOf{\pstates} \to \measuresOf{\pstates}$ is defined inductively on the structure of $C$ by the rules in~\Cref{tab:semantics} (middle column). We call $\pwapp{C}{g}$ the \emph{posterior measure of $C$ w.r.t.\ initial measure $g$.}
\end{definition}%
If the initial measure $g$ is a probability distribution, the intuition on the posterior measure $\pwapp{C}{g}$ is straightforward and as expected:
\begin{center}
    \emph{If $g$ is a probability distribution, then $\pwapp{C}{g}$ is the subprobability\footnote{The \emph{leaked} mass $\abs{g} - \abs{\pwapp{C}{g}}$ is the probability that $C$ diverges on $g$.} distribution of states obtained from executing $C$ on the initial distribution $g$.}
\end{center}
Let us now go over each rule in the middle column of \Cref{tab:semantics}.
A $\pskip$ statement does not modify the initial measure. Since $\pdiverge$ does not terminate in any state, the posterior measure is constantly $0$. For assignments, the mass of the posterior measure in state $\pstatea$ is obtained from accumulating the mass of all states $\pstatea'$ which, after executing the assignment, become $\pstatea$.
For a deterministic assignment $\passign{x}{E}$, these are all states $\pstatea
$ with $E(\pstatea')= \pstatea(x)$ and such that $\pstatea'$ and $\pstatea$ coincide on all other variables. Similarly, for a probabilistic assignment $\psample{x}{\mu}$, the incoming mass $g(\pstatea')$ is weighted by the probability of sampling the value $\pstatea(x)$ from $\mu(\pstatea')$.
Sequential composition $\pcomp{C_1}{C_2}$ propagates the measure through $C_1$ first, and then through $C_2$.
For a conditional branch $\pite{\phi}{C_1}{C_2}$ the measure is restricted to the states satisfying $\phi$ and the states satisfying $\neg \phi$, and then propagated through $C_1$ and $C_2$, respectively.
The final result is obtained by adding the measures resulting from executing the two branches.

The posterior measure of a loop $C = \ploop{\phi}{C_0}$ involves a least fixed point construction. It is this construction where occupation measures come into a play --- a fact we formalize and prove in \Cref{ssec:semantics:relation_to_om_and_ert}. Let us define some auxiliary notation. Given an initial measure $g$, we call
\[
    \charfun \colon \measuresOf{\pstates} \to \measuresOf{\pstates},
    \qquad
    \nu ~\mapsto~ g + \pwapp{C_0}{\iverson{\phi} \cdot \nu}
\]
the \emph{characteristic function of $C$ w.r.t.\ $g$}.
Note that $\pwapp{C}{g} = \iverson{\neg \phi} \cdot (\lfp \nu.\, \charfun(\nu))$ by definition.
If $g$ is a probability distribution, then $(\lfp \nu.\, \charfun(\nu))(\pstatea)$ is the \emph{expected number of times state $\pstatea$ is encountered at the loop head when executed on the initial measure $g$} (as formalized in \Cref{theorem:pw:relation_to_om}) --- the least fixed point thus denotes an occupation measure. After --- so to speak --- cutting off all states from $\lfp \nu.\, \charfun(\nu)$ that satisfy the loop guard, we obtain precisely the (sub-)distribution of final states reached after executing $C$ on $g$ (see \Cref{theorem:pw-equals-kozen}).\\
\begin{wrapfigure}[10]{r}{0.5\textwidth}%
    \vspace{-25pt}%
    \begin{minipage}{0.5\textwidth}%
        \input{figures/geometric_loop}
    \end{minipage}
\end{wrapfigure}%
\begin{example}
    The loop $C$ in \Cref{prog:geometric_loop} generates a geometric distribution.
    In each iteration, the loop body $C_0 = \pchoice{\passign{c}{c+1}}{\sfrac 1 2}{\passign{x}{0}}$ either increments the \emph{counter variable} $\progvar{c}$ with probability $\sfrac{1}{2}$, or terminates the loop by setting $\progvar{x}$ to zero.
    For the initial measure $g = \iverson{x = 1 \land c = 0}$, we have
    \[
        \lfp \charfun ~=~ \sum_{k=0}^{\infty}{2^{-k} \cdot \iverson{x = 1 \land c = k}} + \sum_{k=0}^{\infty}{2^{-(k+1)} \cdot \iverson{x = 0 \land c = k}}~.
    \]
    For every $\pstatea \in \pstates$, the value $(\lfp \charfun)(\pstatea)$ gives the expected number of times $\pstatea$ is encountered at the loop header when executing $C$ on $g$.\footnote{We will show how to obtain this fixed point using~\Cref{theorem:invariants:mass_preserving_is_exact} in~\Cref{example:loop-invariants-mass-preserving-invariant}.}
    This is intuitive: any state $\pstatea$ with $\pstatea(x) \not\in\{0,1\}$ is \emph{never} visited since we execute the loop on $x=1$ and $x$ can only ever be set to $0$. 
    All states with $\pstatea(x) \in\{0,1\}$ are visited at most once. Hence, for such states, $(\lfp\,\charfun)(\pstatea)$ coincides with the probability of reaching $\pstatea$.

    Filtering on $x \neq 1$, we obtain the anticipated geometric distribution:
    \[
        \pwapp{C}{g} \eeq \iverson{x \neq 1} \cdot (\lfp\,\charfun) \eeq \sum_{k=0}^{\infty}{2^{-(k+1)} \cdot \iverson{x = 0 \land c = k}}~.
    \]
    Notably, despite the fact that $g$ and $\pwapp{C}{g}$ are probability distributions, the fixed point is a measure with mass $\abs{\lfp\,\charfun} = \sum_{k=0}^{\infty}{2^{-k}} + \sum_{k=0}^{\infty}{2^{-(k+1)}} = 3 > 1$, showing the need to consider arbitrary measures in the definition of $\pw$.
\end{example}
Well-definedness of our least fixed point construction follows from Tarski's theorem \cite{Tarski1955}: Measures $\measuresOf{\pstates}$ equipped with the pointwise order $\leq$ are a complete lattice and $\pw\sem{C}$ (and thus $\charfun$) is a monotonic function on this lattice.
We collect the latter fact and several other properties in the next theorem, which follows from the equivalence to Kozen's semantics (\Cref{theorem:pw-equals-kozen}) and~\cite{DBLP:journals/jcss/Kozen81}:%
\begin{theorem}[Properties of $\pw$]%
    Let  $C\in\pgcl$. Then:
    \begin{enumerate}
        \item $\pw\sem{C}$ is \emph{monotonic}, i.e., for all $\mu, \nu \in \measuresOf{\pstates}$, we have
              \[
                  \mu \leq \nu
                  \quad\text{implies}\quad
                  \pwapp{C}{\mu} \leq \pwapp{C}{\nu}~.
              \]
        \item $\pw\sem{C}$ is \emph{linear}, i.e., for all $\mu, \nu \in \measuresOf{\pstates}$ and all $a \in \nnextreals$, we have
              \[
                  \pwapp{C}{\mu + \nu} = \pwapp{C}{\mu} + \pwapp{C}{\nu}  \quad\text{and}\quad  \pwapp{C}{a \cdot \mu} = a \cdot \pwapp{C}{\mu}~.
              \]
        \item $\pw\sem{C}$ is \emph{feasible}, i.e., for all $\mu \in \measuresOf{\pstates}$, we have
              $\abs{\pwapp{C}{\mu}} \leq \abs{\mu}$.
        \item $\pw\sem{C}$ \emph{preserves distributions}, i.e., for all $\mu \in \measuresOf{\pstates}$, we have
              \[
                  \mu \in \distOf{\pstates}
                  \quad\text{implies}\quad
                  \pwapp{C}{\mu} \in \distOf{\pstates}~.
              \]
    \end{enumerate}
\end{theorem}

\subsection{Equivalence to Kozen's Semantics}
\label{ssec:semantics:equivalence_to_kozen}
For each program $C$, Kozen's~\cite{DBLP:journals/jcss/Kozen81} measure transformer semantics $\kozensem\sem{C}$ is defined inductively on the structure of $C$ by the rules in \Cref{tab:semantics} (right column).
The definitions of $\pw$ and $\kozensem$ are identical \emph{except for the case of loops}.
This difference, however, is crucial:
Whereas $\pw$'s loop characteristic functional is of type $\measuresOf{\pstates} \to \measuresOf{\pstates}$, that is, a measure transformer, Kozen's fixed point functional is of the \emph{higher-order} type
\[
    \bigl(\measuresOf{\pstates} \to \measuresOf{\pstates}\bigr) ~\to~ \bigl(\measuresOf{\pstates} \to \measuresOf{\pstates} \bigr)~,
\]
that is, a \emph{measure-transformer transformer}.
This is because Kozen's construction does \emph{not take the initial measure $g$ into account}: $\lfp T.\, \lambda \nu. \iverson{\neg \phi} \cdot \nu + T(\kozensemapp{C_0}{\iverson{\phi} \cdot \nu})$ is a function of type $\measuresOf{\pstates} \to \measuresOf{\pstates}$ that maps \emph{every} initial measure to the corresponding posterior measure. It is the fact that our functional \emph{does} take the initial measure into account which significantly simplifies invariant-based reasoning and enables our novel template-based invariant synthesis approach.
Put simply: whereas Kozen's approach requires to find a higher order object, we can get away with a lower order object.

We now state and prove the main result of this section:%
\begin{theorem}[Semantic Equivalence]%
    \label{theorem:pw-equals-kozen}%
    For all programs $C$, $\pw\sem{C} = \kozensem\sem{C}$.
\end{theorem}
Hence, our $\pw$ transformer based on occupation measures indeed soundly determines the sought-after posterior measures even though the functionals for loops are much simpler than the functionals involved in Kozen's semantics. The rest of this section is devoted to the proof of \Cref{theorem:pw-equals-kozen}.

\begin{proof}[of \Cref{theorem:pw-equals-kozen}]%
    The key idea is to exploit the \emph{Kozen duality} between measure transformers and so-called \emph{expectation transformers}~\cite{DBLP:journals/jcss/Kozen85}: Given a random variable $f \colon \pstates \to \nnextreals$, denote by
    \[
        \wpreapp{C}{f} \colon (\pstates \to \nnextreals), \qquad \pstatea \mapsto \expectedValueOf[\kozensemapp{C}{\indicatorFct{\pstatea}}]{f}
    \]
    the \emph{weakest pre-expectation of $C$ w.r.t.\ $f$}, that is, the function which maps every initial program state $\pstatea$ to the expected value of $f$ w.r.t.\ the distribution of final states $\kozensemapp{C}{\indicatorFct{\pstatea}}$ reached after executing $C$ on $\pstatea$. The Kozen duality states that, for every initial measure $\mu$, we have
    \[
        \expectedValueOf[\mu]{\wpreapp{C}{f}} = \accumulate{\pstatea}{\wpreapp{C}{f}(\pstatea)}{\mu(\pstatea)} = \accumulate{\pstatea}{f(\pstatea)}{\kozensemapp{C}{\mu}(\pstatea)} = \expectedValueOf[\kozensemapp{C}{\mu}]{f}.
    \]
    It immediately follows from \cite[Theorem~4.5]{DBLP:journals/pacmpl/ZhangZK024}\footnote{Zhang et al.'s framework~\cite{DBLP:journals/pacmpl/ZhangZK024} is to be instantiated with the semiring $(\nnextreals, {+}, {\cdot}, 0, 1)$.} that $\pw$ analogously satisfies
    \[
        \expectedValueOf[\mu]{\wpreapp{C}{f}} = \accumulate{\pstatea}{\wpreapp{C}{f}(\pstatea)}{\mu(\pstatea)} = \accumulate{\pstatea}{f(\pstatea)}{\pwapp{C}{\mu}(\pstatea)} = \expectedValueOf[\pwapp{C}{\mu}]{f}.
    \]
    This yields $\expectedValueOf[\pwapp{C}{\mu}]{f} = \expectedValueOf[\kozensemapp{C}{\mu}]{f}$ for all $f \in \pstates \to \nnextreals$ and all initial measures $\mu$.
    Define for every state $\pstatea$ the random variable $f_s = \lambda s'.\, \iverson{s' = s}$.
    Then
    \[
        \pwapp{C}{\mu}(s) ~{}={}~ \expectedValueOf[\pwapp{C}{\mu}]{f_s} ~{}={}~ \expectedValueOf[\kozensemapp{C}{\mu}]{f_s} ~{}={}~ \kozensemapp{C}{\mu}(s)~,
    \]
    and thus $\pw\sem{C} = \kozensem\sem{C}$, which proves \Cref{theorem:pw-equals-kozen}.
    \qedhere
\end{proof}

\subsection{Relating Occupation Measures to Expected Runtimes}
\label{ssec:semantics:relation_to_om_and_ert}
In this section, we formally show that the least fixed point in the definition of $\pw$ for loops\footnote{We restrict to loop with \emph{loop-free bodies} for simplicity; with loopy bodies, \Cref{def:big-step-operational-semantics} and subsequent theorems would need to handle potential divergence of inner loops.} is indeed an occupation measure (in a suitably defined Markov chain; see \Cref{def:occ_measure}); we also relate this measure to the loop's \emph{expected runtime}.%
\begin{definition}[Big-step Operational Markov Chain]%
    \label{def:big-step-operational-semantics}%
    Let $C = \ploop{\phi}{C_0}$ be a loop with loop-free body $C_0$.
    We define the Markov chain $\loopMCsem{C}$ with states $S_{\loopMCsem{C}} \coloneqq \pstates \cup \{ \downarrow \}$ and transition probability function
    \[
        P_{\loopMCsem{C}}(\pstatea) ~\coloneqq~ \begin{cases}
            \pwapp{C_0}{\dirac{\pstatea}}, & \text{if } \pstatea \in \pstates \text{ and } \pstatea \models \phi, \\
            \dirac{\downarrow},            & \text{else}.
        \end{cases}
    \]
\end{definition}
Every state of $\loopMCsem{C}$ is either a program state or  $\downarrow$, the latter indicating termination.
For every state $\pstatea$ satisfying the loop guard $\phi$, the probability of transitioning to some state $\pstatea'$ is the probability of reaching $\pstatea'$ from $\pstatea$ through \emph{one} loop iteration.
Every such transition in $\loopMCsem{C}$ thus corresponds to an \emph{entire} loop iteration, hence the term \emph{big-step}.
With this perspective, we get:%
\begin{theorem}[Relation of $\pw$ and $\occupationMeasure$]%
    \label{theorem:pw:relation_to_om}%
    For $C = \ploop{\phi}{C_0}$, $C_0$ loop-free:
    \[
        \forall\, \text{initial measures}~g \in \measuresOf{\pstates}\colon \forall\pstatea\in\pstates\colon \quad
        \left(\lfp \charfun\right)(\pstatea) = \mcOccupationMeasure[\loopMCsem{C}]{g}(\pstatea)~.
    \]
\end{theorem}
\Cref{theorem:pw:relation_to_om} formalizes our claims from \Cref{sec:semantics:occ_measure}:
If $g$ is a probability measure, then $\left(\lfp \charfun\right)(\pstatea)$ is the expected number of times that state $\pstatea$ is visited at the loop head when executing $C$ on $g$.

Let us now relate a loop's occupation measure to its expected runtime.
For every loop $C = \ploop{\phi}{C_0}$ with loop-free body $C_0$, we let
\[
    \ertappnew{C} \colon \pstates \to \nnextreals, \quad \pstatea \quad\mapsto\quad \substack{\text{\normalsize \emph{expected number of loop guard $\phi$}} \\ \text{\normalsize \emph{evaluations when executing $C$ on $\pstatea$}}} \quad .
\]
The function $\ertappnew{C}$ can be defined in a weakest pre-condition style by induction on the structure of $C$; see~\cite{DBLP:conf/esop/KaminskiKMO16} for details.
We obtain the following correspondence:%
\begin{theorem}[Relation of $\pw$ and $\ert$]%
    \label{theorem:pw:relation_to_ert}%
    For $C = \ploop{\phi}{C_0}$:
    \[
        \forall\, \text{initial measures}~g \in \measuresOf{\pstates}\colon \qquad \abs{\lfp \charfun} = \accumulate{\pstatea \in \pstates}{\ertappnew{C}(\pstatea)}{g(\pstatea)}~.
    \]
\end{theorem}
In particular, if $g = \dirac{\pstatea}$ is a Dirac measure, the mass of ${\lfp \charfun}$ equals the expected number of guard evaluations when executing $C$ on $\pstatea$.
This observation yields a connection to \emph{positive almost sure termination} (PAST) --- a property which, by definition, holds for a loop $C$ and initial state $\pstatea$, if $\ertappnew{C}(\pstatea)< \infty$~\cite{DBLP:conf/esop/KaminskiKMO16}.
\begin{corollary}%
    \label{corollary:past_if_occupation_measure_finite}%
    Let $C = \ploop{\phi}{C_0}$ be a loop with loop-free body $C_0$ and let $\pstatea \in \pstates$.
    Then $C$ is PAST on $\pstatea$ if and only if for $g = \dirac{\pstatea}$, we have $\abs{\lfp \charfun} < \infty$.
\end{corollary}

\section{Occupation Measure-Based Loop Invariants}
\label{ssec:invariants}
In this section, we introduce the notion of \emph{occupation invariant}, which allows determining or upper-bounding a loop's posterior measure and forms the basis of our invariant synthesis approach.
Our invariants correspond to \cite[p.\ 306]{DBLP:journals/siamcomp/SharirPH84}, if one uses (countable) discrete-time Markov chains as the program model.%
\begin{definition}[Occupation Invariant]%
    \label{def:superinvariant}%
    Let $C = \ploop{\phi}{C_0}$ and let $g\in\measuresOf{\pstates}$ be an initial measure.
    We call a measure $I \in \measuresOf{\pstates}$ an \emph{(occupation) superinvariant of $C$ w.r.t.\ $g$}, if $\charfun(I) \leq I$.
    If $\charfun(I) = I$, then we call $I$ an \emph{(occupation) invariant}.
\end{definition}
Note that --- as is to be expected from invariant-based reasoning --- determining $\charfun(I)$ requires us to reason about a \emph{single guarded loop iteration} only.
Superinvariants yields upper bounds on a loop's posterior measure as follows:%
\begin{theorem}%
    \label{thm:invariants}%
    If \( I \) is a superinvariant of \(C = \ploop{\phi}{C_0}\) w.r.t.\ \(g\), then \[\pwapp{C}{g} \lleq \iverson{\neg \phi} \cdot I~. \]
\end{theorem}
\begin{proof}
    By Park induction \cite{Park1969}, $\charfun(I) \lleq I$ implies $\lfp \nu.~\charfun(\nu) \lleq I$. Hence,
    \begin{align*}
                            & \charfun(I) \lleq I                                                              \\
        \text{implies}\quad & \lfp \nu.~\charfun(\nu) \lleq I        \tag{Park induction}                      \\
        \text{implies}\quad & \iverson{\neg\phi}\cdot \lfp \nu.~\charfun(\nu)  \lleq \iverson{\neg\phi}\cdot I \\
        \text{implies}\quad & \pwapp{C}{g} \lleq \iverson{\neg\phi}\cdot I ~.
        \tag*{(definition)\quad\qedhere}
    \end{align*}%
    \normalsize%
\end{proof}

\begin{example}%
    \label{example:loop-invariants-intuition}%
    Recall the geometric loop $C$ with initial probability distribution $g = \iverson{x = 1 \land c = 0}$ from~\Cref{prog:geometric_loop}.
    Let us now upper-bound $\pwapp{C}{g}$ by means of a superinvariant.
    To this end, consider the measure
    \[
        I(x,c) \qeq 2^{-(c+1)} \cdot \iverson{x = 0} \pplus 2^{-c} \cdot \iverson{x = 1}~.
    \]
    We have $\charfun(I)\leq I$ and thus $\lfp\,\charfun \leq I$. Hence, $I(\pstatea)$ upper-bounds the expected number of times a state $\pstatea \in \pstates$ is encountered at the loop header when executing $C$ on $g$.
    This aligns with our intuition:
    If $\pstatea(x) \not\in\{0,1\}$, then $\pstatea$ is \emph{never} visited, since we execute the loop on $x=1$ and $x$ can only ever be set to $0$.
    On the other hand, states satisfying $\pstatea(x) \in \{0,1\}$ are visited at most once:
    after $\pstatea(c)$ iterations of incrementing $c$ if $\pstatea(x) = 1$, and after one more iteration if $\pstatea(x) = 0$.
    Thus, $I(\pstatea)$ upper-bounds the probability of reaching $\pstatea$. By \Cref{thm:invariants}, we now get, as expected,
    \[
        \pwapp{C}{g}
        \quad\leq\quad \iverson{x \neq 1}\cdot I
        \qeq 2^{-(c+1)} \cdot \iverson{x = 0} ~;
    \]
    that is, the probability of terminating in a state $\pstatea$ with $\pstatea(x) \neq 0$ is $0$, and for all states with $\pstatea(x)=0$, the probability of terminating in $\pstatea$ is at most $2^{-(\pstatea(c)+1)}$.
\end{example}

If $I$ is a superinvariant with finite mass and moreover that mass is the same as the mass of the input measure $g$, then we can draw an even stronger conclusion:%
\begin{theorem}[Proof Rule for the Exact Posterior Measure]%
    \label{theorem:invariants:mass_preserving_is_exact}%
    Consider a loop $C = \ploop{\phi}{C_0}$ with loop-free body $C_0$ and an initial measure $g\in\measuresOf{\pstates}$.
    Let $I\in\measuresOf{\pstates}$ be a measure such that all of the following hold:
    \begin{enumerate}
        \item
              $I$ is an occupation superinvariant, i.e.,\quad$\charfun(I) \lleq I$.
        \item
              $I$ has finite mass, i.e.,\quad$\abs{I} < \infty$
        \item
              $\iverson{\neg \phi} \cdot I$ has a mass equal to the mass of $g$, i.e.,\quad$\abs{\iverson{\neg \phi} \cdot I} = \abs{g}$.
    \end{enumerate}
    Then $\iverson{\neg\phi} \cdot I$ is the \emph{exact} posterior measure, i.e., $\pwapp{C}{g} = \iverson{\neg \phi} \cdot I$.
\end{theorem}
\begin{proof}[sketch]%
    Due to \Cref{theorem:pw:relation_to_ert}, $\abs{\lfp \charfun}$ is the expected runtime $\mathit{ERT}$ of the loop on input distribution $g$.
    Since $I$ is an occupation superinvariant (Condition (1)), we can conclude by Park induction that $\lfp \charfun \leq I$ and thus also $\abs{\lfp \charfun} \leq \abs{I}$.
    Together with Condition (2), we get that $\mathit{ERT} = \abs{\lfp \charfun} \leq \abs{I} < \infty$ and so the expected runtime of the loop is finite.
    Finite expected runtime implies almost sure termination~\cite{DBLP:conf/esop/KaminskiKMO16}.
    Almost sure termination implies that the mass of the input measure $g$ is equal to the mass of the output measure, so $\abs{\pwapp{C}{g}} \eeq \abs{g}$.
    Now, by Condition (3), we get%
    \begin{align*}
        \abs{\pwapp{C}{g}} \eeq \abs{g} \eeq \abs{\iverson{\neg \phi} \cdot I}~. \tag{$\dagger\dagger$}
    \end{align*}%
    By Condition (1) and \Cref{thm:invariants}, we moreover have%
    \begin{align*}
        \pwapp{C}{g} \lleq \iverson{\neg \phi} \cdot I~. \tag{$\ddagger\ddagger$}
    \end{align*}
    Assume, for contradiction, that there is a state $\orange{s'}$ such that $\pwapp{C}{g}(\orange{s'}) < \bigl(\iverson{\neg\phi} \cdot I\bigr)(\orange{s'})$.
    Then,%
    \begin{align*}
        \abs{\iverson{\neg \phi} \cdot I}
         & \eeq 				\abs{\pwapp{C}{g}} 																				\tag{by ($\dagger\dagger$)}                                                                                                                         \\
         & \eeq 				\pwapp{C}{g}(\orange{s'}) \pplus \sum_{s \neq \orange{s'}} \pwapp{C}{g}(s)                                                                                                                 \\
         & \purple{\morespace{<}}	\bigl(\iverson{\neg\phi} \cdot I\bigr)(\orange{s'}) \pplus \sum_{s \neq \orange{s'}} \pwapp{C}{g}(s) 					\tag{by assumption above}                                          \\
         & \lleq				\bigl(\iverson{\neg\phi} \cdot I\bigr)(\orange{s'}) \pplus \sum_{s \neq \orange{s'}} \bigl(\iverson{\neg\phi} \cdot I\bigr)(s)	\tag{by ($\ddagger\ddagger$)}                               \\
         & \eeq				\abs{\iverson{\neg\phi} \cdot I}~.																	\tag*{\purple{Contradiction:}\quad$\abs{\iverson{\neg\phi} \cdot I} \mathrel{\purple{<}} \abs{\iverson{\neg\phi} \cdot I}$\quad\qedhere}
    \end{align*}%
\end{proof}
Note that an occupation superinvariant that fulfills conditions (1) and (2) for initial distribution $g = \dirac{\pstatea}$ is sufficient to prove PAST on state $\pstatea$ by~\Cref{corollary:past_if_occupation_measure_finite}.

\begin{example}%
    \label{example:loop-invariants-mass-preserving-invariant}%
    We continue the analysis of the geometric loop $C$ from~\Cref{prog:geometric_loop}.
    In~\Cref{example:loop-invariants-intuition}, we established the occupation superinvariant
    $I = 2^{-(c+1)} \cdot \iverson{x = 0} + 2^{-c} \cdot \iverson{x = 1}$
    to conclude $\pwapp{C}{g} \leq \iverson{x \neq 1}\cdot I$.
    Here, we show that this inequality is actually an equality.
    We have $\abs{I} = \sum_{k\in\nats}{2^{-(k+1)}} + \sum_{k\in\nats}{2^{-k}} = 3 < \infty$, and $\abs{\iverson{x \neq 1} \cdot I} = \sum_{k\in\nats}{2^{-(k+1)}} = 1 = \abs{g}$.
    By~\Cref{theorem:invariants:mass_preserving_is_exact}, we now obtain
    \[
        \pwapp{C}{g}
        \qeq \iverson{x \neq 1}\cdot I
        \qeq 2^{-(c+1)} \cdot \iverson{x = 0}~,
    \]
    and conclude that, for all states $\pstatea \in \pstates$, the probability of terminating in $\pstatea$ is \emph{exactly} $2^{-(\pstatea(c)+1)}$ if $\pstatea(x)=0$, and $0$ otherwise.
\end{example}

\section{Invariant Synthesis}
\label{sec:synthesis}
In this section, we present an automatic, template-based synthesis technique for occupation invariants (see~\Cref{def:superinvariant}).
Template-based probabilistic loop invariant synthesis is common in the literature (e.g.,~\cite{DBLP:conf/sas/KatoenMMM10,DBLP:conf/tacas/BatzCJKKM23}); the novelty here is the use of \emph{generating function templates} representing collections of measures.
Given such a template, we can automatically check whether there exists a parameter instantiation that yields a valid occupation invariant.
We begin by introducing generating functions and adapting our $\pw$-semantics to operate on this domain.

\subsection{Generating Function Semantics}
Let \(\vec{X} = (X_v)_{v \in \pvars}\) be a family of formal variables or \emph{indeterminates}, one per program variable $v \in \pvars$ (in examples, we simply write $X$, $Y$ instead of $X_x$, $X_y$, etc.).
Formally, generating functions are (infinite) sequences of values, commonly represented as formal power series.
See~\cite{Wilf2005,DBLP:reference/hfl/Kuich97} for detailed expositions.%
\begin{definition}[Formal Power Series]%
    Let $\mathrm{R}$ be a semiring.
    An $\mathrm{R}$-valued \emph{formal power series} (FPS) in indeterminates $\vec{X}$ is a function $\measureGFa \colon \pstates \to \mathrm{R}$ denoted $\measureGFa = \sum_{\pstatea \in \pstates}{ \gfCoefficient{\measureGFa}{\pstatea} \cdot \vec{X}^{\pstatea} }$, where \(\vec{X}^\pstatea = \prod_{v \in \pvars}{ X_v^{\pstatea(v)} }\).
    A \emph{polynomial} is an FPS $\measureGFa$ with $\gfCoefficient{\measureGFa}{\pstatea} \neq 0$ for finitely many $\pstatea \in \pstates$.
    The set of FPSs is denoted by $\mathrm{R}[[\vec{X}]]$.
    The set of polynomials is denoted by $\mathrm{R}[\vec{X}]$.

    $\mathrm{R}[[\vec{X}]]$ and $\mathrm{R}[\vec{X}]$ are semirings with (pointwise) sum and (Cauchy) product:
    \begin{align*}
        \measureGFa + \measureGFb     {}~\coloneqq~ \sum_{\pstatea \in \pstates}{ (\gfCoefficient{\measureGFa}{\pstatea} + \gfCoefficient{\measureGFb}{\pstatea}) \cdot \vec{X}^{\pstatea} }~,
        \quad
        \measureGFa \cdot \measureGFb  {}~\coloneqq~ \sum_{\pstatea \in \pstates}{ \Big( \!\!\sum_{\pstatea_1 + \pstatea_2 = \pstatea}{ \gfCoefficient{\measureGFa}{\pstatea_1} \cdot \gfCoefficient{\measureGFb}{\pstatea_2} }\Big) \cdot \vec{X}^{\pstatea} }~.
    \end{align*}
    If $\mathrm{R}$ is a ring, then so are $\mathrm{R}[[\vec{X}]]$ and $\mathrm{R}[\vec{X}]$.
\end{definition}
We employ FPSs to represent measures on the set of program states as follows:%
\begin{definition}[Generating Function]%
    \label{def:generatingFunction}%
    The \emph{generating function \(\measureGF{\mu}\)} of a measure \(\mu \in \measuresOf{\pstates}\) is the FPS
    \(
    \measureGF{\mu} \coloneqq \sum_{\pstatea \in \pstates}{ \mu(\pstatea) \cdot \vec{X}^{\pstatea} } \in \nnextreals[[\vec{X}]].
    \)
    Conversely, every FPS \(\measureGFa \in \nnextreals[[\vec{X}]]\) defines a measure \(\mu_{\measureGFa} \in \measuresOf{\pstates}\) with \(\mu_{\measureGFa}(\pstatea) = \gfCoefficient{\measureGFa}{\pstatea}\) for \(\pstatea \in \pstates\).
\end{definition}
Generating functions as in \Cref{def:generatingFunction} do not offer much benefit yet; they are merely a different notation for an infinite expression representing a measure.
To enable synthesis of occupation invariants represented as generating functions, our key idea is to exploit that a large class of measures admit a \emph{closed form representation} of their generating function; in this paper, we consider only \emph{rational functions} (i.e., fractions of polynomials) as closed forms.
For example, the final distribution of the geometric loop running example (\Cref{prog:geometric_loop}) is
\[
    \measureGF{\geometric{\sfrac{1}{2}}} \eeq \frac{1}{2} + \frac{1}{4}  \gfVar{c} + \frac{1}{8}  \gfVar{c}^2 + \ldots ~, \qquad
    \text{with rational closed form }
    ~\frac{1}{2 - \gfVar{c}}
    ~.
\]
Before we explain the details, let us first define a \pgcl fragment that preserves measures in rational closed form.

\paragraph{The \redip Fragment of \pgcl.}
We consider the syntactic fragment of \emph{rectangular discrete probabilistic programs} (\redip) of \pgcl by~\cite{DBLP:conf/cav/ChenKKW22}.
This fragment enables defining a program semantics on generating functions of probability measures which preserves certain closed form representations of the generating functions.
We generalize this result to the generating functions from \Cref{def:generatingFunction}.

The syntax of \redip is given in~\Cref{tab:redip}.
It includes multiple statements to modify variables: assignment of constants, a decrement operation, and an increment operation.
The latter draws a number of i.i.d.\ samples from a probability distribution over $\nats$ and increments the variable's value by the sum of the samples (in plain \pgcl, this operation can be expressed by means of a loop; however, in \redip it is considered a primitive operation rather than a loop).
The number of samples is determined by the value of another program variable in the current state.
The language supports branching and loops with rectangular guards of the form $\progvar{v} < n$ with $\progvar{v} \in \pvars$ and $n \in \nats$.
We can simulate linear expressions on the right hand side of increments (e.g., $\incrasgn{\progvar{x}}{\progvar{x} + 2 \progvar{y} + 1}$), as well as Boolean combinations of comparisons of variables with constants (e.g., $\progvar{x} = 5 \,\mathtt{||}\, \progvar{y} > 3$) with the loop-free fragment of \redip.
We also write increments $\incrasgn{\progvar{x}}{E}$ as $\assign{\progvar{x}}{\progvar{x} + E}$.
See~\cite{DBLP:conf/cav/ChenKKW22} for a detailed exposure of the language.
\begin{table}[t]
    \caption{Syntax (left column) and the $\pw$-semantics (right column) of \redip programs, where $\progvar{v}, \progvar{y} \in \pvars$, $n \in \nats$ and $D \in \probDistOf{\nats}$ denotes a probability distribution on $\nats$ with $\sem{D} \in \nnreals[[T]]$ in a fresh formal variable $T \notin \vec{X}$.}
    \label{tab:redip}
    \renewcommand{\arraystretch}{1.25}
    \tabcolsep=5pt
    \centering
    \begin{tabular}{ll}
        \toprule
        $P$                               & $\pwapp{P}{\measureGFa}$                                                                                                                                                    \\
        \midrule
        $\assign{v}{n}$                   & $\measureGFa[X_v/1] \cdot X_v^n$                                                                                                                                            \\
        $\decr{v}$                        & $(\measureGFa - \gfFilter{\measureGFa}{\progvar{v} < 1}) \cdot X_v^{-1} + \gfFilter{\measureGFa}{\progvar{v} < 1}$                                                          \\
        $\incrasgn{v}{\iid{D}{y}}$        & $\measureGFa[X_y/X_y \cdot \sem{D}[T/X_v]]$                                                                                                                                 \\
        $\ITE{\progvar{v} < n}{P_1}{P_2}$ & $\pwapp{P_1}{\gfFilter{\measureGFa}{\progvar{v} < n}} + \pwapp{P_2}{\measureGFa - \gfFilter{\measureGFa}{\progvar{v} < n}}$                                                 \\
        $\COMPOSE{P_1}{P_2}$              & $\pwapp{P_2}{\pwapp{P_1}{\measureGFa}}$                                                                                                                                     \\
        $\WHILEDO{\progvar{v} < n}{P_1}$  & $\measureGFb - \gfFilter{\measureGFb}{\progvar{v} < n}$, \quad where $\measureGFb = \lfp \measureGFc.\, \measureGFa + \pwapp{P_1}{\gfFilter{\measureGFc}{\progvar{v} < n}}$ \\
        \bottomrule
    \end{tabular}
    \renewcommand{\arraystretch}{1}
\end{table}
\begin{definition}[\redip]%
    \label{definition:redip}%
    The generating function-based $\pw$-semantics of \emph{rectangular discrete probabilistic programs (\redip)} is defined in~\Cref{tab:redip}, where in the definition of $\pw\sem{P} \colon \nnextreals[[\vec{X}]] \to \nnextreals[[\vec{X}]]$ we use the following operations:%
    \begin{itemize}
        \item Substitution: For $\measureGFa, \measureGFb \in \nnextreals[[\vec{X}]]$ and $\progvar{v} \in \pvars$:
              \[
                  \measureGFa[X_{\progvar{v}}/\measureGFb] ~\coloneqq~ \sum_{\pstatea \in \pstates}{ \Big( \sum_{\substack{\pstatea_1 + \pstatea_2 = \pstatea \\ \pstatea_2(\progvar{v}) = \pstatea(\progvar{v})}}{ \sum_{k \in \nats}{ \gfCoefficient{\measureGFa}{\pstatea_1[\progvar{v} \mapsto k]} \cdot \gfCoefficient{\measureGFb^k}{\pstatea_2} } } \Big) \cdot \vec{X}^\pstatea} ~.
              \]
        \item (Downward) Shifting: For $\measureGFa \in \nnextreals[[\vec{X}]]$ and $\progvar{v} \in \pvars$ such that $\gfCoefficient{\measureGFa}{\pstatea} = 0$ for all $\pstatea \in \pstates$ with $\pstatea(\progvar{v}) = 0$:
              \[
                  \measureGFa \cdot X_v^{-1} ~\coloneqq~ \sum_{\pstatea \in \pstates}{ \gfCoefficient{\measureGFa}{\pstatea[\progvar{v} \mapsto \pstatea(\progvar{v}) + 1]} \cdot \vec{X}^\pstatea} ~.
              \]
        \item Formal Derivative: For $\measureGFa \in \nnextreals[[\vec{X}]]$ and $\progvar{v} \in \pvars$:
              \[
                  \partial_{X_v} \measureGFa ~\coloneqq~ \sum_{\pstatea \in \pstates}{(\pstatea(\progvar{v}) + 1) \cdot \gfCoefficient{\measureGFa}{\pstatea[\progvar{v} \mapsto \pstatea(\progvar{v}) + 1]} \cdot \vec{X}^{\pstatea}}~.
              \]
        \item Restriction: For $\measureGFa \in \nnextreals[[\vec{X}]]$, $\progvar{v} \in \pvars$, $n \in \nats$:
              \begin{align*}
                                   & \phantom{\measureGFa {}-{}} \gfFilter{\measureGFa}{\progvar{v} < n} ~\coloneqq~ \sum_{i = 0}^{n-1} \tfrac{1}{i!}(\partial_{X_v}^i \measureGFa)[X_v/0]\cdot X_v^i
                  \\
                  \text{and}\qquad & \measureGFa - \gfFilter{\measureGFa}{\progvar{v} < n} ~\coloneqq~ \sum_{\pstatea \in \pstates}{\iverson{\gfCoefficient{\measureGFa}{\pstatea} \neq \gfCoefficient{\gfFilter{\measureGFa}{\progvar{v} < n}}{\pstatea}} \gfCoefficient{\measureGFa}{\pstatea} \cdot \vec{X}^\pstatea}~.
              \end{align*}
    \end{itemize}
    These operations are all monotonic, with the partial order on $\nnextreals[[\vec{X}]]$ inherited from $\measuresOf{\pstates}$.
    Hence, the least fixed point in \Cref{tab:redip} is well-defined.
\end{definition}

Intuitively, updates are performed by extracting the parts of the generating function that are affected by the update through substitution operations.
For example, $\assign{v}{n}$ marginalizes w.r.t.\ $X_v$ (thus effectively setting $\progvar{v}$ to $0$ temporarily) and then performs an (upward) shift by $n$ in $X_v$.

The notation used for the additional operations is suggestive of their intended effect; for example, even though there is no subtraction on $\nnextreals[[\vec{X}]]$, $\measureGFa - \gfFilter{\measureGFa}{\progvar{v} < n}$ effectively amounts to what is a subtraction between $\measureGFa$ and $\gfFilter{\measureGFa}{\progvar{v} < n}$.
The substitution operation corresponds to the intuitive substitution $\measureGFa[X_{\progvar{v}}/\measureGFb] = \sum_{\pstatea \in \pstates}{ \gfCoefficient{\measureGFa}{\pstatea} \measureGFb^{\pstatea(\progvar{v})} \cdot \vec{Y}^{\pstatea} }$, where $\vec{Y}$ contains all formal variables of $\vec{X}$ except $X_{\progvar{v}}$.
Substitution is always well-defined, as every monotone sequence converges in our domain of nonnegative extended reals.

Our semantics is a generalization of~\cite{DBLP:conf/cav/ChenKKW22} from probability generating functions to general generating functions.

\begin{proposition}%
    Given a distribution $\mu \in \distOf{\pstates}$ and a \redip-program $P$. Then $\measureGF{\mu} \in \nnreals[[\vec{X}]]$ and $\pwapp{P}{\measureGF{\mu}}$ coincides with the \redip semantics $\semapp{P}{\measureGF{\mu}}$ defined in~\cite[Table~2]{DBLP:conf/cav/ChenKKW22}.
\end{proposition}

\begin{proof}[sketch]%
    Our definition is structurally identical to the definitions in~\cite[Table~2]{DBLP:conf/cav/ChenKKW22}, except for the case of loops.
    Working with (sub)probability generating functions, \cite{DBLP:conf/cav/ChenKKW22} uses operations on $\reals[[\vec{X}]]$.
    (Downward) shifting and formal derivation are defined identically.
    Similarly for restriction, where $\measureGFa - \gfFilter{\measureGFa}{\progvar{v} < n}$ simply uses the subtraction of the ring $\reals[[\vec{X}]]$.
    While the substitution operation $\measureGFa[X_{\progvar{v}}/\measureGFb]$ is not well-defined for all $\measureGFa, \measureGFb \in \reals[[\vec{X}]]$, it has the same definition and is well-defined for (sub-)probability generating functions $\measureGFa, \measureGFb$, which is sufficient to define the program semantics.
    Observe that $\gfFilter{\measureGFa}{\progvar{v} < 1} = \gfSubstitute{\measureGFa}{X_{\progvar{v}}}{0}$.

    For loops, \cite{DBLP:conf/cav/ChenKKW22} uses Kozen's higher-order fixed point construction (see \Cref{ssec:semantics:equivalence_to_kozen}) while we use our occupation measure fixed point construction.
    The equivalence follows just as in~\Cref{sec:semantics}.
\end{proof}

\paragraph{Rational Closed Forms.}
Our goal is to use the algebraic structure of FPSs to obtain finite representations of measures as formal fractions such that we can apply the semantics of \redip-programs while preserving the finite representation.
We achieve this for measures $\mu$ with $\mu(\pstatea) \neq \infty$ for all $\pstatea \in \pstates$ by embedding their generating functions into $\reals[[\vec{X}]]$ and using the set of inverses present therein.%
\begin{definition}[Rational Closed Form]%
    For $\measureGFa \in \nnreals[[\vec{X}]]$, a pair of polynomials $\gfRationalNumerator{\measureGFa}, \gfRationalDenominator{\measureGFa} \in \reals[\vec{X}]$ is a \emph{(rational) closed form} of $\measureGFa$, written $\measureGFa = \frac{\gfRationalNumerator{\measureGFa}}{\gfRationalDenominator{\measureGFa}}$, if $\gfRationalDenominator{\measureGFa}$ is invertible in $\reals[[\vec{X}]]$ and $\measureGFa = \gfRationalNumerator{\measureGFa} \cdot \gfRationalDenominator{\measureGFa}^{-1}$ in $\reals[[\vec{X}]]$.
\end{definition}
For example, \(1 \cdot \gfVar{x}^0 + 2 \cdot \gfVar{x}^1 + 4 \cdot \gfVar{x}^2 + \cdots = \sum_{i \in \nats}{ 2^i \cdot \gfVar{x}^i }\) has the rational closed form \(\frac{1}{1 - 2\gfVar{x}}\).
Excluding $\infty$ as a coefficient gives closed forms in the ring $\reals[[\vec{X}]]$.
As developed in~\cite{DBLP:conf/cav/ChenKKW22}, this allows to define the operations downward shifting, formal derivation and restriction from~\Cref{definition:redip} on the closed forms.
The only exception is substitution, as \eg, \((\sum_{i \in \nats}{ 2^i \cdot \gfVar{x}^i })[\gfVar{x}/1] = \infty \neq -1 = \frac{1}{1 - 2\gfVar{x}}[\gfVar{x}/1]\).
Substitution is only well-defined on closed forms, if the substitute has no constant term: if $\measureGFa = \frac{\gfRationalNumerator{\measureGFa}}{\gfRationalDenominator{\measureGFa}}$, $\measureGFb = \frac{\gfRationalNumerator{\measureGFb}}{\gfRationalDenominator{\measureGFb}}$ with $\gfCoefficient{\measureGFb}{\vec{0}} = 0$, then $\frac{\gfRationalNumerator{\measureGFa}[X_{\progvar{v}}/\frac{\gfRationalNumerator{\measureGFb}}{\gfRationalDenominator{\measureGFb}}]}{\gfRationalDenominator{\measureGFa}[X_{\progvar{v}}/\frac{\gfRationalNumerator{\measureGFb}}{\gfRationalDenominator{\measureGFb}}]}$ is a closed form of $\measureGFa[X_{\progvar{v}}/\measureGFb]$.

With this restriction of operations on closed forms, $\pw$ can be computed for loop-free \redip programs on closed forms, with two restrictions:
for constant assignments $\assign{v}{n}$ it cannot be generally well-defined, \eg, applying $\assign{\progvar{x}}{0}$ to $\frac{1}{1 - \gfVar{x}}$ yields $\infty$ which is not in rational closed form.
Secondly, i.i.d.\ sampling increments $\incrasgn{v}{\iid{D}{y}}$ must be restricted to distributions $D$ with a rational closed form.
We call this fragment \emph{closed \redip} (\closedRedip):%
\begin{definition}[Closed \redip]%
    A $\closedRedip$ program is a loop-free \redip program without constant assignments $\assign{v}{n}$ and in which every used probability distribution $D$ has a closed form $\sem{D} = \frac{\gfRationalNumerator{\sem{D}}}{\gfRationalDenominator{\sem{D}}}$.
\end{definition}

\begin{corollary}[\cite{DBLP:conf/cav/ChenKKW22}]%
    Let $\frac{\gfRationalNumerator{\measureGFa}}{\gfRationalDenominator{\measureGFa}}$ be a closed form of generating function $\measureGFa \in \nnreals[[\vec{X}]]$, and $P$ a \closedRedip program.
    Then $\pwapp{P}{\measureGFa}$ has a closed form, obtainable by applying the operations of~\Cref{tab:redip} on the closed form.
\end{corollary}

\begin{remark}[Expressivity of our Language]%
    Even though constant assignments $\assign{\progvar{v}}{n}$ are not allowed in $\closedRedip$, they could be simulated by means of a loop $\ploop{\progvar{v} > 0}{\decr{\progvar{v}}} \fatsemi \incrasgn{\progvar{x}}{n}$.
    However, the resulting loop may not always admit a rational invariant.
    Comparing a variable with another variable cannot be simulated in loop-free \redip~\cite{DBLP:conf/cav/ChenKKW22}, and hence also not in \closedRedip.
    We conjecture that polynomial expressions in guards and assignments cannot be simulated in general either.
    To our knowledge, the exact class of distribution transformers expressible in loop-free \redip and \closedRedip remains unknown.
    We also note that $\redip$ \emph{with loops} is Turing-complete despite its syntactic restrictions, as it can simulate two-counter machines.
\end{remark}

\subsection{Template-based Invariant Synthesis}
Consider a \redip loop $P = \ploop{\progvar{v} < n}{P_0}$, with \closedRedip body $P_0$ and initial measure $\measureGFa$ in closed form.
We use the generating function $\pw$-semantics on rational closed forms to synthesize occupation invariants for the loop $P$ starting from $\measureGFa$.
We will rely on (user-)provided templates to reduce the invariant synthesis problem to a parameter synthesis problem.

Let $\invariantTemplateVars$ denote a finite set of symbolic template parameters which are distinct from the formal variables $\vec{X} = (X_{\progvar{v}})_{\progvar{v} \in \pvars}$.
Our templates are rational closed forms that contain polynomial expressions $\reals[\invariantTemplateVars]$ instead of only values $\reals$ as coefficients.
For $\tau \colon \invariantTemplateVars \to \reals$, $I[\tau]$ denotes the instantiation of $I$ where all occurrences of $t \in \invariantTemplateVars$ have been replaced by $\tau(t)$, provided the result is well-defined.
Note that not every parameter valuation yields a valid instantiation, \eg, while instantiating the template $I = p \cdot \gfVar{v} + (1-p) \cdot \gfVar{v}^2$ with $\tau = [p \mapsto -1]$ yields a valid FPS $I[\tau] = -1 \gfVar{v} + 2 \gfVar{v}^2$ in $\reals[[\vec{X}]]$, this does not represent a nonnegative measure.

Synthesizing an invariant using a template $I$ becomes a two-step process.
First, we apply the loop's characteristic functional $\redipGfCharfun{\measureGFa}{P}$ once on $I$ and get a closed form for $\redipGfCharfun{\measureGFa}{P}(I) = \measureGFa + \pwapp{P_0}{\gfFilter{I}{\progvar{v} < n}}$.
By interpreting the parameters $\invariantTemplateVars$ as fresh program variables, we can directly use the definitions from the previous section.
To ensure that \(I\) represents a valid invariant, we require it to satisfy the fixed-point equation \(\redipGfCharfun{\measureGFa}{P}(I) = I\).
Through cross-multiplication of the denominators one obtains an equation between two polynomials.
Comparing coefficients results in an equation system imposing constraints on the symbolic parameters \(\invariantTemplateVars\), which can be solved using off-the-shelf computer algebra tools (\eg, rational function arithmetic, Gröbner bases).
Second, we verify that an instantiation fulfilling these constraints, results in a generating function of a nonnegative measure.

\begin{theorem}\label{theorem:template:non-negative-instantiation-is-invariant}
    Given $\mu \in \measuresOf{\pstates}$ with $\measureGF{\mu}$ with a rational closed form, $I$ a rational template, and \(\tau \colon \invariantTemplateVars \to \reals\), such that \(\redipGfCharfun{\measureGF{\mu}}{P}(I[\tau]) = I[\tau] \).
    If $I[\tau] \in \nnextreals[[\vec{X}]]$, then $\mu_{I[\tau]}$ is an occupation invariant of $\ploop{\progvar{v} < n}{P_0}$ from $\mu$.
\end{theorem}

\begin{example}
    We continue the geometric loop example from~\Cref{example:loop-invariants-mass-preserving-invariant} from \Cref{prog:geometric_loop}.
    The initial probability distribution \(g = \iverson{x = 1 \land c = 0}\) entering the loop, has generating function $\measureGF{g} = 1 \cdot X^1 C^0 = X$ which is also a closed form.
    Consider the following rational template $I$ with \(\invariantTemplateVars = \{a_0, a_1, a_2, a_3, b_0, b_1, b_2\}\):
    \[
        I \eeqq \frac{a_0 \pplus a_1 \cdot X \pplus a_2 \cdot C \pplus a_3 \cdot X^2}{b_0 \pplus b_1 \cdot X \pplus b_2 \cdot C}~.
    \]
    Applying \(\redipGfCharfun{\measureGF{g}}{P}\) of loop \(P = \ploop{x = 1}{\pchoice{\passign{x}{0}}{\tfrac 1 2}{\passign{c}{c+1}}}\) to \(I\) yields:
    \[
        \redipGfCharfun{\measureGF{g}}{P}(I) \eeq \gfVar{x} + \frac{1}{2} \cdot (1 + \gfVar{x} \gfVar{c}) \cdot \frac{(a_1 b_0 - a_0 b_1) + (a_1 b_2 - a_2 b_1) \cdot \gfVar{c}}{(b_0 + b_2 \cdot \gfVar{c})^2}~.
    \]
    By solving the equation \(\redipGfCharfun{\measureGF{g}}{P}(I) = I\) symbolically over the rational function template, we obtain the assignment
    \(
    \tau = [a_0 \mapsto 1, a_1 \mapsto 2, a_2 \mapsto 0, a_3 \mapsto 0, b_0 \mapsto 2, b_1 \mapsto 0, b_2 \mapsto -1]
    \)
    fulfilling \(\redipGfCharfun{\measureGF{g}}{P}(I[\tau]) = I[\tau]\).
    As
    \[
        I[\tau] \eeq \frac{1 \pplus 2 \cdot X}{2 \mminus C} \eeq (1 + 2 \cdot X) \sum_{k \in \nats}{\frac{1}{2^{k+1}} \cdot C^k} \eeq \sum_{k \in \nats}{\frac{1}{2^{k+1}} \cdot C^k} + \sum_{k \in \nats}{\frac{1}{2^{k}} \cdot X C^k}
    \]
    has only nonnegative coefficients, we conclude that $I[\tau]$ represents an occupation invariant by~\Cref{theorem:template:non-negative-instantiation-is-invariant}.
    Indeed, this is the loop's occupation measure that we obtained in~\Cref{example:loop-invariants-mass-preserving-invariant}:
    \[
        \mu_{I[\tau]} \eeqq \iverson{x = 0} \cdot 2^{-(c+1)} \pplus \iverson{x=1}\cdot 2^{-c}~.
    \]
    With~\Cref{theorem:invariants:mass_preserving_is_exact} for $\mu_{I[\tau]}$, we conclude
    \[
        \pwapp{C}{g} \eeq \iverson{x \neq 1} \cdot \mu_{I[\tau]} \eeq \iverson{x = 0} \cdot 2^{-(c+1)}~.
    \]
\end{example}

This approach appears particularly effective when the program structure leads to regular recurrence patterns, which can be translated into algebraic equalities over generating functions.
Embedding such structure into a template allows analysis of loops on an abstract, intuitive level and leverages automated solvers to delegate the tedious computation of exact coefficients.
In more complex cases, additional approximation techniques or domain-specific templates may be required, as demonstrated by the following example.\\
\begin{wrapfigure}[15]{r}{0.5\textwidth}%
    \vspace{-25pt}%
    \begin{minipage}{0.5\textwidth}%
        \input{figures/fdr_n}
    \end{minipage}
\end{wrapfigure}%
\begin{example}%
    \label{example:invariant_templates:fdr}%
    Consider the \emph{Fast Dice Roller}~\cite{DBLP:journals/corr/abs-1304-1916} algorithm $\mathrm{FDR}_N$~(\Cref{prog:fast_dice_roller_n}) for some fixed integer constant $N > 0$.
    Let $P_N$ denote the loop body.
    The loop is initially reached by $\measureGFa_0 \coloneqq 1 \cdot \gfVar{v}^1 \gfVar{c}^0 \gfVar{f}^0$.
    For $N > 4$, iterating the loop body yields
    \begin{align*}
         & \measureGFa_1 \coloneqq \pwapp{P_N}{\gfFilter{\measureGFa_0}{\progvar{f} = 0}} = \tfrac{1}{2}  \gfVar{v}^2 (\gfVar{c}^0 {+} \gfVar{c}^1), \\
         & \measureGFa_2 \coloneqq \pwapp{P_N}{\gfFilter{\measureGFa_1}{\progvar{f} = 0}}                                                            \\
         & \qquad\qquad\quad = \tfrac{1}{4}  \gfVar{v}^4 (\gfVar{c}^0 + \gfVar{c}^1 + \gfVar{c}^2 + \gfVar{c}^3) .
    \end{align*}
    $c$ is always uniformly distributed between $0$ and $v-1$, as long as $v < N$~\cite{DBLP:journals/corr/abs-1304-1916}.
    Hence, we expect an invariant to contain terms of the form $a_{N,i} \cdot \gfVar{v}^i \sum_{k<i}{\gfVar{c}^k} = a_{N,i} \cdot \gfVar{v}^i \frac{1 - \gfVar{c}^i}{1 - \gfVar{c}}$ for some $a_{N,i} \in \nnreals, i < N$.
    In an iteration in which $v$ is increased to at least $N$, the distribution is split into two parts:
    a terminating part with $c < N$ and $f$ set to $1$, and a continuing part where $c \geq N$ and $c, v$ have been reduced by $N$.
    For example, for $N=6$ we have
    \begin{align*}
        \measureGFa_3 ~\coloneqq~ \pwapp{P_N}{\gfFilter{\measureGFa_2}{\progvar{f} = 0}} \eeq \frac{1}{8} \cdot \gfVar{f} \gfVar{v}^8 \frac{1 - \gfVar{c}^6}{1 - \gfVar{c}} \pplus \frac{1}{8} \cdot \gfVar{v}^2 \frac{1 - \gfVar{c}^2}{1 - \gfVar{c}}
        ~.
    \end{align*}
    We thus expect an invariant to contain the terms $a_{N,i} \cdot \gfVar{f}\, \gfVar{v}^i \frac{1 - \gfVar{c}^N}{1 - \gfVar{c}}$ for some $a_{N,i} \in \nnreals, i \geq N$.
    We make the following additional observations:
    As $v$ is only doubled, or reduced by $N$ if $v \geq N$, $v$ must be congruent to a power of $2$ modulo $N$, and $v < 2 N$ always holds between loop iterations.
    We thus propose the following invariant template:
    \[
        I_N \eeq \sum_{\substack{1 \leq i < N \,\land\, \exists k \colon i \equiv_N 2^k}}{\!\!\!\!
            a_{N,i} \cdot \gfVar{v}^i \frac{1 - \gfVar{c}^i}{1 - \gfVar{c}}
        }
        \pplus
        \sum_{\substack{N \leq i < 2 N \,\land\, \exists k \colon i \equiv_N 2^k}}{
        \!\!\!\!a_{N,i} \cdot \gfVar{f}\, \gfVar{v}^i \frac{1 - \gfVar{c}^N}{1 - \gfVar{c}}
        }~.
    \]
    As the detailed calculations in~\iftoggle{arxiv}{\Cref{apx:fdr}}{\cite[Appendix B]{arxiv}} show, $I_N = \measureGFa_0 + \pwapp{P_N}{\gfFilter{I_N}{\progvar{f} = 0}}$ has a solution $\tau_N$ which yields a nonnegative, finite instantiation with
    \[
        \sum_{\substack{N \leq i < 2 N \,\land\, \exists k \colon i \equiv_N 2^k}}{\!\!\!\!\tau_N(a_{N,i})} \eeq \frac{1}{N}
        ~.
    \]
    By~\Cref{theorem:template:non-negative-instantiation-is-invariant}, $\mu_{I_N[\tau_N]}$ is an occupation invariant for $\mathrm{FDR}_N$, which we use to bound the program's semantics.
    We obtain the expected uniform distribution of $c$ between $0$ and $N-1$ as an upper bound, after marginalizing out $v$ and $f$:
    \begin{align*}
        \pwapp{\mathrm{FDR}_N}{\measureGFa_0}[\gfVar{v}, \gfVar{f}/1]
         & \lleq (\gfFilter{I_N[\tau_N]}{f > 0})[\gfVar{v}, \gfVar{f}/1]                    \\
         & \eeq \sum_{\substack{N \leq i < 2 N \,\land\, \exists k \colon i \equiv_N 2^k}}{
            \!\!\!\!\tau_N(a_{N,i}) \cdot \frac{1 - \gfVar{c}^N}{1 - \gfVar{c}}
        }
        \eeq \frac{1}{N} \cdot \frac{1 - \gfVar{c}^N}{1 - \gfVar{c}}~.
    \end{align*}
    Since $\abs{\mu_{I_N[\tau_N]}} = 1$ is finite, we conclude that the loop is PAST from $\mu_{\measureGFa_0}$ and thus $\abs{\pwapp{\mathrm{FDR}_N}{\measureGFa_0}} = \abs{\measureGFa_0} = 1$.
    Hence, by~\Cref{theorem:invariants:mass_preserving_is_exact}, the uniform distribution upper bound is tight, and we conclude $\pwapp{\mathrm{FDR}_N}{\measureGFa_0}[\gfVar{v}, \gfVar{f}/1] = \frac{1}{N} \cdot \frac{1 - \gfVar{c}^N}{1 - \gfVar{c}}$.
\end{example}

\paragraph{Nested Loops.}
To handle nested loops, an invariant has to be found for each nested loops simultaneously, as the measure reaching an inner loop is directly influenced by the measure entering an outer loop.
Verification is done starting with the innermost loop, with the resulting measure used to prove the invariant for the next outer loop.

Alternatively, as a pre-processing step, we could rewrite an arbitrary program (automatically) as a loop with a loop-free body (\eg, \cite{DBLP:conf/tase/RabehajaS09}), assuming such a transformation could be carried out within \redip.

\section{Implementation and Evaluation}
\label{sec:implementation}
To validate the practical feasibility of our invariant synthesis method for probabilistic loops, we developed a proof-of-concept implementation.
The tool realizes the approach presented in \Cref{sec:synthesis} and enables the validation and, in many cases, the automated synthesis of inductive invariants in rational closed-form.

\subsection{Synthesis Overview}
Our approach for synthesizing rational function invariants is based on symbolic fixed-point reasoning over probabilistic loops.
Given a probabilistic \pwhile-loop and an initial distribution in rational closed form, we synthesize candidates via a \emph{rational function template heuristic}.
Each candidate is interpreted as a symbolic distribution and subjected to a forward semantics step: we compute the effect of one loop iteration under the assumption that the candidate holds inductively.
This yields a rational equation between the candidate and the semantic result of one loop iteration, which we \emph{solve symbolically} using either an SMT solver (e.g., Z3) or the computer algebra system \sympy, depending on configuration.

Solutions to the resulting system correspond to instantiations of the symbolic coefficients in the rational candidate.
We filter out trivial or invalid solutions, e.g., those that make the denominator identically zero.
The remaining candidates are checked using \emph{heuristic positivity tests} to determine whether the resulting generating functions represent valid (i.e., nonnegative) measures.
If such a candidate is found, we conclude that it defines a true inductive invariant for the probabilistic loop.
If no candidate can be validated, the process fails, signaling that either the heuristic space was insufficient or no rational invariant of the considered form exists.

\subsection{Implementation Details}
The tool is implemented in Python~\cite{python} and supports two symbolic computation backends: the \sympy~\cite{sympy} library and a custom set of Python bindings to the C++ library \ginac~\cite{ginac}.
The latter offers improved performance for rational function manipulations.
Backend selection is configurable at runtime.

Initial distributions are specified in terms of generating functions, typically in rational closed form.
These generating functions are represented internally as symbolic expressions and manipulated according to the operations required by the fixed-point characterization of inductive invariants.
Simplification and algebraic reasoning tasks are delegated to the selected backend.

Beyond the language presented in~\Cref{sec:synthesis} our implementation provides limited support for guards with modulo expressions on univariate generating functions.
As in~\cite{ictac25}, we support expressions of the form $\progvar{v} = c \bmod d$, where $\progvar{v} \in \pvars$ and $c,d \in \nats$.
This is implemented via a standard method based on Hadamard products; see, e.g.,~\cite[Section 2.4]{lando2003}.

\paragraph{Invariant Template Generation.}
For our benchmark, we explore candidate invariants by systematically generating multivariate rational expressions of increasing structural complexity.
Specifically, the heuristic organizes the search by first fixing a total degree bound for the denominator and then enumerating all possible numerator polynomials whose degree does not exceed that of the denominator.
For each such degree configuration, it constructs symbolic templates for the numerator and denominator, introducing fresh coefficient parameters for each admissible monomial.
The resulting rational expressions are interpreted as generating functions and wrapped into symbolic distribution objects.
This breadth-first enumeration strategy ensures that simpler candidates are considered before more complex ones, facilitating an efficient and structured search through the space of rational inductive invariants.
Crucially, this enumeration is exhaustive in theory: if a rational function invariant exists, then at some finite degree bound the heuristic will produce a corresponding template that can be instantiated to that invariant.

\paragraph{Positivity Test Heuristics.}
Determining whether an arbitrary rational function has a nonnegative FPS expansion is an open problem and believed to be decidable~\cite{DBLP:conf/soda/OuaknineW14}.
However, to ensure that a candidate rational function indeed represents a valid generating function---that is, a power series with nonnegative coefficients---we apply sufficient heuristics to establish positivity.
For the benchmarks, we use heuristics based on structural properties of the rational expression.
These decompose the candidate into numerator and denominator.
For the numerator polynomial it is verified that it contains only nonnegative coefficients.
For the denominator a simple sign condition is checked: the constant term must be positive, and all other coefficients must be non-positive.
If either the function or its negative meets this pattern, the heuristic concludes non-negativity.
This test is sound, as it captures the shape of (generalized) geometric distributions and similar rational FPS with guaranteed nonnegative expansions.

\subsection{Benchmark Configuration}
The evaluation is based on a suite of small to mid-sized benchmark loops containing probabilistic updates, branching, and control-flow structures of varying complexity.
Our experiments aim to address the following research questions:%
\begin{itemize}
	\item \textbf{Feasibility}: Can non-trivial rational invariants be synthesized in practice?
	\item \textbf{Performance}: How fast are invariants found for the benchmark programs?
	\item \textbf{Coverage}: For which loop patterns does the heuristic succeed or fail?
\end{itemize}

All benchmarks were executed on an Apple MacBook Pro equipped with an Apple M3 Pro processor and 18GB of RAM, running macOS Sequoia~15.4.1.
We report only the execution time of the core synthesis phase, excluding preprocessing steps such as input parsing or output formatting.

\subsection{Benchmark Results}
\Cref{tab:benchmarks} summarizes the results obtained from running the synthesis tool on each benchmark.
Each entry records the fastest time taken to complete the synthesis, whether a valid invariant was computed automatically (\texttt{A}) or by a user specified template (\texttt{U}), and in the latter case the inferred parameter values.

\begin{table}[t]
	\centering
	\caption{Verification and Synthesis benchmark results for \prodigy using the fastest backend (\ginac or \sympy) for the given benchmark.}
	\label{tab:benchmarks}
	\begin{tabular}{lrcl}
		\toprule
		\textbf{Benchmark}             & \textbf{Time (s)} & \textbf{Auto/User} & \textbf{Additional Comments}                                  \\
		\midrule
		\texttt{subdist\_enter}        & \texttt{0.568013} & \texttt{A}         &                                                               \\
		\texttt{cond\_and}             & \texttt{0.017394} & \texttt{A}         &                                                               \\
		\texttt{faulty\_decrement}     & \texttt{0.120409} & \texttt{A}         &                                                               \\
		\texttt{geometric}             & \texttt{0.026373} & \texttt{A}         &                                                               \\
		\texttt{geometric\_counter}    & \texttt{0.124910} & \texttt{A}         &                                                               \\
		\texttt{modulo\_geometric}     & \texttt{0.035103} & \texttt{U}         & \makecell[cl]{$a=4f, b=2f, c=f, d=8f,$                        \\ $e=6f$} \\
		\texttt{thirds\_geometric}     & \texttt{6.681842} & \texttt{A}         & Used Z3 to solve                                              \\
		\texttt{random\_walk}          & \texttt{0.118271} & \texttt{A}         &                                                               \\
		\texttt{random\_walk\_counter} & \texttt{0.106519} & \texttt{U}         & $a=1, b=1, d= -1$                                             \\
		\texttt{fast\_dice\_roller}    & \texttt{0.155679} & \texttt{U}         & \makecell[cl]{$w_0=\tfrac 1 3, w_1=0, w_2=\tfrac 2 3, w_3=1,$ \\$w_4=\tfrac 1 6, w_5=0, w_6=0, w_m=0$} \\
		\texttt{nontermination}        & \texttt{Failed}   & \texttt{A}         & Actual invariant is $\infty\cdot X + \tfrac 1 2 \cdot X^2$    \\
		\texttt{sequential\_loops}     & \texttt{Failed}   & \texttt{A}         & Cannot determine pos. of $\tfrac{2C}{C^2 - 3C + 2} + 1$       \\
		\bottomrule
	\end{tabular}
\end{table}

\subsection{Discussion}
The results demonstrate that our symbolic synthesis framework is capable of handling a variety of probabilistic loops and is able to synthesize rational inductive invariants efficiently in most cases.
Synthesis times are typically below one second.
Benchmarks marked as \texttt{U} synthesized invariants given a user specified template solving for specific parameter values, which the tool recovered automatically from the symbolic equation system.
We also observed that the choice of backend (\ginac or \sympy) affects performance but not correctness, with \ginac generally offering faster simplification for rational expressions.

Failures were rare and typically attributable to the limitations of the heuristic search space or positivity tests.
Future work will explore expanding the space of templates and developing more refined positivity checks.

\section{Related Work}
\label{sec:relwork}
The main theoretical technique of this paper --- reasoning about probabilistic loops with occupation invariants --- is inspired by~\cite{DBLP:journals/siamcomp/SharirPH84}.
However, \cite{DBLP:journals/siamcomp/SharirPH84} views probabilistic programs semantically as (countably infinite) Markov chains.
As a consequence, their definitions and case studies all operate on the level of explicit infinite stochastic matrices.
We, on the other hand, define and reason about occupation invariants \emph{directly on the program text}.
The latter is essential for our automated verification and synthesis based on generating functions.

In the remainder of the section, we summarize further related work on (1) generating functions, (2) inference of posterior distributions, (3) synthesis of quantitative invariants, and (4) probabilistic model checking.

\paragraph{Generating Functions in Probabilistic Program Analysis and Semantics.}
The use of generating functions (GFs) in connection with probabilistic programs is a relatively recent trend, whose foundations originate from~\cite{DBLP:conf/lopstr/KlinkenbergBKKM20}.
Subsequent work implemented and applied this framework to program equivalence~\cite{DBLP:conf/cav/ChenKKW22} and conditioning~\cite{DBLP:journals/pacmpl/KlinkenbergBCHK24}; see~\cite{KlinkenbergPHD} for a comprehensive overview.
Independently, \cite{DBLP:conf/nips/ZaiserMO23} amended the method with a basic support for continuous distributions and explored automatic differentiation techniques for exact posterior inference.
Recently, \cite{10.1145/3747534} introduced a functional language (without unbounded loops/recursion) compiling to GF expressions.
Further recent work applies GFs in combination with Banach’s fixed-point theorem to derive exact runtime distributions~\cite{DBLP:conf/ictcs/CollodiBG25}, and encodes GFs via weighted automata~\cite{ictac25}, which appears compatible with our approach as well.
However, none of the mentioned papers provides an approach for automatically deriving loop posteriors as closed-form GFs.

\paragraph{Posterior Inference: Approximate vs Guaranteed vs Exact.}
Algorithmic posterior inference for probabilistic programs has been extensively studied across communities, yielding mature tools such as \toolfont{STAN}~\cite{carpenter2017stan}, \toolfont{WebPPL}~\cite{dippl}, and \toolfont{Pyro}~\cite{DBLP:journals/jmlr/BinghamCJOPKSSH19}. These systems typically rely on sampling-based \emph{approximate} inference without formal guarantees.
As shown in~\cite{DBLP:conf/pldi/BeutnerOZ22}, some of these methods can yield highly inaccurate results in the presence of unbounded recursion or loops.
In the remainder of the discussion, we restrict attention to verification-oriented approaches providing \emph{guaranteed inference}, i.e., mathematically certified results.

Guaranteed inference methods either establish sound \emph{bounds} or compute truly \emph{exact} results.
Within the former class, \cite{DBLP:conf/pldi/BeutnerOZ22,DBLP:journals/pacmpl/WangYFLO24} prove upper and lower bounds over a finite interval partition of the posterior domain.
These approaches support unbounded loops and recursion, continuous distributions, and general soft conditioning (scoring).
For discrete programs, \cite{DBLP:journals/pacmpl/ZaiserMO25} proposes two bounding techniques: one based on loop unrolling with residual masses, and another exploiting so-called \emph{contraction invariants}, a special case of occupation invariants, and eventually geometric distributions.
Exact methods include \toolfont{PSI}~\cite{DBLP:conf/cav/GehrMV16}, which symbolically represents posterior densities, and \toolfont{Dice}~\cite{DBLP:journals/pacmpl/HoltzenBM20}, which employs weighted model counting for scalable inference in finite-state programs.
Both are restricted to statically bounded loops.
\toolfont{Dice} was later extended with support for unbounded iteration~\cite{DBLP:conf/fscd/Torres-RuizP0Z24}.
It seems promising to investigate   whether occupation measures could be applied in this setting as well.

\paragraph{Invariant Synthesis for Probabilistic Programs.}
To our knowledge, this paper is the first to synthesize general occupation invariants for probabilistic loops.
\citeauthor{DBLP:journals/pacmpl/ZaiserMO25}~\cite{DBLP:journals/pacmpl/ZaiserMO25} recently introduced and synthesized \emph{contraction invariants}, which can be viewed as a special case of occupation invariants.
Importantly, in contrast to the latter, contraction invariants are \emph{incomplete}: they cannot certify all valid bounds on the posterior measure.
See \iftoggle{arxiv}{\Cref{apx:zaiser}}{\cite[Appendix D]{arxiv}} for a concrete counterexample.

The synthesis of other types of quantitative invariants has been widely explored.
We briefly review some representative papers.
Early work includes~\cite{DBLP:conf/sas/KatoenMMM10}, which synthesizes linear \emph{expectation invariants} in the weakest-preexpectation calculus~\cite{DBLP:series/mcs/McIverM05} via constraint solving.
Such invariants are random variables over program states whose expectation does not decrease across loop iterations; they were later termed \emph{sub-invariants}~\cite{DBLP:phd/dnb/Kaminski19}.
Polynomial sub-invariants can be synthesized via templates combined with Positivstellensatz reasoning and semidefinite programming~\cite{DBLP:conf/atva/FengZJZX17}.
The dual notion of \emph{super-invariants} has been synthesized using a CEGIS-style approach~\cite{DBLP:conf/tacas/BatzCJKKM23}.
Expectation invariants are closely related to probabilistic \emph{martingales} (see, e.g.,~\cite{DBLP:journals/pacmpl/HarkKGK20}).
The synthesis of explicit martingales and their variants, initiated in~\cite{DBLP:conf/cav/ChakarovS13}, has also been studied extensively, often using templates and constraint solving; see~\cite{DBLP:journals/toplas/TakisakaOUH21} for an overview.
A widely used variant are \emph{ranking super-martingales}~\cite{DBLP:conf/cav/ChakarovS13}, a generalization of ranking functions for probabilistic termination, with several approaches to automatic synthesis~\cite{DBLP:conf/popl/FioritiH15,DBLP:journals/pacmpl/AgrawalC018}. Also in the context of termination proofs,~\cite{DBLP:conf/popl/ChatterjeeNZ17} introduces \emph{stochastic invariants}, which are predicates that hold indefinitely with at least a desired probability.
\cite{DBLP:conf/atva/BartocciKS19} initiated the automatic synthesis of \emph{moment-based invariants}, predicates over expectations and higher moments of program variables, via recurrence solving.
Yet another, orthogonal notion are \emph{distributional invariants}~\cite{DBLP:journals/ijfcs/HartogV02,DBLP:conf/esop/BartheEGGHS18}, though we are not aware of any work on their automatic synthesis.

\paragraph{Probabilistic Model Checking.}
The model checker \textsc{Storm}~\cite{DBLP:journals/sttt/HenselJKQV22} can compute expected visiting times for \emph{finite} Markov chains~\cite{DBLP:journals/jar/MertensKQW25}, which --- in principle --- could yield occupation invariants for loops with finite state spaces.
However, the underlying the numeric techniques do not readily extend to infinite state spaces.

\section{Conclusion}
\label{sec:conclusion}
We defined the measure transformer semantics of probabilistic loops in terms of the occupation measure --- the expected number of visits to each state at the loop header --- induced by a given initial distribution.
Building on this foundation, we introduced the notion of \emph{occupation superinvariants}, which enable bounding a loop's output distribution and proving PAST.
Furthermore, we demonstrated that occupation invariants can be synthesized automatically through a template-based enumeration of generating functions in rational closed form.

Future work includes incorporating additional template enumerations and heuristics for establishing the positivity of generating functions in closed form (e.g., for eventually geometric distributions of~\cite{DBLP:journals/pacmpl/ZaiserMO25}), and extending the approach to conditioning, as in~\cite{DBLP:journals/pacmpl/KlinkenbergBCHK24}.
Another promising direction is to extend the closed-form approach to handle $\infty$-coefficients, thereby broadening the class of supported programs.
However, since the extended domain $\reals^{\infty}[[\vec{X}]]$ is not a ring, it lacks important algebraic properties required to define and manipulate closed forms.
It is unclear, how our invariant synthesis could be generalized to different data types, as it uses generating functions on the discrete non-negative integer domain.
We conjecture that most of the theory can also be transferred to the continous setting, likely by switching from (probability) generating functions to moment generating functions or characteristic functions.
Finally, it would be interesting to investigate whether occupation measures can be leveraged to compute \emph{stationary distributions} of nonterminating probabilistic loops, analogous to the approach of~\cite{DBLP:journals/jar/MertensKQW25} for finite discrete-time Markov chains.

\section*{Acknowledgments}
This work was partially supported by the ERC POC Grant \textit{VERIPROB} (grant no.\ 101158076), by the DFG RTG 2236 \textit{UnRAVeL}, and by the EU's Horizon 2020 research and innovation program under the Marie Sk\l{}odowska-Curie grant agreement \textit{MISSION} (grant no.\ 101008233).

\section*{Data Availability Statement}
A software artifact containing the implementation and benchmark programs described in \Cref{sec:implementation} is publicly available on Zenodo~\cite{haase_2026_18184810}.

\printbibliography

\iftoggle{arxiv}{%
  \newpage
  \appendix
  \allowdisplaybreaks %
  \section{Proofs}
\label{apx:proofs}

To prove the correspondence between $\lfp \charfun$ and the occupation measure of the big-step Markov chain, we will use Kleene's fixed point theorem.
This requires our semantics to be order-continuous.
\begin{proposition}\label{proposition:pw:scott_continuous}
	A directed set is a non-empty set $D \subseteq \powerset{\measuresOf{\pstates}}$, if for all $\mu_1, \mu_2 \in D$, there exists $\mu_3 \in D$ with $\mu_1 \sqcup \mu_2 \leq \mu_3$.
	For $C \in \pgcl$: $\pw\sem{C}$ is order-continuous, \ie, for every directed set $D \subseteq \powerset{\measuresOf{\pstates}}$: $\sup \pwapp{C}{D} = \pwapp{C}{\sup D}$.
\end{proposition}
\begin{proof}
	Follows from the equivalence of $\pw$ to the order-continuous Kozen semantics $\kozensem$~\cite{DBLP:journals/jcss/Kozen81}.
\end{proof}

\begin{proof}[Proof of~\Cref{theorem:pw:relation_to_om}]
	Let $C = \ploop{\phi}{C_0}$ with $C_0$ loop-free.
	We show by induction on $k \in \nats$: for all $g \in \measuresOf{\pstates}$, $\pstatea \in \pstates$:
	\begin{align}
		(\pwapp{C_0}{\iverson{\phi} \cdot -})^k(g)(\pstatea) = \expectedValueOf[\mcInducedMeasure{g}]{\lambda~ \pi.~ \iverson{ \pi_k = \pstatea}}. \label{theorem:lfp_is_om:proof:inductive_equality}
	\end{align}
	For $k = 0$:
	\begin{align*}
		(\pwapp{C_0}{\iverson{\phi} \cdot -})^0(g)(\pstatea) & = g(\pstatea) = \mcInducedMeasure{g}(\cylinderSet{\pstatea})
		= \expectedValueOf[\mcInducedMeasure{g}]{\indicatorFct{\cylinderSet{s}}}
		= \expectedValueOf[\mcInducedMeasure{g}]{\lambda~ \pi.~ \iverson{ \pi_0 = \pstatea}}.
	\end{align*}
	For $k \to k+1$:
	First, we have for $\pstatea, \pstatea' \in \pstates$:
	\begin{align*}
		(\iverson{\phi} \cdot g)(\pstatea') \cdot \pwapp{C_0}{\dirac{\pstatea'}}(\pstatea)
		 & = \begin{cases}
			     g(\pstatea') \cdot \pwapp{C_0}{\dirac{\pstatea'}}(\pstatea), & \text{if } \pstatea' \models \phi, \\
			     0, \text{else}
		     \end{cases}     \\
		 & = \begin{cases}
			     g(\pstatea') \cdot \pwapp{C_0}{\dirac{\pstatea'}}(\pstatea), & \text{if } \pstatea' \models \phi, \\
			     g(\pstatea') \cdot \dirac{\downarrow}(\pstatea), \text{else}
		     \end{cases} \tag{$\pstatea \in \pstatea$}     \\
		 & = g(\pstatea') \cdot \kernelEvalAt{\mcTransitionProbs}{\pstatea'}(\pstatea). \tag{Def. $\mcTransitionProbs$}
	\end{align*}
	Hence,
	\begin{align*}
		\pwapp{C_0}{\iverson{\phi} \cdot g}(\pstatea)
		 & = \sum_{\pstatea' \in \pstates}{(\iverson{\phi} \cdot g)(\pstatea') \cdot \pwapp{C_0}{\dirac{\pstatea'}}(\pstatea)} \tag{$\pw$ is linear and continuous} \\
		 & = \sum_{\pstatea' \in \pstates}{g(\pstatea') \cdot \kernelEvalAt{\mcTransitionProbs}{\pstatea'}(\pstatea)}. \tag{eq. above}
	\end{align*}
	Now, to show the inductive step, we have for $\pstatea \in \pstates$:
	{\allowdisplaybreaks
	\begin{align*}
		 & (\pwapp{C_0}{\iverson{\phi} \cdot -})^{k+1}(g)(\pstatea)                                                                                                                                                                                                                                                                                                                                    \\
		 & = (\pwapp{C_0}{\iverson{\phi} \cdot -})^{k}(\pwapp{C_0}{\iverson{\phi} \cdot g})(\pstatea) \tag{unfold}                                                                                                                                                                                                                                                                                     \\
		 & = \expectedValueOf[\mcInducedMeasure{\pwapp{C_0}{\iverson{\phi} \cdot g}}]{\lambda~ \pi.~ \iverson{ \pi_k = \pstatea}} \tag{IH}                                                                                                                                                                                                                                                             \\
		 & = \expectedValueOf[\mcInducedMeasure{\pwapp{C_0}{\iverson{\phi} \cdot g}}]{\lambda~ \pi.~ \sum_{\bar{\pi} \in \pstates^k}{\iverson{ \pi_0 = \bar{\pi}_0 \land \dots \land \pi_{k-1} = \bar{\pi}_{k-1} \land \pi_k = \pstatea}}} \tag{split based on prefix}                                                                                                                                 \\
		 & = \sum_{\bar{\pi} \in \pstates^k} \expectedValueOf[\mcInducedMeasure{\pwapp{C_0}{\iverson{\phi} \cdot g}}]{\lambda~ \pi.~ {\iverson{ \pi_0 = \bar{\pi}_0 \land \dots \land \pi_{k-1} = \bar{\pi}_{k-1} \land \pi_k = \pstatea}}} \tag{$\expectedValueOf[]{\cdot}$ is linear}                                                                                                                \\
		 & = \sum_{\bar{\pi} \in \pstates^k} \expectedValueOf[\mcInducedMeasure{\pwapp{C_0}{\iverson{\phi} \cdot g}}]{\indicatorFct{\cylinderSet{\bar{\pi} \pstatea}}} \tag{def. cylinder set}                                                                                                                                                                                                         \\
		 & = \sum_{\bar{\pi} \in \pstates^k} \pwapp{C_0}{\iverson{\phi} \cdot g}(\bar{\pi}_0) \cdot \kernelEvalAt{\mcTransitionProbs}{\bar{\pi}_0}(\bar{\pi}_1) \cdot \ldots \cdot \kernelEvalAt{\mcTransitionProbs}{\bar{\pi}_{k-2}}(\bar{\pi}_{k-1}) \cdot \kernelEvalAt{\mcTransitionProbs}{\bar{\pi}_{k-1}}(\pstatea)   \tag{def. induced measure $\mcInducedMeasure[]{}$}                         \\
		 & = \sum_{\bar{\pi} \in \pstates^k} \sum_{\pstatea' \in \pstates}{g(\pstatea') \cdot \kernelEvalAt{\mcTransitionProbs}{\pstatea'}(\bar{\pi}_0)} \cdot \kernelEvalAt{\mcTransitionProbs}{\bar{\pi}_0}(\bar{\pi}_1) \cdot \ldots \cdot \kernelEvalAt{\mcTransitionProbs}{\bar{\pi}_{k-2}}(\bar{\pi}_{k-1}) \cdot \kernelEvalAt{\mcTransitionProbs}{\bar{\pi}_{k-1}}(\pstatea)   \tag{eq. above} \\
		 & = \sum_{\tilde{\pi} \in \pstates^{k+1}} g(\tilde{\pi}_0) \cdot \kernelEvalAt{\mcTransitionProbs}{\tilde{\pi}_0}(\tilde{\pi}_1) \cdot \ldots \cdot \kernelEvalAt{\mcTransitionProbs}{\tilde{\pi}_{k-1}}(\tilde{\pi}_{k}) \cdot \kernelEvalAt{\mcTransitionProbs}{\tilde{\pi}_{k}}(\pstatea)   \tag{combine sums}                                                                             \\
		 & = \sum_{\tilde{\pi} \in \pstates^{k+1}} \mcInducedMeasure{g}(\cylinderSet{\tilde{\pi}\pstatea})  \tag{def. $\mcInducedMeasure{g}$}                                                                                                                                                                                                                                                          \\
		 & = \sum_{\tilde{\pi} \in \pstates^{k+1}} \expectedValueOf[\mcInducedMeasure{g}]{\indicatorFct{\cylinderSet{\tilde{\pi}\pstatea}}}                                                                                                                                                                                                                                                            \\
		 & = \expectedValueOf[\mcInducedMeasure{g}]{\lambda~ \pi.~ \sum_{\tilde{\pi} \in \pstates^{k+1}}{\iverson{ \pi_0 = \tilde{\pi}_0 \land \dots \land \pi_{k} = \tilde{\pi}_{k} \land \pi_{k+1} = \pstatea}}}  \tag{$\expectedValueOf{\cdot}$ is linear, def. cylinder set}                                                                                                                       \\
		 & = \expectedValueOf[\mcInducedMeasure{g}]{\lambda~ \pi.~ \iverson{\pi_{k+1} = \pstatea}}.  \tag{merge based on prefix}                                                                                                                                                                                                                                                                       \\
	\end{align*}
	}
	This concludes the induction.

	Now, let $g \in \measuresOf{\pstates}$, $\pstatea \in \pstates$.
	As $\charfun$ is order-continuous by~\Cref{proposition:pw:scott_continuous}, we can apply Kleene's fixed point theorem~\cite{Davey2002} and get
	\[
		(\lfp \charfun)(\pstatea) = (\sup_{k \in \nats}{\charfun^k(0)})(\pstatea) = \sum_{k \in \nats}{(\pwapp{C_0}{\iverson{\phi} \cdot -})^k(g)(\pstatea)}.
	\]
	Similarly, by~\cite{DBLP:journals/siamcomp/SharirPH84,Pitman_1977},
	\[
		\mcOccupationMeasure[\loopMCsem{C}]{g}(\pstatea) = \sum_{k \in \nats}{\expectedValueOf[\mcInducedMeasure{g}]{\lambda~ \pi.~ \iverson{ \pi_k = \pstatea}}}.
	\]
	By~\cref{theorem:lfp_is_om:proof:inductive_equality}, both sums must be the same, and we conclude
	\[
		(\lfp \charfun)(\pstatea) = \mcOccupationMeasure[\loopMCsem{C}]{g}(\pstatea).
	\]
\end{proof}

To prove the relation of $\charfun$ to the expected runtime of loop, it is helpful to define both the weakest pre expectation $\wpre$ and the expected runtime $\ert$ inductively over the structure of the program, as given in~\Cref{tab:semantics:wp_ert:inductively}.
The definition of~$\wpre$ is due to~\cite{DBLP:series/mcs/McIverM05,DBLP:phd/dnb/Kaminski19}, and $\ert$ is adapted from~\cite{DBLP:conf/esop/KaminskiKMO16} to only increment the runtime on evaluation of the outermost loop guard.
\begin{table}[h]
	\centering
	\caption[Rules defining $\wpre$ and $\ert$ inductively for programs.]{Rules defining the weakest pre expectation $\wpreapp{C}{f}$ and expected runtime $\ertapp{C}{t}$ of program $C$ w.r.t.\ post expectation $f \colon \pstates \to \nnextreals$ resp.\ continuation time $t \colon \pstates \to \nnextreals$.}
	\label{tab:semantics:wp_ert:inductively}
	\smallskip
	\begin{adjustbox}{max width=\textwidth}%
		\renewcommand{\arraystretch}{1.25}%
		\setlength{\tabcolsep}{6pt}%
		\begin{tabular}{@{}llll@{}}
			\toprule
			program $C$             & $\wpreapp{C}{f}$                                                                          & $\ertapp{C}{t}$                                                                     \\
			\midrule
			$\pskip$                & $f$                                                                                       & $t$                                                                                 \\
			$\pdiverge$             & $0$                                                                                       & $\infty$                                                                            \\
			$\pchoice{C_1}{p}{C_2}$ & $p \cdot \wpreapp{C_1}{f} + (1-p) \cdot \wpreapp{C_2}{f}$                                 & $p \cdot \ertapp{C_1}{t} + (1-p) \cdot \ertapp{C_2}{t}$                             \\
			$\passign{x}{E}$        & $f[x/E]$                                                                                  & $t[x/E]$                                                                            \\
			$\psample{x}{\mu}$      & $\lambda \pstatea.\, \accumulate{v}{ f(\pstatea[x/v]) }{\kernelEvalAt{\mu}{\pstatea}(v)}$ & $\lambda \pstatea.\, \accumulate{v}{ t(\pstatea[x/v]) }{\mu(\pstatea)(v)}$          \\
			$\pcomp{C_1}{C_2}$      & $\wpreapp{C_1}{\wpreapp{C_2}{f}}$                                                         & $\ertapp{C_1}{\ertapp{C_2}{t}}$                                                     \\
			$\pite{\phi}{C_1}{C_2}$ & $\iverson{\phi} \cdot \wpreapp{C_1}{f} + \iverson{\neg \phi} \cdot \wpreapp{C_2}{f}$      & $\iverson{\phi} \cdot \ertapp{C_1}{t} + \iverson{\neg \phi} \cdot \ertapp{C_2}{t}$  \\
			$\ploop{\phi}{C_0}$     & $\lfp X.\, \iverson{\neg \phi} \cdot f + \iverson{\phi} \cdot \wpreapp{C_0}{X}$           & $\lfp X.\, 1 + \iverson{\neg \phi} \cdot t + \iverson{\phi} \cdot \wpreapp{C_0}{X}$ \\
			\bottomrule
		\end{tabular}
	\end{adjustbox}
\end{table}

\begin{proof}[Proof of~\Cref{theorem:pw:relation_to_ert}]%
	With the inductive definition of $\ert$ in~\Cref{tab:semantics:wp_ert:inductively}, the statement to show becomes:
	\begin{quote}
		For $C = \ploop{\phi}{C_0}$:
		\[
			\forall\, \text{initial measures}~g \in \measuresOf{\pstates}\colon \qquad \abs{\lfp \charfun} = \accumulate{\pstatea \in \pstates}{\ertapp{C}{0}(\pstatea)}{g(\pstatea)}~.
		\]
	\end{quote}

	Let $C = \ploop{\phi}{C_0}$ and $g \in \measuresOf{\pstates}$.

	Using the $\pw$-$\wpre$ Kozen duality, one shows by induction on $k \in \nats$:
	\begin{align}
		\expectedValueOf[(\pwapp{C_0}{\iverson{\phi} \cdot -})^k(g)]{1} = \expectedValueOf[g]{(\iverson{\phi} \cdot \wpreapp{C_0}{\cdot})^k(1)}. \label{theorem:lfp_mass_is_ert:proof:inductive_equality}
	\end{align}

	For $t \colon \pstates \to \nnextreals$, the loop characteristic function $\Psi_{t, C}(X) = 1 + \iverson{\neg \phi} \cdot t + \iverson{\phi} \cdot \wpreapp{C_0}{X}$ of $\ert$ is order-continuous~\cite{DBLP:conf/esop/KaminskiKMO16}.
	We apply Kleene's fixed point theorem~\cite{Davey2002} and have
	\begin{align}
		\ertapp{C}{t} = \sup_{k \in \nats}{\Psi_{t, C}^k(0)} = \sum_{k \in \nats}{(\iverson{\phi} \cdot \wpreapp{C_0}{\cdot})^k(1 + \iverson{\neg \phi} \cdot t)}\label{theorem:lfp_mass_is_ert:proof:ert_kleene}
	\end{align}

	Similarly, $\charfun$ is order-continuous by~\Cref{proposition:pw:scott_continuous}, and again with Kleene's fixed point theorem we get
		{\allowdisplaybreaks
			\begin{align*}
				\abs{(\lfp \charfun)} & = \expectedValueOf[\lfp \charfun]{1} = \expectedValueOf[\sup_{k \in \nats}{\charfun^k(0)}]{1}                                                           \\
				                      & = \expectedValueOf[\sum_{k \in \nats}{(\pwapp{C_0}{\iverson{\phi} \cdot -})^k(g)}]{1}                                                                   \\
				                      & = \sum_{k \in \nats}\expectedValueOf[{(\pwapp{C_0}{\iverson{\phi} \cdot -})^k(g)}]{1} \tag{$\expectedValueOf{\cdot}$ is linear}                         \\
				                      & = \sum_{k \in \nats}\expectedValueOf[g]{(\iverson{\phi} \cdot \wpreapp{C_0}{\cdot})^k(1)} \tag{\cref{theorem:lfp_mass_is_ert:proof:inductive_equality}} \\
				                      & = \expectedValueOf[g]{\sum_{k \in \nats}(\iverson{\phi} \cdot \wpreapp{C_0}{\cdot})^k(1)} \tag{$\expectedValueOf{\cdot}$ is linear}                     \\
				                      & = \expectedValueOf[g]{\ertapp{C}{0}} \tag{\cref{theorem:lfp_mass_is_ert:proof:ert_kleene} with $t = 0$}                                                 \\
				                      & = \accumulate{\pstatea \in \pstates}{\ertapp{C}{0}(\pstatea)}{g(\pstatea)}. \tag{def. $\expectedValueOf{\cdot}$}                                        \\
			\end{align*}
		}
\end{proof}

  \newpage
  \newcommand{\powersMod}{\mathfrak{P}_{2,N}}

\begin{Program}[H]
	{
		\small
		\begin{smashedalign}
			&\pcomment{$
					\sum_{\substack{1 \leq i < N \\ \exists k \colon i \equiv_N 2^k}}{
						a_{N,i} \cdot \gfVar{v}^i \frac{1 - \gfVar{c}^i}{1 - \gfVar{c}}
					}
				$}\\
			&\passign{v}{2v}\fatsemi\\
			&\pcomment{$
					\sum_{\substack{1 \leq i < N \\ \exists k \colon i \equiv_N 2^k}}{
						a_{N,i} \cdot \gfVar{v}^{2i} \frac{1 - \gfVar{c}^i}{1 - \gfVar{c}}
					}
				$}\\
			&\pchoice{\passign{c}{2c}}{\tfrac 1 2}{\passign{c}{2c + 1}}\fatsemi\\
			&\pcomment{$
					\frac{1}{2} \cdot \sum_{\substack{1 \leq i < N \\ \exists k \colon i \equiv_N 2^k}}{
						a_{N,i} \cdot \gfVar{v}^{2i} \frac{1 - \gfVar{c}^{2i}}{1 - \gfVar{c}^2}
					}
					+
					\frac{1}{2} \cdot \gfVar{C} \sum_{\substack{1 \leq i < N \\ \exists k \colon i \equiv_N 2^k}}{
						a_{N,i} \cdot \gfVar{v}^{2i} \frac{1 - \gfVar{c}^{2i}}{1 - \gfVar{c}^2}
					}
				$}\\
			&\pcomment{${} =
					\sum_{\substack{1 \leq i < N \\ \exists k \colon i \equiv_N 2^k}}{
						\frac{1}{2} a_{N,i} \cdot \gfVar{v}^{2i} \frac{1 - \gfVar{c}^{2i}}{1 - \gfVar{c}}
					}
				$}\\
			&\pif\,(N \leq v)\{\\
			&\quad \pcomment{$
					\sum_{\substack{\frac{N}{2} \leq i < N \\ \exists k \colon i \equiv_N 2^k}}{
						\frac{1}{2} a_{N,i} \cdot \gfVar{v}^{2i} \frac{1 - \gfVar{c}^{2i}}{1 - \gfVar{c}}
					}
				$}\\
			&\quad \pif\,(c < N)\{\\
			&\quad \quad \pcomment{$
					\sum_{\substack{\frac{N}{2} \leq i < N \\ \exists k \colon i \equiv_N 2^k}}{
						\frac{1}{2} a_{N,i} \cdot \gfVar{v}^{2i} \frac{1 - \gfVar{c}^{N}}{1 - \gfVar{c}}
					}
				$}\\
			&\quad \quad \passign{f}{f + 1}\\
			&\quad \quad \pcomment{$
					\sum_{\substack{\frac{N}{2} \leq i < N \\ \exists k \colon i \equiv_N 2^k}}{
						\frac{1}{2} a_{N,i} \cdot \gfVar{f}\, \gfVar{v}^{2i} \frac{1 - \gfVar{c}^{N}}{1 - \gfVar{c}}
					}
				$}\\
			&\quad \}\, \pelse\, \{\\
			&\quad \quad \pcomment{$
					\sum_{\substack{\frac{N}{2} \leq i < N \\ \exists k \colon i \equiv_N 2^k}}{
						\frac{1}{2} a_{N,i} \cdot \gfVar{v}^{2i} \frac{\gfVar{c}^N - \gfVar{c}^{2i}}{1 - \gfVar{c}}
					}
				$}\\
			&\quad \quad \passign{v}{v-N}\fatsemi\\
			&\quad \quad \pcomment{$
					\sum_{\substack{\frac{N}{2} \leq i < N \\ \exists k \colon i \equiv_N 2^k}}{
						\frac{1}{2} a_{N,i} \cdot \gfVar{v}^{2i-N} \frac{\gfVar{c}^N - \gfVar{c}^{2i}}{1 - \gfVar{c}}
					}
				$}\\
			&\quad \quad \passign{c}{c-N}\\
			&\quad \quad \pcomment{$
					\sum_{\substack{\frac{N}{2} \leq i < N \\ \exists k \colon i \equiv_N 2^k}}{
						\frac{1}{2} a_{N,i} \cdot \gfVar{v}^{2i-N} \frac{1 - \gfVar{c}^{2i-N}}{1 - \gfVar{c}}
					}
				$}\\
			&\quad \}\\
			&\}\\
			&\pcomment{$
					\sum_{\substack{\frac{N}{2} \leq i < N \\ \exists k \colon i \equiv_N 2^k}}{
						\frac{1}{2} a_{N,i} \cdot \gfVar{f}\, \gfVar{v}^{2i} \frac{1 - \gfVar{c}^{N}}{1 - \gfVar{c}}
					}
					+
					\sum_{\substack{\frac{N}{2} \leq i < N \\ \exists k \colon i \equiv_N 2^k}}{
						\frac{1}{2} a_{N,i} \cdot \gfVar{v}^{2i-N} \frac{1 - \gfVar{c}^{2i-N}}{1 - \gfVar{c}}
					}
					+
					\sum_{\substack{1 \leq i < \frac{N}{2} \\ \exists k \colon i \equiv_N 2^k}}{
						\frac{1}{2} a_{N,i} \cdot \gfVar{v}^{2i} \frac{1 - \gfVar{c}^{2i}}{1 - \gfVar{c}}
					}
				$}\\
		\end{smashedalign}
	}
	\caption{Calculation of $\pwapp{P_N}{\gfFilter{I_N}{\progvar{f} = 0}}$.}
	\label{prog:fdr:calculate_pw_body_of_inv}
\end{Program}

\section{Analyzing the Fast Dice Roller}
\label{apx:fdr}

Recall the template constructed in~\Cref{example:invariant_templates:fdr} for the Fast Dice Roller for $N \in \nats, N > 1$ with loop body $P_N$:
\[
	I_N = \sum_{\substack{1 \leq i < N \\ \exists k \colon i \equiv_N 2^k}}{
		a_{N,i} \cdot \gfVar{v}^i \frac{1 - \gfVar{c}^i}{1 - \gfVar{c}}
	}
	+
	\sum_{\substack{N \leq i < 2 N \\ \exists k \colon i \equiv_N 2^k}}{
		a_{N,i} \cdot \gfVar{f}\, \gfVar{v}^i \frac{1 - \gfVar{c}^N}{1 - \gfVar{c}}
	}.
\]

We want to show that this can be instantiated to yield an occupation invariant.
We begin by calculating the result of applying the loop body to
\[
	\gfFilter{I_N}{\progvar{f} = 0} =
	\sum_{\substack{1 \leq i < N \\ \exists k \colon i \equiv_N 2^k}}{
		a_{N,i} \cdot \gfVar{v}^i \frac{1 - \gfVar{c}^i}{1 - \gfVar{c}}
	}.
\]

Following the calculation in~\Cref{prog:fdr:calculate_pw_body_of_inv} we obtain
\begin{align}
	\measureGFa_0 + \pwapp{P_N}{\gfFilter{I_N}{\progvar{f} = 0}}
	 & \eeq 1 \cdot \gfVar{v}^1 \gfVar{c}^0 \gfVar{f}^0  \label{fdr:calculation:phi_i:sum:initial} \\
	 & \quad {}+
	\sum_{\substack{1 \leq i < \frac{N}{2}                                                         \\ \exists k \colon i \equiv_N 2^k}}{
		\frac{1}{2} a_{N,i} \cdot \gfVar{v}^{2i} \frac{1 - \gfVar{c}^{2i}}{1 - \gfVar{c}}
	}                                                   \label{fdr:calculation:phi_i:sum:double}   \\
	 & \quad {}+
	\sum_{\substack{\frac{N}{2} \leq i < N                                                         \\ \exists k \colon i \equiv_N 2^k}}{
		\frac{1}{2} a_{N,i} \cdot \gfVar{v}^{2i-N} \frac{1 - \gfVar{c}^{2i-N}}{1 - \gfVar{c}}
	}                                                   \label{fdr:calculation:phi_i:sum:overflow} \\
	 & \quad {}+ \sum_{\substack{\frac{N}{2} \leq i < N                                            \\ \exists k \colon i \equiv_N 2^k}}{
		\frac{1}{2} a_{N,i} \cdot \gfVar{f}\, \gfVar{v}^{2i} \frac{1 - \gfVar{c}^{N}}{1 - \gfVar{c}}
	}.                                                  \label{fdr:calculation:phi_i:sum:terminate}
\end{align}
For $I_N = \measureGFa_0 + \pwapp{P_N}{\gfFilter{I_N}{\progvar{f} = 0}}$ we compare coefficients of the summands which results in the following constraints, where we write $l \in \powersMod$ for $l \in \nats$ and $\exists k \colon l \equiv_N 2^k$:
\begin{align}
	a_{N,1} & \eeq \underbrace{1}_{(\ref{fdr:calculation:phi_i:sum:initial})}
	+ \underbrace{\iverson{\frac{1+N}{2} \in \powersMod} \cdot \frac{1}{2} a_{N, \frac{1+N}{2}}}_{(\ref{fdr:calculation:phi_i:sum:overflow})},
	\tag{A}\label{fdr:parameters:equation_system:initial}                                                                                               \\
	a_{N,i} & \eeq \underbrace{\iverson{\frac{i}{2} \in \powersMod} \cdot \frac{1}{2} a_{N, \frac{i}{2}}}_{(\ref{fdr:calculation:phi_i:sum:double})}
	+ \underbrace{\iverson{\frac{i+N}{2} \in \powersMod} \cdot \frac{1}{2} a_{N, \frac{i+N}{2}}}_{(\ref{fdr:calculation:phi_i:sum:overflow})}
	        & \text{if } 1 < i < N, i \in \powersMod,
	\tag{B}\label{fdr:parameters:equation_system:running}                                                                                               \\
	a_{N,i} & \eeq \underbrace{\iverson{\frac{i}{2} \in \powersMod} \cdot \frac{1}{2} a_{N, \frac{i}{2}}}_{(\ref{fdr:calculation:phi_i:sum:terminate})}
	        & \text{if } N \leq i < 2N, i \in \powersMod.
	\tag{C}\label{fdr:parameters:equation_system:terminating}                                                                                           \\
\end{align}
This is a linear system of equations.
An instantiation of the parameters that is a solution to these equations, and hence yields a valid invariant, can be motivated by studying the sequence of values for $v$ after each loop iteration.
Crucially, $v$ is always a power of two modulo $N$.
Starting from $v = 1$, it is doubled and either decremented by $N$ to stay in $0 < v < N$ or the loop terminates with $f=1$.

\begin{definition}
	Denote the sequence of powers of two modulo $N$ by $\pi_N \coloneqq (2^k \bmod N)_{k \in \nats}$.
\end{definition}
As there are only finitely many residues modulo $N$ this sequence must have repeating elements.
Indeed, as multiplication by two modulo $N$ is deterministic, each element in the sequence uniquely determines its successor:
\begin{proposition}\label{proposition:fdr:powersMod_sequence_deterministic}
	For $k \in \nats$: $\pi_{N,k+1} = (2 \cdot \pi_{N, k}) \bmod N$.
\end{proposition}
Consequently, $\pi_N$ must have \emph{the form of a lasso}, \ie, it has a repeating pattern after an initial prefix: $\pi_N = \pi_{N, \text{init}} (\pi_{N, \text{loop}})^\omega$ for some finite $\pi_{N, \text{init}}, \pi_{N, \text{loop}} \in [0, N)^*$.
Such a decomposition is not unique.

Intuitively, this sequence corresponds to the values of $v$, as long as $f = 0$.
In each loop iteration, the coin flip for the new value of $c$ results in the probability of the next value of $v$ being halved from its predecessor.
As the loop starts from the distribution $\measureGFa_0 = 1 \cdot \gfVar{v}^1 \gfVar{c}^0 \gfVar{f}^0$ with the single value $1$ for $v$, the probability for a value $i$ of $v$ depends on
\begin{enumerate*}
	\item the number of steps required to reach $i$ for the first time, and
	\item if $i$ occurs repeatedly (\ie, infinitely often), the number of steps between successive occurrences of $i$.
\end{enumerate*}

We introduce some notation to capture these concepts.
\newcommand{\firstOcc}[1]{\mathrm{F}_{N}(#1)}
\newcommand{\loopOcc}[1]{\mathrm{L}_{N}(#1)}
\begin{definition}
	For $i \in \nats$ let
	\begin{itemize}
		\item $\firstOcc{i} \coloneqq \inf \{ k \in \nats \mid i \equiv_N 2^k \} \in \nats \cup \{\infty\}$ denote the smallest exponent $k$ such that $i$ is congruent to $2^k$ modulo $N$ if such a $k$ exists, else $\infty$, and
		\item $\loopOcc{i} \coloneqq \inf \{ l \in \nats_+ \mid i \equiv_N i \cdot 2^l \} \in \nats_+ \cup \{\infty\}$ denote the smallest positive exponent $l$ such that $i \cdot 2^l$ is congruent to $i$ modulo $N$ if such an $l$ exists, else $\infty$.
	\end{itemize}
\end{definition}

Intuitively, for $i \in [0,N)$, $\loopOcc{i}$ is the period between occurrences of $i$ in $\pi_N$ after an initial offset of $\firstOcc{i}$.

This allows us to formally give a solution for the linear equation system.
\begin{definition}\label{definition:fdr:parameter_instantiation}
	For $N \in \nats, N > 1$, where $\invariantTemplateVars_N = \{ a_{N,i} \mid i \in \powersMod, 1 \leq i < 2N \}$, let $\tau_N \colon \invariantTemplateVars_N \to \reals$ with
	\begin{align*}
		\tau_N(a_{N,i}) & {}\coloneqq \left(\frac{1}{2}\right)^{\firstOcc{i}} \cdot \frac{1}{1 - \left(\frac{1}{2}\right)^{\loopOcc{i}}}
		                & \text{if } 1 \leq i < N, i \in \powersMod,                                                                     \\
		\tau_N(a_{N,i}) & {}\coloneqq \iverson{\frac{i}{2} \in \powersMod} \cdot \frac{1}{2} \tau_N(a_{N,\frac{i}{2}})
		                & \text{if } N \leq i < 2N, i \in \powersMod.
	\end{align*}
	where we use $\left(\frac{1}{2}\right)^\infty = 0$.
\end{definition}

To show that this instantiation is a solution, we establish some properties related to the sequence $\pi_N$ of powers of two modulo $N$.
First, if numbers occur repeatedly in $\pi_N$, they have the same period.
\begin{proposition}\label{proposition:fdr:powersMod_loop_length_is_unique}
	For $i,j \in [0, N)$: if $\firstOcc{i}, \firstOcc{j} \in \nats$ and $\loopOcc{i}, \loopOcc{j} \in \nats_+$, then $\loopOcc{i} = \loopOcc{j}$.
\end{proposition}
\begin{proof}
	If $\firstOcc{i} = \firstOcc{j}$, then $i \equiv_N 2^{\firstOcc{i}} \equiv_N j$ and we must have $i = j$, since $i,j \in [0, N)$.
	Thus, the statement is trivial.

	Else, w.l.o.g.\ assume $\firstOcc{i} < \firstOcc{j}$.
	As $\loopOcc{i} \in \nats_+$, using the deterministic construction of $\pi_N$~(\Cref{proposition:fdr:powersMod_sequence_deterministic}) we have
	\[
		\pi_N = (\pi_{N,0} \cdots \pi_{N, \firstOcc{i} - 1}) (\pi_{N, \firstOcc{i}} \cdots \pi_{N, \firstOcc{i} + \loopOcc{i} - 1})^\omega.
	\]
	Consequently, $\firstOcc{j} < \firstOcc{i} + \loopOcc{i}$, \ie, $j$ occurs in the repeating part of $\pi_N$.
	Write $k \coloneqq \firstOcc{j} - \firstOcc{i} > 0$ for the offset of $j$ into the repeating part, \ie, $\firstOcc{j} = \firstOcc{i} + k$.
	But then also $2^{\firstOcc{j} + \loopOcc{i}} \equiv_N \pi_{N, \firstOcc{j} + \loopOcc{i}} = \pi_{N, \firstOcc{i} + k + \loopOcc{i}} = \pi_{N, \firstOcc{i} + k} = \pi_{N, \firstOcc{j}} \equiv_N 2^{\firstOcc{j}} \equiv_N j$, giving us an upper bound on the period of $j$: $\loopOcc{j} \leq \loopOcc{i}$.

	Note that we can write $\loopOcc{i} = k + (\loopOcc{i} - k)$ with both $k$ and $\loopOcc{i} - k$ nonnegative.
	Then
	\begin{align*}
		i \cdot 2^{\loopOcc{j}} & {}\equiv_N 2^{\firstOcc{i} + \loopOcc{j}}                                                                                                                                \\
		                        & {}\equiv_N 2^{\firstOcc{i} + \loopOcc{i} + \loopOcc{j}}                                                                           & \text{(def.\ of $\loopOcc{i}$)}      \\
		                        & {}\equiv_N 2^{\firstOcc{i} + k + (\loopOcc{i} - k) + \loopOcc{j}} \equiv_N 2^{\firstOcc{i} + k + \loopOcc{j} + (\loopOcc{i} - k)} & \text{(rearrange)}                   \\
		                        & {}\equiv_N 2^{\firstOcc{j} + \loopOcc{j} + (\loopOcc{i} - k)} \equiv_N 2^{\firstOcc{j} + (\loopOcc{i} - k)}                       & \text{(def.\ of $k$, $\loopOcc{j}$)} \\
		                        & {}\equiv_N 2^{\firstOcc{j} - k + \loopOcc{i}} \equiv_N 2^{\firstOcc{i} + \loopOcc{i}} \equiv_N i.                                 & \text{(def.\ of $k$, $\loopOcc{i}$)} \\
	\end{align*}
	Hence, $\loopOcc{j} \geq \loopOcc{i}$ since $\loopOcc{i}$ is the minimal period of $i$ by definition.

	Combining the two inequalities, we conclude $\loopOcc{i} = \loopOcc{j}$.
\end{proof}

\newcommand{\sequenceIncept}{\bullet}
\newcommand{\sequencePre}[1]{\mathrm{pre}_N(#1)}
Just as the successors of elements in the sequence $\pi_N$ are uniquely defined, their predecessors are almost unique.
When executing the loop once, the measure on the new value of $v$ collects mass from these predecessor values.
We use an artificial element $\sequenceIncept$ to represent a predecessor of $1$.
For our purpose this will be used to represent probability mass entering the loop from the initial distribution $\measureGFa_0 = \gfVar{v}^1$ which is a point mass at $v = 1$.
\begin{definition}
	For $i \in [0,N)$ let $\sequencePre{i} \coloneqq \{ x \in [0,N) \mid \exists k, \pi_{N,k} = x \land \pi_{N,k+1} = i \} \cup \{ \sequenceIncept \mid i = 1 \}$.
\end{definition}

Characterisation of this predecessor set is the crucial step to establish $\tau_N$ as a solution of the linear equation system.
The following lemma captures the important properties which are easily observed from the minimal lasso decomposition of $\pi_N$.

\newcommand{\sequenceInitLength}{S_N}
\newcommand{\sequenceLoopLength}{L_N}
\begin{definition}
	The minimal lasso decomposition of $\pi_N$ is the decomposition
	\[
		\pi_N = (\pi_{N,0} \cdots \pi_{N, \sequenceInitLength - 1}) (\pi_{N, \sequenceInitLength} \cdots \pi_{N, \sequenceInitLength + \sequenceLoopLength - 1})^\omega
	\]
	such that $\sequenceInitLength \in \nats$ is minimal and for this $\sequenceInitLength$, $\sequenceLoopLength \in \nats_+$ is minimal.
\end{definition}

\begin{lemma}\label{lemma:fdr:powersSequence:properties}
	For $i \in [0,N)$:
	\begin{enumerate}
		\item\label{lemma:fdr:powersSequence:properties:firstOcc_bounded} $\firstOcc{i} < \sequenceInitLength + \sequenceLoopLength$ or $\firstOcc{i} = \infty$,
		\item\label{lemma:fdr:powersSequence:properties:initPart_is_not_looping} $\firstOcc{i} < \sequenceInitLength \implies \loopOcc{i} = \infty$,
		\item\label{lemma:fdr:powersSequence:properties:loopOcc_equal_loopLength} $\sequenceInitLength \leq \firstOcc{i} < \sequenceInitLength + \sequenceLoopLength \implies \loopOcc{i} = \sequenceLoopLength$,
		\item\label{lemma:fdr:powersSequence:properties:predecessors_possible_elements} $\sequencePre{1} \subseteq \{ \sequenceIncept, \frac{1 + N}{2} \}$ and $\sequencePre{i} \subseteq \{ \frac{i}{2}, \frac{i + N}{2}\}$ if $i \neq 1$,
		\item\label{lemma:fdr:powersSequence:properties:predecessors_count} $\abs{\sequencePre{i}} \eeq
			      \begin{cases}
				      1, & \text{if }\ \firstOcc{i} < \sequenceInitLength,                                             \\
				      2, & \text{if }\ \firstOcc{i} = \sequenceInitLength,                                             \\
				      1, & \text{if }\ \sequenceInitLength < \firstOcc{i} < \sequenceInitLength + \sequenceLoopLength, \\
				      0, & \text{if }\ \firstOcc{i} = \infty.                                                          \\
			      \end{cases}$
	\end{enumerate}
\end{lemma}
\begin{proof}
	\begin{enumerate}
		\item Assume $\firstOcc{i} \in \nats, \firstOcc{i} \geq \sequenceInitLength + \sequenceLoopLength$.
		      Then $i \equiv_N 2^{\firstOcc{i}} \equiv_N \pi_{N, \firstOcc{i}} = \pi_{N, \firstOcc{i} - \sequenceLoopLength} \equiv_N 2^{\firstOcc{i} - \sequenceLoopLength}$.
		      Thus, $\firstOcc{i} \leq \firstOcc{i} - \sequenceLoopLength$.
		      As $\sequenceLoopLength > 0$ this is a contradiction.
		\item Assume $\firstOcc{i} < \sequenceInitLength$ and $\loopOcc{i} \in \nats_+$.
		      Then the lasso decomposition
		      \[
			      \pi_N = (\pi_{N,0} \cdots \pi_{N, \firstOcc{i} - 1})  (\pi_{N, \firstOcc{i}} \cdots \pi_{N, \firstOcc{i} + \loopOcc{i} - 1})^\omega
		      \]
		      contradicts the minimality of $\sequenceInitLength$.
		\item We have $i \equiv_N 2^{\firstOcc{i}} \equiv_N \pi_{N, \firstOcc{i}} = \pi_{N, \firstOcc{i} + \sequenceLoopLength} \equiv_N 2^{\firstOcc{i} + \sequenceLoopLength}$.
		      Hence, $\loopOcc{i} \leq \sequenceLoopLength$.
		      Let $j \in [0,N)$ with $\firstOcc{j} = \sequenceInitLength$.
		      Then, by~\Cref{proposition:fdr:powersMod_loop_length_is_unique}, we have $\loopOcc{i} = \loopOcc{j}$.
		      Now, $\loopOcc{j} = \sequenceLoopLength$, because if $\loopOcc{j} < \sequenceLoopLength$, the lasso decomposition $\pi_N = (\pi_{N,0} \cdots \pi_{N, \sequenceInitLength - 1})  (\pi_{N, \sequenceInitLength} \cdots \pi_{N, \sequenceInitLength + \loopOcc{j} - 1})^\omega$ contradicts the minimality of $\sequenceLoopLength$.
		\item Let $x \in \sequencePre{i}$ with $x \in [0,N)$.
		      By~\Cref{proposition:fdr:powersMod_sequence_deterministic} $i = (2 \cdot x) \bmod N$.
		      As $x \in [0,N)$, $2x \in [0, 2N)$, thus either $i = 2 x$ or $i = 2 x - N$, \ie either $x = \frac{i}{2}$ or $x = \frac{i + N}{2}$.

		      Thus, $\sequencePre{i} \subseteq \{ \frac{i}{2}, \frac{i + N}{2}, \sequenceIncept \}$.
		      Finally, if $i = 1$, then $\frac{1}{2} \notin \nats$, thus $\frac{1}{2} \notin \sequencePre{1}$; and if $i \neq 1$, then $\sequenceIncept \notin \sequencePre{i}$.
		\item First, observe that the case distinction covers all possible cases by~(\ref{lemma:fdr:powersSequence:properties:firstOcc_bounded}).

		      If $\firstOcc{i} = \infty$, $i \neq 1$ as $\pi_{N,0} = 2^0 = 1$.
		      Thus, if $\sequencePre{i} \neq \emptyset$ there must exist $x \in [0,N), k \in \nats$ with $x \in \sequencePre{i}$, $\pi_{N,k} = x, \pi_{N,k+1} = i$.
		      But then $\firstOcc{i} \leq k + 1$, a contradiction.
		      Thus, indeed $\sequencePre{i} = \emptyset$.

		      Now, for the cases $\firstOcc{i} < \sequenceInitLength + \sequenceLoopLength$.
		      If $\firstOcc{i} = 0$, then $i = 1$ and $\sequenceIncept \in \sequencePre{i}$.
		      Else, $\firstOcc{i} > 0$, then
		      \(
		      (2 \cdot \pi_{N, \firstOcc{i} - 1}) \bmod N = (2 \cdot (2^{\firstOcc{i} - 1} \bmod N)) \bmod N = 2^{\firstOcc{i}} \bmod N = \pi_{N, \firstOcc{i}} = i,
		      \)
		      yields $\pi_{N, \firstOcc{i} - 1} \in \sequencePre{i}$.
		      In both cases, $\abs{\sequencePre{i}} \geq 1$.

		      \begin{description}
			      \item[Case $\firstOcc{i} < \sequenceInitLength$:]
			            For a contradiction, assume $\abs{\sequencePre{i}} > 1$.
			            Then there exists $k > \firstOcc{i}$ with $\pi_{N,k} = i$.
			            Thus, $\loopOcc{i} \leq k - \firstOcc{i}$, contradicting~(\ref{lemma:fdr:powersSequence:properties:initPart_is_not_looping}).
			      \item[Case $\firstOcc{i} = \sequenceInitLength$:]
			            For a contradiction, assume $\abs{\sequencePre{i}} = 1$, let $x$ denote the unique element of $\sequencePre{i}$.
			            If $i = 1$, then $\firstOcc{i} = 0$ and $x = \sequenceIncept$.
			            But $1 = \pi_{N, 0} = \pi_{N, 0 + \sequenceLoopLength} = (2 \cdot \pi_{N, \sequenceLoopLength - 1}) \bmod N$, \ie, $\pi_{N, \sequenceLoopLength - 1} \in \sequencePre{1}$.
			            But $\sequenceIncept \neq \pi_{N, \sequenceLoopLength - 1} \in [0,N)$, a contradiction.

			            If $i \neq 1$, then $\firstOcc{i} > 0$ and we have $x = \pi_{N, \firstOcc{i} - 1} = \pi_{N, \firstOcc{i} + \sequenceLoopLength - 1}$.
			            But then $\firstOcc{x} < \firstOcc{i} = \sequenceInitLength$ and $\loopOcc{x} \leq \sequenceLoopLength$ contradicts~(\ref{lemma:fdr:powersSequence:properties:initPart_is_not_looping}).
			      \item[Case $\sequenceInitLength < \firstOcc{i} < \sequenceInitLength + \sequenceLoopLength$:]
			            As $\firstOcc{i} > \sequenceInitLength \geq 0$ we must have $i \neq 1$.
			            By the lasso decomposition of $\pi_N$ we have
			            \[
				            \pi_N \eeq (\pi_{N,0} \cdots \pi_{N, \sequenceInitLength - 1}) (\pi_{N, \sequenceInitLength} \cdots \pi_{N, \firstOcc{i}} \cdots \pi_{N, \sequenceInitLength + \sequenceLoopLength - 1})^\omega
				            ~.
			            \]
			            Due to the minimality of this decomposition, for all $k \in \nats$ with $\pi_{N, k+1} = i$ we have $k = \firstOcc{i} + l_k \cdot \sequenceLoopLength - 1$ for some $l_k \in \nats$.
			            Since $\loopOcc{i} = \sequenceLoopLength$ by~(\ref{lemma:fdr:powersSequence:properties:loopOcc_equal_loopLength}), we have
			            \(
			            \pi_{N, k} \equiv_N 2^k \equiv_N 2^{\firstOcc{i} + l_k \cdot \sequenceLoopLength - 1} \equiv_N 2^{\firstOcc{i} - 1}
			            \), \ie, $\pi_{N, k} = (2^{\firstOcc{i} - 1}) \bmod N$.
			            As for all $x \in \sequencePre{i}$ there is a $k \in \nats$ with $\pi_{N, k} = x, \pi_{N, k+1} = i$, we have $x = (2^{\firstOcc{i} - 1}) \bmod N$, \ie, $\sequencePre{i} = \{ x \}$.
		      \end{description}
	\end{enumerate}
\end{proof}

\begin{corollary}\label{corollary:fdr:parameter_instantiation_solves_equations}
	For $N \in \nats, N > 1$, $\tau_N$ is a solution of the linear system of~\cref{fdr:parameters:equation_system:initial,fdr:parameters:equation_system:running,fdr:parameters:equation_system:terminating}.
\end{corollary}
\begin{proof}
	\begin{description}
		\item[\Cref{fdr:parameters:equation_system:initial}:]
		      We do a case distinction based on $\frac{1 + N}{2} \in \powersMod$.
		      If $\frac{1 + N}{2} \in \powersMod$, then $\frac{1 + N}{2} \in \sequencePre{1}$, since $\frac{1 + N}{2} < N$.
		      As also $\sequenceIncept \in \sequencePre{1}$, we have $\sequenceInitLength = \firstOcc{1} = 0$ by~\Cref{lemma:fdr:powersSequence:properties}~(\ref{lemma:fdr:powersSequence:properties:predecessors_count}).
		      Consequently, from the minimal lasso decomposition we have $\firstOcc{\frac{1 + N}{2}} = \sequenceLoopLength - 1$ and $\loopOcc{\frac{1 + N}{2}} = \loopOcc{1} = \sequenceLoopLength$ by~\Cref{lemma:fdr:powersSequence:properties}~(\ref{lemma:fdr:powersSequence:properties:loopOcc_equal_loopLength}).
		      We calculate
			      {\allowdisplaybreaks\begin{align*}
					      1 + \iverson{\frac{1+N}{2} \in \powersMod} \cdot \frac{1}{2} \tau_N(a_{N, \frac{1+N}{2}})
					       & \eeq 1 + \frac{1}{2} \left(\frac{1}{2}\right)^{\firstOcc{\frac{1 + N}{2}}} \cdot \frac{1}{1 - \left(\frac{1}{2}\right)^{\loopOcc{\frac{1 + N}{2}}}}                 \\
					       & \eeq 1 + \frac{1}{2} \left(\frac{1}{2}\right)^{\sequenceLoopLength - 1} \cdot \frac{1}{1 - \left(\frac{1}{2}\right)^{\sequenceLoopLength}}                          \\
					       & \eeq 1 + \left(\frac{1}{2}\right)^{\sequenceLoopLength} \cdot \frac{1}{1 - \left(\frac{1}{2}\right)^{\sequenceLoopLength}}                                          \\
					       & \eeq \frac{1 - \left(\frac{1}{2}\right)^{\sequenceLoopLength} + \left(\frac{1}{2}\right)^{\sequenceLoopLength}}{1 - \left(\frac{1}{2}\right)^{\sequenceLoopLength}} \\
					       & \eeq \frac{1}{1 - \left(\frac{1}{2}\right)^{\sequenceLoopLength}}                                                                                                   \\
					       & \eeq \left(\frac{1}{2}\right)^{\firstOcc{1}} \cdot \frac{1}{1 - \left(\frac{1}{2}\right)^{\loopOcc{1}}}                                                             \\
					       & \eeq \tau_N(a_{N, 1}).
				      \end{align*}}

		      If instead $\frac{1 + N}{2} \notin \powersMod$, then $\frac{1 + N}{2} \notin \sequencePre{1}$.
		      By~\Cref{lemma:fdr:powersSequence:properties}~(\ref{lemma:fdr:powersSequence:properties:predecessors_count}), thus $0 = \firstOcc{1} < \sequenceInitLength$.
		      From~\Cref{lemma:fdr:powersSequence:properties}~(\ref{lemma:fdr:powersSequence:properties:initPart_is_not_looping}) we thus get $\loopOcc{1} = \infty$, and calculate
		      \begin{align*}
			      1 + \iverson{\frac{1+N}{2} \in \powersMod} \cdot \frac{1}{2} \tau_N(a_{N, \frac{1+N}{2}})
			       & \eeq 1                                                                                                                                                \\
			       & \eeq \left(\frac{1}{2}\right)^{\firstOcc{1}} \cdot \frac{1}{1 - \left(\frac{1}{2}\right)^{\infty}}      & \text{(recall $(\sfrac{1}{2})^\infty = 0$)} \\
			       & \eeq \left(\frac{1}{2}\right)^{\firstOcc{1}} \cdot \frac{1}{1 - \left(\frac{1}{2}\right)^{\loopOcc{1}}}                                               \\
			       & \eeq \tau_N(a_{N, 1}).
		      \end{align*}
		\item[\Cref{fdr:parameters:equation_system:running}:]
		      Similar to the previous case, we make use of the fact that for $x \in \{\frac{i}{2}, \frac{i + N}{2} \}$ we have $x \in \powersMod \iff x \in \sequencePre{i}$, since $x < N$.
		      This allows to rewrite~\cref{fdr:parameters:equation_system:running}:
		      \[
			      a_{N,i}
			      = \iverson{\frac{i}{2} \in \powersMod} \cdot \frac{1}{2} a_{N, \frac{i}{2}} + \iverson{\frac{i+N}{2} \in \powersMod} \cdot \frac{1}{2} a_{N, \frac{i+N}{2}}
			      = \sum_{x \in \sequencePre{i}}{\frac{1}{2} a_{N, x}}.
		      \]

		      \begin{description}
			      \item[Case $\firstOcc{i} < \sequenceInitLength$ and Case $\sequenceInitLength < \firstOcc{i} < \sequenceInitLength + \sequenceLoopLength$:]
			            Let $x \in \sequencePre{i}$ denote the unique predecessor of $i$ by~\Cref{lemma:fdr:powersSequence:properties}~(\ref{lemma:fdr:powersSequence:properties:predecessors_count}).
			            From the minimal lasso decomposition we have $\firstOcc{x} = \firstOcc{i} - 1$ and
			            \begin{itemize}
				            \item $\loopOcc{x} = \loopOcc{i} = \infty$ by~\Cref{lemma:fdr:powersSequence:properties}~(\ref{lemma:fdr:powersSequence:properties:initPart_is_not_looping}) if $\firstOcc{i} < \sequenceInitLength$, or
				            \item $\loopOcc{x} = \loopOcc{i} = \sequenceLoopLength$ by~\Cref{lemma:fdr:powersSequence:properties}~(\ref{lemma:fdr:powersSequence:properties:loopOcc_equal_loopLength}) if $\sequenceInitLength < \firstOcc{i} < \sequenceInitLength + \sequenceLoopLength$.
			            \end{itemize}
			            In both cases we compute
			            \begin{align*}
				            \sum_{x' \in \sequencePre{i}}{\frac{1}{2} \tau_N(a_{N, x'})}
				             & \eeq \frac{1}{2} \left(\frac{1}{2}\right)^{\firstOcc{x}} \cdot \frac{1}{1 - \left(\frac{1}{2}\right)^{\loopOcc{x}}}     \\
				             & \eeq \frac{1}{2} \left(\frac{1}{2}\right)^{\firstOcc{i} - 1} \cdot \frac{1}{1 - \left(\frac{1}{2}\right)^{\loopOcc{i}}} \\
				             & \eeq  \left(\frac{1}{2}\right)^{\firstOcc{i}} \cdot \frac{1}{1 - \left(\frac{1}{2}\right)^{\loopOcc{i}}}                \\
				             & \eeq \tau_N(a_{N, i}).
			            \end{align*}
			      \item[Case $\firstOcc{i} = \sequenceInitLength$:]
			            Let $x_1, x_2 \in \sequencePre{i}, x_1 \neq x_2$ denote the two distinct predecessors of $i$ by~\Cref{lemma:fdr:powersSequence:properties}~(\ref{lemma:fdr:powersSequence:properties:predecessors_count}).
			            W.l.o.g.\ (by swapping $x_1$, $x_2$) we have $\firstOcc{x_1} = \firstOcc{i} - 1$ and $\firstOcc{x_2} = \firstOcc{i} + \sequenceLoopLength - 1$ from the minimal lasso decomposition.
			            In addition,
			            \begin{itemize}
				            \item $\loopOcc{x_1} = \infty$ by~\Cref{lemma:fdr:powersSequence:properties}~(\ref{lemma:fdr:powersSequence:properties:initPart_is_not_looping}), and
				            \item $\loopOcc{x_2} = \loopOcc{i} = \sequenceLoopLength$ by~\Cref{lemma:fdr:powersSequence:properties}~(\ref{lemma:fdr:powersSequence:properties:loopOcc_equal_loopLength}).
			            \end{itemize}
			            We compute
			            \begin{align*}
				             & \sum_{x' \in \sequencePre{i}}{\frac{1}{2} \tau_N(a_{N, x'})}
				            \\
				             & \eeq \frac{1}{2} \left(\frac{1}{2}\right)^{\firstOcc{x_1}} \cdot \frac{1}{1 - \left(\frac{1}{2}\right)^{\loopOcc{x_1}}} + \frac{1}{2} \left(\frac{1}{2}\right)^{\firstOcc{x_2}} \cdot \frac{1}{1 - \left(\frac{1}{2}\right)^{\loopOcc{x_2}}}                                                                 \\
				             & \eeq \frac{1}{2} \left(\frac{1}{2}\right)^{\firstOcc{i} - 1} \cdot \frac{1}{1 - \left(\frac{1}{2}\right)^{\infty}} + \frac{1}{2} \left(\frac{1}{2}\right)^{\firstOcc{i} + \sequenceLoopLength - 1} \cdot \frac{1}{1 - \left(\frac{1}{2}\right)^{\sequenceLoopLength}}                                        \\
				             & \eeq \left(\frac{1}{2}\right)^{\firstOcc{i}} + \left(\frac{1}{2}\right)^{\firstOcc{i} + \sequenceLoopLength} \cdot \frac{1}{1 - \left(\frac{1}{2}\right)^{\sequenceLoopLength}}                                                                                       & \text{($(\sfrac{1}{2})^\infty = 0$)} \\
				             & \eeq \left(\frac{1}{2}\right)^{\firstOcc{i}} \cdot \frac{1 - \left(\frac{1}{2}\right)^{\sequenceLoopLength} + \left(\frac{1}{2}\right)^{\sequenceLoopLength}}{1 - \left(\frac{1}{2}\right)^{\sequenceLoopLength}}                                                                                            \\
				             & \eeq \left(\frac{1}{2}\right)^{\firstOcc{i}} \cdot \frac{1}{1 - \left(\frac{1}{2}\right)^{\loopOcc{i}}}                                                                                                                                                                                                      \\                                         \\
				             & \eeq \tau_N(a_{N, i}).
			            \end{align*}
		      \end{description}
		\item[\Cref{fdr:parameters:equation_system:terminating}:]
		      This is immediate from the definition of $\tau_N$.
	\end{description}
\end{proof}

\begin{lemma}\label{lemma:fdr:past}
	For $N \in \nats$, $\mathrm{FDR}_N$ is PAST.
\end{lemma}
\begin{proof}
	For $N = 1$, it is easy to check that $I_1 \coloneqq 2 \cdot \measureGFa_0 + \gfVar{f}\, \gfVar{v}^2 = 2 \gfVar{v} + \gfVar{f}\, \gfVar{v}^2$ is an occupation invariant.
	Since $\abs{I_1} = 3 < \infty$, the loop and thus $\mathrm{FDR}_1$ is PAST by~\Cref{corollary:past_if_occupation_measure_finite}.

	For $N > 1$, instantiating the template $I_N$ with $\tau_N$ yields a valid invariant by~\Cref{corollary:fdr:parameter_instantiation_solves_equations}.
	Since
	\begin{align*}
		\abs{I_N[\tau_N]}
		 & \eeq \sum_{\substack{1 \leq i < N                \\ \exists k \colon i \equiv_N 2^k}}{
			\abs{\tau_N(a_{N,i})} \cdot \underbrace{\abs{\gfVar{v}^i \frac{1 - \gfVar{c}^i}{1 - \gfVar{c}}}}_{= i < N}
		}
		+
		\sum_{\substack{N \leq i < 2 N                      \\ \exists k \colon i \equiv_N 2^k}}{
			\abs{\tau_N(a_{N,i})} \cdot \underbrace{\abs{\gfVar{f}\, \gfVar{v}^i \frac{1 - \gfVar{c}^N}{1 - \gfVar{c}}}}_{= N}
		}                                                   \\
		 & {}\leq
		\sum_{\substack{1 \leq i < N                        \\ \exists k \colon i \equiv_N 2^k}}{
		\underbrace{\left(\frac{1}{2}\right)^{\firstOcc{i}}}_{\leq 1} \cdot \underbrace{\frac{1}{1 - \left(\frac{1}{2}\right)^{\loopOcc{i}}}}_{\leq 2} \cdot N
		}
		+
		\sum_{\substack{N \leq i < 2 N                      \\ \exists k \colon i \equiv_N 2^k}}{
			\underbrace{\iverson{\frac{i}{2} \in \powersMod} \cdot \frac{1}{2} \tau_N(a_{N,\frac{i}{2}})}_{\leq 1} \cdot N
		}                                                   \\
		 & {}\leq N \cdot 2N + N \cdot 1N = 3 N^2 < \infty,
	\end{align*}
	we can conclude that the loop and thus $\mathrm{FDR}_N$ is PAST.
\end{proof}

The occupation invariant $I_N$ provides an upper-bound
\[
	\pwapp{\mathrm{FDR}_N}{\measureGFa_0}[\gfVar{v}, \gfVar{f}/1]
	\leq (\gfFilter{I_N[\tau_N]}{f > 0})[\gfVar{v}, \gfVar{f}/1]
	= \sum_{\substack{N \leq i < 2 N \\ \exists k \colon i \equiv_N 2^k}}{
		\tau_N(a_{N,i}) \cdot \frac{1 - \gfVar{c}^N}{1 - \gfVar{c}}
	}.
\]
While this is a uniform measure of $c$ between $0$ and $N-1$, the unknown total mass leaves too much ambiguity for $\pwapp{\mathrm{FDR}_N}{\measureGFa_0}[\gfVar{v}, \gfVar{f}/1]$ to not be uniformly distributed.
To bridge this gap, we show $\sum_{\substack{N \leq i < 2 N \\ \exists k \colon i \equiv_N 2^k}}{ \tau_N(a_{N,i}) } = \frac{1}{N}$.
We achieve this by showing that the sum corresponds to a closed representation of the binary representation of $\frac{1}{N}$.

\newcommand{\fractionalPart}[1]{\left\{#1\right\}}
With $\fractionalPart{\cdot} \colon \nnreals \to [0,1)_\reals$ denoting the fractional part of a nonnegative real number, the binary representation of $\frac{1}{N}$ is given as
\[
	\frac{1}{N} = \sum_{k \in \nats}{ \iverson{\fractionalPart{\frac{1}{N} \cdot 2^k} \geq \frac{1}{2}} \cdot \frac{1}{2^{k+1}}}.
\]
We can relate the summands directly to the residue sequence $\pi_N$.
Whenever the sequence \emph{overflows} in a doubling step, \ie, $2 \pi_{N, k} \geq N$, the binary representation contains $\frac{1}{2^{k+1}}$.
\begin{proposition}\label{proposition:fdr:sequence_overflow_to_binary_representation}
	For $k \in \nats$, $\fractionalPart{\frac{1}{N} \cdot 2^k} = \frac{\pi_{N,k}}{N}$.
\end{proposition}
\begin{proof}
	By induction on $k \in \nats$:
	For $k = 0$, $\fractionalPart{\frac{1}{N} \cdot 2^0} = \frac{1}{N} = \frac{2^0 \bmod N}{N} = \frac{\pi_{N,0}}{N}$.
	Now, let $k \in \nats$, such that the statement holds.
	Then for $k + 1$:
	\begin{align*}
		 & \fractionalPart{\frac{1}{N} \cdot 2^{k+1}}                                                                                                                                                                                 \\
		 & \eeq 2 \cdot \fractionalPart{\frac{1}{N} \cdot 2^k} - \iverson{2 \cdot \fractionalPart{\frac{1}{N} \cdot 2^k} \geq 1} & \text{(inductive structure of $\fractionalPart{\cdot}$)}                                           \\
		 & \eeq 2 \cdot \frac{\pi_{N,k}}{N} - \iverson{2 \cdot \frac{\pi_{N,k}}{N} \geq 1}                                       & \text{(IH)}                                                                                        \\
		 & \eeq \frac{2 \cdot \pi_{N,k} - \iverson{2 \cdot \pi_{N,k} \geq N} \cdot N}{N}                                                                                                                                              \\
		 & \eeq \frac{\pi_{N,k+1}}{N}.                                                                                           & \text{(inductive structure of $\pi_{N}$,~\Cref{proposition:fdr:powersMod_sequence_deterministic})} \\
	\end{align*}
\end{proof}

\begin{lemma}\label{lemma:fdr:terminating_parameters_give_mass_one}
	$\sum_{\substack{N \leq i < 2 N \\ \exists k \colon i \equiv_N 2^k}}{ \tau_N(a_{N,i}) } = \frac{1}{N}$.
\end{lemma}
\begin{proof}
	Unfolding $\tau_N$ gives
	\[
		\sum_{\substack{N \leq i < 2 N \\ \exists k \colon i \equiv_N 2^k}}{ \tau_N(a_{N,i}) } = \sum_{\substack{N \leq i < 2 N \\ \exists k \colon i \equiv_N 2^k}}{ \iverson{\frac{i}{2} \in \powersMod} \cdot \frac{1}{2} \tau_N(a_{N,\frac{i}{2}}) }.
	\]
	We can change the sum's index set by noting that
	\begin{align*}
		                                                                      & \{ i \in \nats \mid N \leq i < 2N \land \exists k \colon i \equiv_N 2^k \land \frac{i}{2} \in \powersMod \} \\
		                                                                      & \eeq \{ i \in \nats \mid N \leq i < 2N \land i \in \powersMod \land \frac{i}{2} \in \powersMod \}           \\
		                                                                      & \eeq \{ i \in \nats \mid N \leq i < 2N \land \frac{i}{2} \in \powersMod \}                                  \\
		\intertext{bijectively corresponds to}
		\{ i \in \nats \mid N \leq i < 2N \land \frac{i}{2} \in \powersMod \} & {} \overset{\substack{i \mapsto \frac{i}{2}                                                                 \\ 2j \mapsfrom j}}{\cong} \{ j \in \nats \mid 0 \leq j < N \land j \in \powersMod \land 2j \geq N \}. \\
	\end{align*}
	This in turn is a subset of the entries in $\pi_N$, as $j < N$, hence can be identified by its first occurence, \ie, $\firstOcc{j}$.
	Using the minimal lasso decomposition of $\pi_N$, these indices lie between $0$ and $\sequenceInitLength + \sequenceLoopLength - 1$.
	Hence,
	{\allowdisplaybreaks
			\begin{align*}
				 & \sum_{\substack{N \leq i < 2 N                                                                                                                                                                                                                                                                                                                     \\ \exists k \colon i \equiv_N 2^k}}{ \iverson{\frac{i}{2} \in \powersMod} \cdot \frac{1}{2} \tau_N(a_{N,\frac{i}{2}}) }                                                                                                                                                                                            \\
				 & \eeq \sum_{\substack{0 \leq j < N                                                                                                                                                                                                                                                                                                                  \\ \exists k \colon j \equiv_N 2^k}}{ \iverson{2j \geq N} \cdot \frac{1}{2} \tau_N(a_{N,j}) }                                                                                                                                                                                                                     \\
				 & \eeq \sum_{\substack{0 \leq j < N                                                                                                                                                                                                                                                                                                                  \\ \firstOcc{j} \in \nats}}{ \iverson{2j \geq N} \cdot \frac{1}{2} \tau_N(a_{N,j}) }                                                                                                                                                                                                                              \\
				 & \eeq \sum_{0 \leq k < \sequenceInitLength + \sequenceLoopLength}{ \iverson{2 \pi_{N,k} \geq N} \cdot \frac{1}{2} \tau_N(a_{N,{\pi_{N,k}}}) }                                                                                                                                                                                                       \\
				 & \eeq \sum_{0 \leq k < \sequenceInitLength + \sequenceLoopLength}{ \iverson{2 \pi_{N,k} \geq N} \cdot \frac{1}{2} \left(\frac{1}{2}\right)^{\firstOcc{\pi_{N,k}}} \cdot \frac{1}{1 - \left(\frac{1}{2}\right)^{\loopOcc{\pi_{N,k}}}} }     & \text{(Def.\ of $\tau_N$ with $\pi_{N,k} < N$)}                                                        \\
				 & \eeq \sum_{0 \leq k < \sequenceInitLength + \sequenceLoopLength}{ \iverson{2 \pi_{N,k} \geq N} \cdot \frac{1}{2} \left(\frac{1}{2}\right)^{k} \cdot \frac{1}{1 - \left(\frac{1}{2}\right)^{\loopOcc{\pi_{N,k}}}} }                                                                                                                                 \\
				 & \eeq \sum_{0 \leq k < \sequenceInitLength}{ \iverson{2 \pi_{N,k} \geq N} \cdot \frac{1}{2} \left(\frac{1}{2}\right)^{k} \cdot \frac{1}{1 - \left(\frac{1}{2}\right)^{\infty}} }                                                                                                                                                                    \\
				 & \quad {}+ \sum_{\sequenceInitLength \leq k < \sequenceInitLength + \sequenceLoopLength}{ \iverson{2 \pi_{N,k} \geq N} \cdot \frac{1}{2} \left(\frac{1}{2}\right)^{k} \cdot \frac{1}{1 - \left(\frac{1}{2}\right)^{\sequenceLoopLength}} } & \text{(\Cref{lemma:fdr:powersSequence:properties})}                                                    \\
				 & \eeq \sum_{0 \leq k < \sequenceInitLength}{ \iverson{2 \pi_{N,k} \geq N} \cdot \frac{1}{2^{k+1}} }                                                                                                                                                                                                                                                 \\
				 & \quad {}+ \sum_{\sequenceInitLength \leq k < \sequenceInitLength + \sequenceLoopLength}{ \iverson{2 \pi_{N,k} \geq N} \cdot \frac{1}{2^{k+1}} \cdot \sum_{i \in \nats}{ \frac{1}{2^{\sequenceLoopLength \cdot i}} } }                     & \text{(Geometric series)}                                                                              \\
				 & \eeq \sum_{0 \leq k < \sequenceInitLength}{ \iverson{2 \pi_{N,k} \geq N} \cdot \frac{1}{2^{k+1}} }                                                                                                                                                                                                                                                 \\
				 & \quad {}+ \sum_{\sequenceInitLength \leq k < \sequenceInitLength + \sequenceLoopLength}{ \sum_{i \in \nats}{ \iverson{2 \pi_{N,k + \sequenceLoopLength \cdot i} \geq N} \frac{1}{2^{k + \sequenceLoopLength \cdot i + 1}} } }             & \text{(Lasso: $\pi_{N,k} = \pi_{N, k + \sequenceLoopLength \cdot i}$ if $k \geq \sequenceInitLength$)} \\
				 & \eeq \sum_{0 \leq k < \sequenceInitLength}{ \iverson{2 \pi_{N,k} \geq N} \cdot \frac{1}{2^{k+1}} } + \sum_{\sequenceInitLength \leq j}{ \iverson{2 \pi_{N,j} \geq N} \frac{1}{2^{j + 1}} }                                                & \text{(Simplify nested sums)}                                                                          \\
				 & \eeq \sum_{k \in \nats}{ \iverson{2 \pi_{N,k} \geq N} \cdot \frac{1}{2^{k+1}} }                                                                                                                                                                                                                                                                    \\
				 & \eeq \sum_{k \in \nats}{ \iverson{2 N \cdot \fractionalPart{\frac{1}{N} \cdot 2^k} \geq N} \cdot \frac{1}{2^{k+1}} }                                                                                                                      & \text{(\Cref{proposition:fdr:sequence_overflow_to_binary_representation})}                             \\
				 & \eeq \sum_{k \in \nats}{ \iverson{\fractionalPart{\frac{1}{N} \cdot 2^k} \geq \frac{1}{2}} \cdot \frac{1}{2^{k+1}} }                                                                                                                                                                                                                               \\
				 & \eeq \frac{1}{N}.                                                                                                                                                                                                                         & \text{(Binary representation)}                                                                         \\
			\end{align*}
		}
\end{proof}

\begin{theorem}
	For $N \in \nats$, $\pwapp{\mathrm{FDR}_N}{\measureGFa_0}[\gfVar{v}, \gfVar{f}/1] = \frac{1}{N} \cdot \frac{1 - \gfVar{c}^N}{1 - \gfVar{c}}$, where $\measureGFa_0 = \gfVar{V}$.
\end{theorem}
\begin{proof}
	For $N = 1$, we reuse the occupation invariant $I_1 = 2 \cdot \measureGFa_0 + \gfVar{f}\, \gfVar{v}^2 = 2 \gfVar{v} + \gfVar{f}\, \gfVar{v}^2$ from~\Cref{lemma:fdr:past}.
	Thus, $\pwapp{\mathrm{FDR}_1}{\measureGFa_0}[\gfVar{v}, \gfVar{f}/1] \leq \gfFilter{I_1}{f > 0}[\gfVar{v}, \gfVar{f}/1] = 1 \cdot \gfVar{c}^0$.
	Since $\mathrm{FDR}_1$ is PAST, $\abs{\pwapp{\mathrm{FDR}_1}{\measureGFa_0}[\gfVar{v}, \gfVar{f}/1]} = \abs{1} = 1$.
	But as $\abs{1 \cdot \gfVar{c}^0} = 1$ this implies $\pwapp{\mathrm{FDR}_1}{\measureGFa_0}[\gfVar{v}, \gfVar{f}/1] = \gfFilter{I_1}{f > 0}[\gfVar{v}, \gfVar{f}/1] = 1 \cdot \gfVar{c}^0$.

	For $N > 1$, the template $I_N$ instantiated with $\tau_N$ is an occupation invariant by~\Cref{corollary:fdr:parameter_instantiation_solves_equations}, and we get an over-approximation after marginalizing out $v$ and $f$ and using~\Cref{lemma:fdr:terminating_parameters_give_mass_one}:
	\begin{align*}
		\pwapp{\mathrm{FDR}_N}{\measureGFa_0}[\gfVar{v}, \gfVar{f}/1]
		 & {}\leq (\gfFilter{I_N[\tau_N]}{f > 0})[\gfVar{v}, \gfVar{f}/1] \\
		 & {}= \sum_{\substack{N \leq i < 2 N                             \\ \exists k \colon i \equiv_N 2^k}}{
			\tau_N(a_{N,i}) \cdot \frac{1 - \gfVar{c}^N}{1 - \gfVar{c}}
		}                                                                 \\
		 & {}= \frac{1}{N} \cdot \frac{1 - \gfVar{c}^N}{1 - \gfVar{c}},
	\end{align*}
	Now, we use that $\mathrm{FDR}_N$ is PAST by~\Cref{lemma:fdr:past} and have $\abs{\pwapp{\mathrm{FDR}_N}{\measureGFa_0}[\gfVar{v}, \gfVar{f}/1]} = \abs{1} = 1$.
	Since also $\abs{\frac{1}{N} \cdot \frac{1 - \gfVar{c}^N}{1 - \gfVar{c}}} = \frac{1}{N} \cdot N = 1$, the upper-bound is tight.
	Hence, we conclude $\pwapp{\mathrm{FDR}_N}{\measureGFa_0}[\gfVar{v}, \gfVar{f}/1] = \frac{1}{N} \cdot \frac{1 - \gfVar{c}^N}{1 - \gfVar{c}}$.
\end{proof}

  \newpage
  \section{Benchmark Programs}
\label{apx:benchmarks}

\begin{Program}[H]
    \begin{smashedalign}
        &\pcomment{Initial geometric distribution of $\progvar{x}$ with parameter $\sfrac{1}{2}$}\\
        &\annotate{\frac{1}{2 - \gfVar{x}}}\\
        &\pcomment{Probabilistically counting down}\\
        &\pwhile\,(x > 0)\{\\
        &\quad \pchoice{\pskip}{\tfrac 1 2}{\passign{x}{x-1}}\\
        &\}
    \end{smashedalign}
    \caption{Counting down until seeing $x$ successes (Negative Binomial). (\texttt{faulty\_decrement})}
    \label{prog:faulty_decrement}
\end{Program}

\begin{Program}[H]
    \begin{smashedalign}
        &\pcomment{Initial uniform distribution with $x \in \{1,2\}$}\\
        &\annotate{\frac{1}{2} \gfVar{x} + \frac{1}{2} \gfVar{x}^2}\\
        &\pcomment{A never terminating loop for x=1}\\
        &\pwhile\,(x = 1)\{\\
        &\quad \pskip\\
        &\}
    \end{smashedalign}
    \caption{A partially terminating program. (\texttt{nontermination})}
    \label{prog:nontermination_loop}
\end{Program}

\begin{Program}[H]
    \begin{smashedalign}
        &\pcomment{Initial Dirac distribution on (x=1,c=0)}\\
        &\annotate{\gfVar{x}}\\
        &\pcomment{The random walk}\\
        &\pwhile\,(x > 1)\{\\
        &\quad \pchoice{\passign{x}{x-1}}{\tfrac 1 2}{\passign{x}{x+1}}\fatsemi\\
        &\quad \passign{c}{c+1}\\
        &\}
    \end{smashedalign}
    \caption{The symmetric random walk with counter. (\texttt{random\_walk\_counter})}
    \label{prog:random_walk_counter}
\end{Program}

\begin{Program}[H]
    \begin{smashedalign}
        &\pcomment{Initial Dirac distribution on (v=1,c=0,f=0)}\\
        &\annotate{\gfVar{v}}\\
        &\pcomment{The fast dice roller algorithm}\\
        &\pwhile\,(f = 0)\{\\
        &\quad \passign{v}{2v}\fatsemi\\
        &\quad \pchoice{\passign{c}{2c}}{\tfrac 1 2}{\passign{c}{2c + 1}}\fatsemi\\
        &\quad \pif\,(6 \leq v)\{\\
        &\quad \quad \pif\,(c < 6)\{\\
        &\quad \quad \quad \passign{f}{f + 1}\\
        &\quad \quad \}\, \pelse\, \{\\
        &\quad \quad \quad \passign{v}{v-6}\fatsemi\\
        &\quad \quad \quad \passign{c}{c-6}\\
        &\quad \quad \}\\
        &\quad \}\\
        &\}
    \end{smashedalign}
    \caption{The fast dice roller algorithm for N=6. (\texttt{fast\_dice\_roller})}
    \label{prog:fast_dice_roller_6}
\end{Program}

\begin{Program}[H]
    \begin{smashedalign}
        &\pcomment{Initial Dirac distribution on (x=0,y=0)}\\
        &\annotate{1}\\
        &\pcomment{Randomly walk to one of the axis}\\
        &\pwhile\,(x > 0 ~\wedge~ y > 0)\{\\
        &\quad \pchoice{\passign{x}{x-1}}{\tfrac 1 2}{\passign{y}{y-1}}\\
        &\}
    \end{smashedalign}
    \caption{Randomly walk to one of the axis. (\texttt{cond\_and})}
    \label{prog:cond_and}
\end{Program}

\begin{Program}[H]
    \begin{smashedalign}
        &\pcomment{Initial Dirac distribution on (x=1)}\\
        &\annotate{\gfVar{x}}\\
        &\pcomment{Filtered geometric loop where every 4th term is missing}\\
        &\pwhile\,(x \equiv 1 \mod 3)\{\\
        &\quad \pchoice{\passign{x}{x+1}}{\tfrac 1 3}{\pchoice{\passign{x}{x+2}}{\tfrac 1 2 }{\passign{x}{x+3}}}\\
        &\}
    \end{smashedalign}
    \caption{The filtered geometric loop. (\texttt{thirds\_geometric})}
    \label{prog:thirds_geometric}
\end{Program}

  \newpage
  \section{Details on Contraction Invariants~\cite{DBLP:journals/pacmpl/ZaiserMO25}}
\label{apx:zaiser}

\begin{figure}[t]
    \centering
    \begin{tikzpicture}
        \node[state,initial where=above,initial text=1,initial] (s1) {$s_1$};
        \node[state,left = of s1] (s2) {$s_2$};
        \node[state,right = of s1] (s3) {$s_3$};
        \draw[->] (s1) edge node[auto] {$\sfrac 1 3$} (s2);
        \draw[->] (s1) edge node[auto] {$\sfrac 1 3$} (s3);
        \draw[->] (s1) edge[loop below] node[auto] {$\sfrac 1 3$} (s1);
    \end{tikzpicture}
    \caption{A Markov chain where contraction invariants~\cite{DBLP:journals/pacmpl/ZaiserMO25} cannot prove any non-trivial bounds on the probability of termination in $s_2$ and $s_3$ (see \Cref{apx:zaiser}).}
    \label{fig:zaiser}
\end{figure}
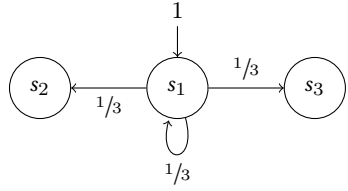

We recall the following notion from \cite{DBLP:journals/pacmpl/ZaiserMO25}:
For a loop $\ploop{\phi}{C}$ with initial probability distribution $\mu$ and a constant $c \in [0,1)$, a \emph{$c$-contraction invariant} is a measure $\nu$ such that $\mu \leq \nu$ and $\pwapp{C}{\iverson{\phi} \cdot \nu} \leq c \cdot \nu$.
Such a $\nu$ is a witness that $\pwapp{C}{\mu} \leq \frac{\iverson{\neg\phi} \cdot \nu}{1-c}$~\cite{DBLP:journals/pacmpl/ZaiserMO25}.
In contrast, recall from \Cref{def:superinvariant} that an occupation (super)invariant is a measure $\nu'$ satisfying $\mu + \pwapp{C}{\iverson{\phi} \cdot \nu'} \leq \nu'$.
It certifies the upper bound $\pwapp{C}{\mu} \leq \iverson{\neg\phi} \cdot \nu'$.
These definitions imply that a $c$-contraction invariant $\nu$ is an occupation invariant for the initial measure $(1-c)\cdot \nu$.
In other words, contraction invariants form a special subclass of occupation invariants.
We now show that the latter are strictly more powerful, \ie, that \emph{occupation invariants can certify bounds not attainable by contraction invariants}.
We shall illustrate this by means of a concrete example.

Consider the Markov chain in \Cref{fig:zaiser}, corresponding to a loop with guard $s=1$ that, in each iteration, either executes $\assign{s}{2}$, $\assign{s}{3}$, or does nothing (each with equal probability).
For readability, we write measures $\mu$ on the three relevant states as row vectors $(\mu(s_1), \mu(s_2), \mu(s_3))$.
Assume the initial distribution $\mu = (1, 0, 0)$.
The resulting final distribution is $(0, \tfrac{1}{2}, \tfrac{1}{2})$, as witnessed by the (exact) occupation measure $(\tfrac{3}{2}, \tfrac{1}{2}, \tfrac{1}{2})$.
For every $c \in (0,1)$, the best (pointwise smallest) $c$-contraction invariant is $(1, \tfrac{1}{3c}, \tfrac{1}{3c})$.
As described above, this invariant yields the upper bound
\[
\left(0,\, \frac{1}{3c(1-c)},\, \frac{1}{3c(1-c)}\right)
\]
on the final distribution.
The expressions are minimized for $c = \tfrac{1}{2}$, which, however, implies only the trivial upper bound $(0, \tfrac{4}{3}, \tfrac{4}{3})$.

}{}

\end{document}